\documentclass[11pt,letterpaper]{article}
\usepackage[T1]{fontenc}
\usepackage{amsmath,amssymb,amsthm,mathtools}
\usepackage{newtxtext,newtxmath}
\usepackage[margin=1in,headheight=14pt]{geometry}
\usepackage{microtype}
\usepackage{booktabs,array,enumitem}
\usepackage{needspace}
\usepackage{tikz}
\usetikzlibrary{arrows.meta,positioning}
\usepackage{placeins}
\usepackage[
  backend=biber,
  style=alphabetic,
  sorting=nyt,
  giveninits=true,
  maxbibnames=99,
  maxalphanames=4,
  minalphanames=3,
  doi=true,
  isbn=false,
  url=false,
  eprint=true
]{biblatex}
\usepackage{xcolor}
\definecolor{linkblue}{RGB}{0,70,140}
\definecolor{citegreen}{RGB}{0,110,70}
\definecolor{urlpurple}{RGB}{100,40,130}

\usepackage[
  colorlinks=true,
  linkcolor=linkblue,
  citecolor=citegreen,
  urlcolor=urlpurple
]{hyperref}

\IfFormatAtLeastTF{2026-06-01}{}{\usepackage{aliascnt}}
\usepackage[nameinlink,noabbrev,capitalise]{cleveref}

\numberwithin{equation}{section}

\newtheorem{theorem}{Theorem}[section]
\newcommand{\newsharedtheorem}[2]{%
  \IfFormatAtLeastTF{2026-06-01}
    {\newtheorem{#1}[theorem]{#2}}
    {\newaliascnt{#1}{theorem}%
     \newtheorem{#1}[#1]{#2}%
     \aliascntresetthe{#1}}%
}
\newsharedtheorem{lemma}{Lemma}
\newsharedtheorem{proposition}{Proposition}
\newsharedtheorem{corollary}{Corollary}
\newtheorem*{informaltheorem}{Theorem}

\theoremstyle{definition}
\newsharedtheorem{definition}{Definition}
\newsharedtheorem{example}{Example}

\theoremstyle{remark}
\newsharedtheorem{remark}{Remark}

\crefname{theorem}{Theorem}{Theorems}
\Crefname{theorem}{Theorem}{Theorems}

\crefname{lemma}{Lemma}{Lemmas}
\Crefname{lemma}{Lemma}{Lemmas}

\crefname{proposition}{Proposition}{Propositions}
\Crefname{proposition}{Proposition}{Propositions}

\crefname{corollary}{Corollary}{Corollaries}
\Crefname{corollary}{Corollary}{Corollaries}

\crefname{definition}{Definition}{Definitions}
\Crefname{definition}{Definition}{Definitions}

\crefname{example}{Example}{Examples}
\Crefname{example}{Example}{Examples}

\crefname{remark}{Remark}{Remarks}
\Crefname{remark}{Remark}{Remarks}

\newcommand{\R}{\mathbb R}

\newcommand{\F}{\mathcal F}
\newcommand{\B}{\mathcal B}
\newcommand{\CM}{\operatorname{CM}}
\newcommand{\SC}{\operatorname{SC}}
\newcommand{\OPT}{\operatorname{OPT}}
\newcommand{\AR}{\operatorname{AR}}
\newcommand{\med}{\operatorname{med}}

\newcommand{\dist}{\operatorname{dist}}
\newcommand{\diam}{\operatorname{diam}}
\newcommand{\aff}{\operatorname{aff}}
\newcommand{\conv}{\operatorname{conv}}

\newcommand{\inter}{\operatorname{int}}
\newcommand{\boxB}{\operatorname{box}_{\B}}
\newcommand{\ip}[2]{\langle #1,#2\rangle}
\newcommand{\norm}[1]{\lVert #1\rVert}
\newcommand{\eps}{\varepsilon}

\newcommand{\prfref}[1]{\hyperref[#1]{Section~\ref*{#1}}}

\newcounter{proofstage}[section]
\renewcommand{\theproofstage}{\arabic{proofstage}}

\crefname{proofstage}{stage}{stages}
\Crefname{proofstage}{Stage}{Stages}
\newcommand{\proofstage}[2]{%
  \par\medskip
  \refstepcounter{proofstage}\label{#1}%
  \pdfbookmark[3]{Stage \theproofstage: #2}{navigation.stage.\thesection.\theproofstage}%
  \noindent\textbf{Stage \theproofstage. #2.}\par\nobreak\smallskip\noindent
}

\newcommand{\ExtS}{\mathcal S}
\newcommand{\ExtA}{\mathcal A}

\newcommand{\ExtLeb}{\mathcal L}
\newcommand{\Exttr}{\operatorname{tr}}

\newcommand{\ExtId}{I_d}

\DeclareFieldFormat[article,inproceedings]{title}{#1}
\AtEveryBibitem{\clearfield{issn}\clearfield{month}\clearfield{day}\clearlist{location}}

\title{Strategyproof Aggregation in Euclidean Spaces:\\
Rigidity and Median Optimality}
\author{
  Jianhao Jia\\
  The University of Hong Kong\\
  \texttt{jianhao@hku.hk}
}
\date{}
\renewenvironment{abstract}{%
  \begin{center}\normalfont\normalsize\bfseries\abstractname\end{center}%
  \begin{quote}\normalfont\normalsize}{\end{quote}}
\hypersetup{pdftitle={Strategyproof Aggregation in Euclidean Spaces: Rigidity and Median Optimality},pdfauthor={Jianhao Jia},bookmarksdepth=3}
\begin{document}
\hypersetup{pageanchor=false}
\begin{titlepage}
\maketitle
\thispagestyle{empty}
\vspace{-1.5em}

\begin{abstract}
We study deterministic strategyproof mechanisms for facility location in
finite-dimensional Euclidean spaces. For every population of at least
three agents and every finite dimension, the coordinate-wise median
minimizes the worst-case approximation ratio among anonymous mechanisms.
Agent identities can improve the ratio at a fixed population size, and we
give an explicit four-agent separation. Nevertheless, for every fixed
dimension, the optimal ratio over all deterministic strategyproof
mechanisms converges, as the population grows, to the supremum of the
median's ratios over finite populations. Combined with the sharp bound
for the median, this gives the optimal guarantee
\(\sqrt{6\sqrt3-8}\) uniform over populations and dimensions.

The main structural ingredient is a rigidity theorem for strategyproof
aggregation. We characterize unanimous, anonymous,
translation-equivariant strategyproof mechanisms with bounded individual
influence as fixed coordinate-wise order-statistic rules in a single
orthonormal frame. We also prove that, for strategyproof
mechanisms, continuity is equivalent to convexity of the full range in
every finite dimension. An approximation-preserving normalization and a
reflection argument yield the fixed-population anonymous comparison. To
remove anonymity in the large-population limit, we develop an extraction
theorem that converts a common guarantee on arbitrarily large populations
into an anonymous deterministic mechanism on any prescribed finite
population.
\end{abstract}

\end{titlepage}
\hypersetup{pageanchor=true}
\pagenumbering{roman}
\setcounter{tocdepth}{1}
\tableofcontents
\clearpage
\pagenumbering{arabic}

\section{Introduction}\label{sec:intro}
Strategic aggregation is a fundamental problem in social choice and
mechanism design. A group of agents must select a common outcome, such
as a facility location, a compromise policy, or a vector of decisions
on several issues. An aggregation rule chooses the outcome from the
agents' reports, and an agent may be able to obtain a more preferred
outcome by misreporting. The central incentive requirement is therefore
strategyproofness: truthful reporting must be optimal for every agent,
regardless of the reports submitted by the others. However, classical
impossibility results severely restrict non-dictatorial strategyproof
rules on unrestricted preference domains
~\cite{Satterthwaite75,barbera2011strategyproof}. Consequently, positive results have
focused on structured domains, most notably spatial and
single-peaked preferences. In a spatial model, each agent has an ideal
point in the outcome space and ranks alternatives according to their
distance from that point.

In algorithmic mechanism design without money, this spatial setting is
studied under the name of strategic facility
location~\cite{procaccia2013approximate,FotakisT14,FotakisT16,
SerafinoV16,walsh2020strategy,chan2021mechanism}. Agents report
locations, a mechanism places a facility, and its approximation ratio
compares the resulting total distance with the minimum possible total
distance. On a line, the median is strategyproof and achieves the
optimal social cost. Its coordinate-wise extension selects a median
separately in each coordinate of a fixed orthonormal frame and remains
strategyproof. In the Euclidean
plane, Goel and Hann-Caruthers~\cite{GoelHC23} proved that, for every
odd population, the coordinate-wise median minimizes the worst-case
approximation ratio among deterministic, anonymous, strategyproof
mechanisms. They also determined its exact ratio for the sum-of-distances objective.

Despite the median's long history, its approximation guarantee in
arbitrary Euclidean dimension was understood only recently. The standard
upper bound in $\mathbb{R}^d$ was $\sqrt d$%
~\cite{meir2019strategyproof}. Gravin and Jia~\cite{jia1} proved
the dimension-independent constant bound $\sqrt{6\sqrt3-8}$,
and constructed examples approaching this value as the dimension grows. The optimality of the ratio remains unknown.
Thus, an interesting question to ask regarding approximation guarantee is:
\begin{quote}
\emph{Does the coordinate-wise median always achieve the best approximation
ratio permitted by strategyproofness in Euclidean space?}
\end{quote}

Answering this question requires understanding what competing
strategyproof mechanisms can look like. This characterization problem
predates the approximation literature and has long been central to social
choice%
~\cite{moulin1980strategy,border1983straightforward,
kim1984nonmanipulability,peters1993range,
barbera1993generalized,ching1997strategy,
barbera1998strategy,schummer2002strategy,
peremans1997strategy}. Median rules and their generalizations have played a
central role in this literature since Black's work on single-peaked
preferences~\cite{Black48,moulin1980strategy,border1983straightforward,
barbera1993generalized,ching1997strategy,barbera1998strategy}.

On a line, the median of the reported ideal points is strategyproof.
Moulin~\cite{moulin1980strategy} showed that anonymity, Pareto
efficiency, and strategyproofness characterize generalized median rules
among peak-reporting rules on the single-peaked domain. These rules add
fixed phantom reports and select the median of the enlarged collection.
In the Euclidean plane, Kim and
Roush~\cite{kim1984nonmanipulability} characterized continuous,
anonymous, strategyproof rules as generalized coordinate-wise medians.
Peters, van der Stel, and Storcken~\cite{peters1993range} proved that
continuity is equivalent to convexity of the full range for planar
strategyproof rules, yielding the corresponding characterization under
unanimity. For odd populations, Pareto optimality singles out the
ordinary coordinate-wise median within the anonymous strategyproof
class~\cite{peters1992pareto}.

Higher-dimensional characterizations have also been obtained, mostly on
broader domains of separable preferences%
~\cite{border1983straightforward,barbera1993generalized,
peremans1997strategy,lebreton1999separable}. Border and
Jordan~\cite{border1983straightforward} study separable star-shaped
preferences and show that, under unanimity, strategyproof rules break
into one-dimensional rules with phantom voters. Barber\`a, Gul, and
Stacchetti~\cite{barbera1993generalized} obtain related
generalized-median descriptions on multidimensional single-peaked
domains. Peremans, Peters, van der Stel, and
Storcken~\cite{peremans1997strategy} study separable quadratic
preferences and show that, under unanimity, strategyproof rules are
continuous and decompose coordinatewise. Le Breton and
Sen~\cite{lebreton1999separable} explain more generally when a separable
preference domain contains enough variation for strategyproofness to
force such a decomposition.

The important difference is that these domains allow agents with the
same ideal point to care differently about different coordinates. For
example, one agent may care much more about horizontal distance, while
another agent with the same ideal point may care more about vertical
distance. These additional preference variations allow the coordinates
to be studied one at a time. In the standard Euclidean model, once an
agent's ideal point is fixed, its entire preference ordering is determined
by ordinary Euclidean distance. The preference domain does not distinguish
any coordinate axes, and a strategyproof mechanism need not treat the
coordinates independently. Any coordinatewise decomposition, together
with the orthonormal frame in which it holds, must therefore follow from
strategyproofness and the other properties of the mechanism rather than
from the preference domain itself. This is also the structure needed for
the approximation problem, since each agent's Euclidean cost combines all
coordinates.

We address the structural and approximation questions together. We
characterize anonymous, unanimous, translation-equivariant strategyproof
mechanisms with bounded individual influence as fixed coordinate-wise
order-statistic rules in a single orthonormal frame. An all-dimensional
regularity theorem derives continuity from unanimity and
strategyproofness, while an approximation-preserving normalization brings
every anonymous competitor relevant to the optimality comparison into
the class covered by the characterization.

Using this structure, we prove that the coordinate-wise median is optimal
among anonymous mechanisms for every population of at least three agents
and every finite dimension. Anonymity creates a genuine boundary at a
fixed population size: for four agents in dimensions at least three, we
construct a nonanonymous strategyproof mechanism whose approximation
ratio is strictly smaller than that of every fixed-middle coordinate-wise
median. Nevertheless, this advantage disappears as the population grows.
For every fixed dimension, the optimal ratio over all deterministic
strategyproof mechanisms converges to the supremum of the median's ratios
over population sizes. The next subsection states these results and
explains the assumptions behind them.
\subsection{Main results and assumptions}
\label{sec:results}
We state the main results informally and explain the assumptions that enter
our characterization. Section~\ref{sec:prelim} gives the formal model and
notation, and the exact theorem statements appear where their proofs begin.
Throughout this subsection, all mechanisms are deterministic.

\begin{definition}[Anonymity]
A mechanism $f$ is \emph{anonymous} if
$f(x_{\pi(1)},\ldots,x_{\pi(n)})=f(x_1,\ldots,x_n)$ for every permutation
$\pi$ of the agents.
\end{definition}

Anonymity means that the outcome depends on the reported locations, but not
on the identities attached to those reports. It is a standard symmetry and
fairness requirement in strategyproof social choice and facility
location~\cite{moulin1980strategy,kim1984nonmanipulability,
peters1992pareto,peters1993range,GoelHC23}. Our first result determines the
best approximation ratio under this symmetry.

\Needspace{5\baselineskip}
\begin{informaltheorem}[Anonymous median optimality]
For every population of at least three agents and every finite dimension,
the coordinate-wise median has the smallest worst-case approximation ratio
among anonymous strategyproof mechanisms.
\end{informaltheorem}

The formal statement is \cref{ext:thm:fixed-optimality}. Anonymity also marks
a genuine boundary at a fixed population size. With four agents, designate
one agent and count its report twice when taking the coordinate-wise median,
while counting every actual agent once in the social cost. The resulting
strategyproof mechanism has ratio $\sqrt2$, whereas every coordinate-wise
median has ratio at least $(1+2\sqrt3)/3$ in dimension at least three.
Thus agent identities can improve the ratio at a particular population
size; \cref{ext:prop:four-agents} gives the complete example.

The advantage from identities disappears as the population grows.

\Needspace{5\baselineskip}
\begin{informaltheorem}[Large-population optimality]
Fix a finite dimension. As the population grows, the best approximation
ratio among all strategyproof mechanisms converges to the supremum of the
coordinate-wise median's ratios over finite population sizes.
\end{informaltheorem}

The formal statement is \cref{ext:thm:uniform}. Together with the upper bound
and matching examples of Gravin and Jia~\cite{jia1}, it gives
$\sqrt{6\sqrt3-8}$ as the optimal guarantee uniform over both populations
and dimensions.

The proofs use a structural characterization. We next introduce its
assumptions and explain why they arise naturally in the facility-location
problem.

\begin{definition}[Unanimity]
A mechanism $f$ is \emph{unanimous} if $f(a,\ldots,a)=a$ for every
$a\in\mathbb R^d$.
\end{definition}

When all agents agree on one location, selecting that location is the
natural outcome. Unanimity is also forced by every finite approximation
guarantee: at a unanimous profile, the optimal social cost is zero.

\begin{definition}[Continuity and full range]
A mechanism is \emph{continuous} with respect to the usual Euclidean
topologies. Its \emph{full range} is the set of all outcomes that it can
select.
\end{definition}

Continuity means that small changes in the reports produce only a small
change in the outcome. It makes the mechanism stable under measurement or
communication errors. Peters, van der Stel, and Storcken proved that, for
multi-agent
strategyproof mechanisms in the Euclidean plane, continuity is equivalent
to convexity of the full range~\cite[Theorem~5]{peters1993range}. We prove
the same equivalence in every finite dimension.

\Needspace{5\baselineskip}
\begin{informaltheorem}[Automatic continuity]
For every finite population and dimension, a strategyproof mechanism is
continuous if and only if its full range is convex.
\end{informaltheorem}

The formal statement is \cref{intro:thm:range-continuity}. A unanimous
mechanism has full range $\mathbb R^d$, which is convex. Hence unanimity and
strategyproofness imply continuity, and every finite-ratio mechanism is
continuous.

\begin{definition}[Translation equivariance]
A mechanism $f$ is \emph{translation-equivariant} if
$f(x_1+v,\ldots,x_n+v)=f(x_1,\ldots,x_n)+v$ for every
$v\in\mathbb R^d$.
\end{definition}

If all agents move in the same direction by the same distance,
translation equivariance moves the selected outcome with them.
Equivalently, changing the origin of the coordinate system changes the
reports and the outcome in the same way. This is natural when the problem
has no externally distinguished location: the mechanism should depend on
relative positions rather than on absolute coordinates. Closely related
assumptions appear under the names \emph{translation invariance},
\emph{shift invariance}, and \emph{position invariance} in facility-location
characterizations~\cite{Lin2020,TangYZ20}. They rule out mechanisms tied to
a fixed landmark; \cref{ext:ex:phantom} gives such an origin-dependent
strategyproof rule.

The structural proof also uses \emph{positive homogeneity}, meaning that
$f(tx_1,\ldots,tx_n)=t f(x_1,\ldots,x_n)$ for every $t\ge0$. We call a
mechanism \emph{normalized} when it is both translation-equivariant and
positively homogeneous. Positive homogeneity is not an assumption of the
final characterization. Unanimity and strategyproofness first imply
continuity, and for a continuous mechanism translation equivariance implies
positive homogeneity. Thus translation equivariance is the spatial
assumption in the theorem, while \emph{normalized} is shorthand for the
properties available in the structural proof.

\begin{definition}[Bounded individual influence]
A mechanism has \emph{bounded individual influence} if, after the reports
of all other agents are fixed, one agent can induce only a bounded set of
outcomes by changing its report. The bound may depend on the fixed reports
of the other agents.
\end{definition}

This condition limits the effect of one extreme report: a single agent
cannot send the outcome arbitrarily far away while all other reports remain
fixed. The set of outcomes available to one agent is the standard option
set, or unilateral menu, used in strategyproof social
choice~\cite{BordesLaffondLeBreton2011}. For the normalized anonymous
mechanisms used in the proof, bounded individual influence is simply equivalent to
$f(a,b,\ldots,b)=b$: one objecting report cannot overturn the agreement of all the other agents.

The condition is also supplied by approximation. Every mechanism with ratio
below $n-1$ has bounded individual influence, and the strict threshold is
natural. On the line, the minimum rule is anonymous, unanimous,
translation-equivariant, and strategyproof, but one report can move its
outcome arbitrarily far. On the profile $(0,1,\ldots,1)$ its ratio is
$n-1$. Since the coordinate-wise median has ratio below $2$, every
mechanism that could improve on it for $n\ge3$ satisfies bounded individual
influence after normalization.

\Needspace{5\baselineskip}
\begin{informaltheorem}[Rigidity]
Every anonymous, unanimous, translation-equivariant strategyproof mechanism
with bounded individual influence is a fixed coordinate-wise order-statistic
rule in one orthonormal frame, with an interior rank in every coordinate.
Conversely, every such rule has these properties.
\end{informaltheorem}

The formal statement is \cref{intro:thm:rigidity}. The orthonormal frame is
part of the conclusion: the mechanism's incentive, symmetry, and influence
properties determine one frame that works at every profile. At three agents,
the only interior rank is the median rank.

The assumptions of the rigidity theorem arise directly in the optimality
proof. A finite ratio gives unanimity and hence continuity. The normalization
theorem produces a translation-equivariant representative with no larger
ratio. A ratio below $n-1$ gives bounded individual influence. Rigidity then
reduces every possible anonymous improvement over the median to a fixed-rank
rule, and coordinate reflections rule out such an improvement.

The unrestricted optimality result uses one further principle.

\Needspace{5\baselineskip}
\begin{informaltheorem}[Anonymous extraction]
If arbitrarily large populations admit strategyproof mechanisms with ratio
at most $R$, then any prescribed finite population admits an anonymous
strategyproof mechanism with ratio at most $R$.
\end{informaltheorem}

The formal statement is \cref{ext:thm:extraction}. Applying anonymous
optimality to the extracted mechanism gives the large-population theorem.
\subsection{Proof overview}\label{review:sec:overview}

The proof has four parts. First, we prove the structural characterization
for continuous mechanisms. Second, we show that unanimity and
strategyproofness already imply continuity. Third, we combine the
characterization with normalization and coordinate reflections to prove
anonymous median optimality. Finally, we construct an anonymous mechanism
from mechanisms on large populations and obtain the unrestricted
large-population result.

\subsubsection{Geometric rigidity}
We first assume continuity. Let
\(f:(\mathbb R^d)^n\to\mathbb R^d\) be anonymous, unanimous,
translation-equivariant, and strategyproof. Fixing the reports of all
agents except one leaves a set of outcomes that the remaining agent can
produce; we call this the agent's \emph{unilateral menu}, also known as its \emph{option set}.
Section~\ref{sec:prelim} shows that the mechanism selects the point in
this set closest to the agent's report. We assume that every such menu
is bounded and prove that \(f\) selects a fixed order statistic in
each coordinate of one orthonormal frame.


The order of Sections~\ref{sec:reconstruction}
and~\ref{sec:classification} reflects their roles in the proof.
Section~\ref{sec:reconstruction} explains how to recover a larger
mechanism once the smaller one has been classified.
Section~\ref{sec:classification} uses this recovery step inside its
induction.

\paragraph{Section~\ref{sec:absorption}: reducing the population.}
Fix the reports \(z\) of \(n-2\) agents, leaving two reports \(a\) and
\(b\) free. This gives a symmetric two-agent rule
\(h_z(a,b)=f(a,b,z)\). Its menus are compact.
Section~\ref{sec:absorption} proves that there is a unique point \(g(z)\)
such that \(f(a,g(z),z)=g(z)\) for every \(a\). In other words, once one
of the two free agents reports \(g(z)\), the outcome remains \(g(z)\)
regardless of the other report.

As \(z\) varies, these points define a new mechanism \(g\) on \(n-2\)
agents. The new mechanism inherits anonymity, strategyproofness,
continuity, and the translation and scaling properties used in the
structural proof. This step reduces the population, but it does not classify \(g\). In
particular, the menus of \(g\) may be unbounded even when those of
\(f\) are bounded. When only one report remains, for example, unanimity
makes \(g\) the identity map, whose range is all of \(\mathbb R^d\).

\paragraph{Section~\ref{sec:core-geometry}: determining the shape of the
smaller menus.}
Fix all but one report of \(g\), and let \(C\) be the resulting menu. Section~\ref{sec:absorption} represents \(C\) through a three-agent
restriction of the original mechanism. For fixed outside reports \(w\),
define \(H(a,b,c)=f(a,b,c,w)\). If \(p\) is the point of \(C\) closest to
\(c\), the defining property of the smaller mechanism gives
\(H(a,p,c)=p\) for every \(a\).

Section~\ref{sec:core-geometry} uses this identity to determine both the
operation of \(H\) on \(C\) and the shape of \(C\). Fix \(a,b\in C\) and
vary the third report. The resulting outcomes form a set that plays the
role of the interval between \(a\) and \(b\). Strategyproofness implies
that these intervals are nested in the expected way and that the three
intervals determined by \(a,b,c\) meet at the single point \(H(a,b,c)\).
These are the interval properties of a median operation.

The Euclidean nearest-point identities then determine an orthonormal
frame in which \(H(a,b,c)\) is the coordinate-wise median of \(a,b,c\).
This describes the operation on \(C\), but it does not yet determine the
boundary of \(C\). To do that, the proof also uses reports outside \(C\).
These reports rule out supporting boundaries that mix two or more
coordinate directions. It follows that \(C\) is an orthogonal box in the
same frame.

Applying this argument to every menu of \(g\) shows that all of
them are orthogonal boxes. Their frames may still depend on the agent
and on the fixed outside reports. Aligning these frames is the task of
Section~\ref{sec:classification}.

\paragraph{Section~\ref{sec:reconstruction}: recovering the original
mechanism.}
Section~\ref{sec:reconstruction} proves the step needed after the smaller
mechanism has been classified. Suppose that \(g\) uses one frame and
selects rank \(k_j\) in coordinate \(j\). We want to determine the
outcome selected by the original mechanism \(f\).

For a three-agent restriction with core box \(C\),
Section~\ref{sec:reconstruction} proves a clipping identity: clip each
of the three reports to \(C\), take their coordinate-wise median, and
one obtains the part of the original outcome that lies inside \(C\).
The fixed outside reports can be chosen so that the faces of \(C\)
identify the same frame used by \(g\).

The proof then applies the clipping identity to three reports occupying
three consecutive positions in coordinate \(j\). Their median is the
middle of these three positions. Consequently, if \(g\) selects rank
\(k_j\), then \(f\) selects rank \(k_j+1\) in the same coordinate and
the same frame.

Thus Section~\ref{sec:reconstruction} proves the implication
\[
    \text{\(g\) uses fixed ranks \(k_j\)}
    \quad\Longrightarrow\quad
    \text{\(f\) uses fixed ranks \(k_j+1\)}.
\]

\paragraph{Section~\ref{sec:classification}: finding one frame for all
box-shaped menus.}
Section~\ref{sec:classification} studies a continuous mechanism whose
menus are already known to be orthogonal boxes, although different
boxes may initially use different axes. Its goal is to prove that all
boxes use one frame and that the mechanism selects one fixed rank in
each coordinate.

The first step is to show that all menus extend without bound in
the same directions. The proof then separates into two cases.

If there is no nonzero common unbounded direction, every menu is
bounded. Section~\ref{sec:absorption} can then be applied again,
reducing the population by two. Section~\ref{sec:core-geometry} shows
that the smaller mechanism still has box-shaped menus. The
induction hypothesis classifies the smaller mechanism, and
Section~\ref{sec:reconstruction} transfers that classification back to
the original one.

If there is a common unbounded direction,
Section~\ref{sec:classification} chooses one such direction that is an
axis of every menu. Projection onto each box then separates into
two parts: one along this axis and one in the perpendicular space. The
mechanism therefore splits into a one-dimensional strategyproof rule
and a mechanism in dimension \(d-1\). The one-dimensional rule is a
fixed order statistic, while the induction hypothesis classifies the
mechanism in the perpendicular space. Combining the two parts gives one
frame and fixed ranks for the original mechanism.

The base cases are mechanisms on a line and mechanisms with one or two
agents. The two cases above therefore give an induction on both the
population size and the dimension.

\paragraph{Completing the characterization.}
We now return to the original mechanism \(f\).
Section~\ref{sec:absorption} constructs the smaller mechanism \(g\).
Section~\ref{sec:core-geometry} shows that every menu of \(g\) is
an orthogonal box. Section~\ref{sec:classification} proves that \(g\)
uses one common frame and fixed coordinate ranks.
Section~\ref{sec:reconstruction} then recovers \(f\) in the same frame
and increases every selected rank by one. Since the ranks of \(g\) lie
between \(1\) and \(n-2\), the ranks of \(f\) lie between \(2\) and
\(n-1\).

This proves rigidity under continuity.
Section~\ref{sec:automatic-regularity} shows that a unanimous
strategyproof mechanism is automatically continuous, and hence gives
the final characterization stated in the introduction.
\begin{figure}[htbp]
\centering
\begin{tikzpicture}[
 node distance=4mm and 3.2mm,
 box/.style={draw=black!65,rounded corners=1.5pt,align=center,
             text width=2.15cm,minimum height=1.05cm,
             inner xsep=3pt,inner ysep=3pt,font=\small},
 branch/.style={draw=black!55,rounded corners=1.5pt,align=center,
                text width=2.7cm,minimum height=.9cm,
                inner xsep=3pt,inner ysep=2pt,font=\footnotesize},
 arrow/.style={-{Latex[length=1.7mm]},line width=.55pt}
]
\node[box] (f) {bounded-menu\\mechanism $f_n$};
\node[box,right=of f] (g) {Section~\ref{sec:absorption}\\absorber $g_{n-2}$};
\node[box,right=of g] (boxes) {Section~\ref{sec:core-geometry}\\box menus of $g$};
\node[box,right=of boxes] (cone) {Section~\ref{sec:classification}\\common cone $K$};
\node[box,right=of cone] (ranks) {one frame\\fixed ranks for $g$};
\node[box,right=of ranks] (lift) {Section~\ref{sec:reconstruction}\\lift ranks to $f_n$};
\draw[arrow] (f)--(g);
\draw[arrow] (g)--(boxes);
\draw[arrow] (boxes)--(cone);
\draw[arrow] (ranks)--(lift);
\node[branch,below left=7mm and -2mm of cone] (zero)
 {$K=\{0\}$: reduce the\\population by two};
\node[branch,below right=7mm and -2mm of cone] (nonzero)
 {$K\ne\{0\}$: split one axis\\and reduce the dimension};
\draw[arrow] (cone)--(zero);
\draw[arrow] (cone)--(nonzero);
\draw[arrow] (zero.east)-|([xshift=-2mm]ranks.south);
\draw[arrow] (nonzero.north east)-|([xshift=2mm]ranks.south);
\end{tikzpicture}
\caption{Rigidity roadmap. Core geometry gives box-shaped menus. The common
recession cone selects a population or dimension reduction, producing one
frame and fixed ranks for the absorber. Rank lifting reconstructs the
original mechanism.}
\label{fig:rigidity-roadmap}
\end{figure}

\subsubsection{Automatic continuity}
\label{overview:regularity}

Peters, van der Stel, and Storcken proved that, for multi-agent
strategyproof mechanisms in the Euclidean plane, continuity is equivalent
to convexity of the full range~\cite{peters1993range}. We extend this
equivalence to every finite dimension.

Suppose first that the mechanism is continuous. The map obtained by making
all agents submit the same report has the same range as the original
mechanism. Applying the one-report projection theorem to this map shows that
the full range is convex.

For the converse, Appendix~\ref{ext:app:regularity} first studies a
two-report mechanism with convex full range. On the diagonal, its outcome
is projection onto that range. The incentive inequalities are then written
as inequalities for two convex scalar functions, one for each report. Their
derivatives describe how the outcome changes when either report moves.
Moving the two reports toward each other compares these derivatives with
the derivative of the diagonal projection. Since projection is
nonexpansive, this comparison bounds the response to each report and proves
that the two-report mechanism is Lipschitz continuous.

The result for an arbitrary population follows by grouping agents that
submit a common report. Starting with the full range, the two-report result
shows successively that the menus of smaller coalitions are convex, until
every unilateral response is a convex projection. Since unanimity gives
full range \(\mathbb R^d\), this theorem supplies continuity for the
mechanisms covered by the rigidity result.

\subsubsection{Anonymous optimality}
\label{overview:anonymous-optimality}

Section~\ref{sec:approximation} applies rigidity to the approximation
problem. We first replace an anonymous competitor by a
translation-equivariant mechanism with no larger approximation ratio. The
representative remains continuous, unanimous, and strategyproof, so
translation equivariance also implies positive homogeneity; hence the
representative is normalized.

If its approximation ratio is below \(n-1\), every unilateral menu is
bounded. The rigidity theorem therefore reduces the representative to a
fixed coordinate-wise order-statistic rule in one orthonormal frame. It
remains to compare such a rule with the coordinate-wise median.

Reflecting coordinate \(j\) changes the selected rank \(k_j\) into the
complementary rank \(n+1-k_j\). Across all coordinate reflections, the
resulting outcomes are the corners of a coordinate box containing the
median outcome. Since social cost is convex in the facility location, the
median's cost is at most the largest cost among these reflected outcomes.
Each reflected mechanism has the same approximation ratio as the original
fixed-rank rule, which proves the comparison. This is the
\(2^d\)-reflection extension of the planar argument of Goel and
Hann-Caruthers~\cite[Lemma~6]{GoelHC23}.

\subsubsection{Anonymous extraction}
\label{overview:extraction}

Section~\ref{ext:sec:extraction} removes anonymity from the
large-population result. Fix a target profile of finitely many
\emph{virtual} agents. We use these virtual reports to generate reports
for a much larger \emph{physical} population and apply a source mechanism
on that population. Because the sampling rule treats the virtual reports
symmetrically, the expected source output is anonymous.

A few physical labels may have a large effect on the outcome. We fix a
bounded number of such labels at the origin. The remaining labels each
have small influence, and a product-measure concentration inequality then
shows that the random source output is close to its expectation. The bound
on the number of fixed labels does not depend on the size of the source
population, so their fraction vanishes as that population grows.

The expected outputs have a common Lipschitz bound and therefore converge,
along a subsequence, to a deterministic anonymous mechanism. To prove
strategyproofness, we couple the samples for a truthful virtual report and
a deviation. The physical agents assigned to that virtual report are
changed one at a time, so the incentive inequalities of the source
mechanism can be chained. Concentration allows the random outputs to be
replaced by their expectations, and the resulting inequality passes to the
limit. The same limiting argument transfers the approximation guarantee.

The extracted mechanism has the same approximation bound on any prescribed
finite population. Applying fixed-population anonymous optimality to this
mechanism yields the unrestricted large-population lower bound.

\subsection{Further related work}
\label{sec:further-related-work}

\paragraph{Social choice.}
Beyond the Euclidean characterizations discussed in the introduction,
strategyproof aggregation has been studied on several other structured
domains. Ching~\cite{ching1997strategy} gives a direct characterization
of median rules on a line, while Barber\`a, Mass\'o, and
Serizawa~\cite{barbera1998strategy} allow the feasible outcomes to form
general compact subsets of Euclidean space. Nehring and
Puppe~\cite{NehringPuppe2007} develop a common framework for generalized
single-peaked preferences based on median spaces and describe
strategyproof rules as issue-by-issue voting. Schummer and
Vohra~\cite{schummer2002strategy} study location rules on trees and
networks containing cycles. The option-set approach to Euclidean
preferences and coalitional incentives is developed further by Bordes,
Laffond, and Le Breton~\cite{BordesLaffondLeBreton2011}, and Tang, Yu,
and Zhao~\cite{TangYZ20} characterize group-strategyproof location
mechanisms in strictly convex spaces. Barber\`a's
survey~\cite{barbera2011strategyproof} gives a broader account of the
strategyproof social-choice literature.

\paragraph{Facility location.}
The algorithmic facility-location literature studies many variants of
the basic model; see~\cite{chan2021mechanism} for a survey. Randomized
mechanisms have been used to improve approximation and fairness
guarantees on networks and in Euclidean spaces
\cite{AlonFPT10,AzizLSW22,barak26}. Work on multiple facilities studies
both deterministic and randomized mechanisms on the line and in more
general metric spaces
\cite{LuSWZ10,FotakisT10,FotakisT14,EscoffierGTPS11,
chan2026,ma2026,jia2026}.

A separate line of work considers objectives other than total distance.
These include the maximum or egalitarian cost
\cite{procaccia2013approximate,nguyen2026,barak2026exante,
hastings2026}, least-squares and concave cost functions
\cite{FeldmanW13,FotakisT16}, and the \(L_p\)-norm of the vector of
agents' costs
\cite{FeigenbaumSY17,ChanLW26}. Other variants allow heterogeneous,
dual, or obnoxious preferences
\cite{SerafinoV16,zou15,cheng2013strategy}, restrict the location space
to trees or cycles
\cite{DokowFMN12,FilimonovM22}, or allow facilities and agents to move
between stages
\cite{fotakis2019,dekeijzer2022}. Prediction-augmented models study the
tradeoff between improved performance under accurate advice and
worst-case robustness
\cite{AgrawalBGTX22,barak2024mac,BalkanskiGS24,
ChristodoulouSV24,istrate2022mechanism,fang2024mechanism,chen2024}.

Concurrent independent work studies fixed-population optimality without
anonymity for a single facility in the Euclidean plane
\cite{lin2026opt}. Its anonymity-free optimality result is restricted to
\(\mathbb R^2\).

\paragraph{Roadmap.}
Section~\ref{sec:prelim} introduces the model and the projection tools
used throughout the paper. Sections~\ref{sec:absorption}--\ref{sec:classification}
develop the structural characterization: Section~\ref{sec:absorption}
constructs the smaller mechanism, Section~\ref{sec:core-geometry}
determines the geometry of its menus, Section~\ref{sec:reconstruction}
recovers the original mechanism, and Section~\ref{sec:classification}
aligns all menus in one frame. Section~\ref{sec:automatic-regularity}
derives continuity and completes the characterization.
Section~\ref{sec:approximation} proves anonymous median optimality, and
Section~\ref{ext:sec:extraction} proves the unrestricted
large-population result. The appendices contain the analytic details of
normalization, automatic continuity, and anonymous extraction, together
with the separating examples.

\section{The model and projection tools}\label{sec:prelim}
Section~\ref{sec:results} described the assumptions and results without
introducing the technical notation. We now give the formal model, define the
approximation benchmarks and coordinate-wise median, and distinguish the
outcomes attainable by one agent, a coalition, and the whole population. We
then record the projection tools used in the structural proof.
\Cref{sec:automatic-regularity} later derives continuity from unanimity and
completes the characterization stated in the introduction.

\subsection{Model, approximation, and coordinates}\label{sec:model}
Fix $E=\R^d$ with $d\ge1$, inner product $\ip{\cdot}{\cdot}$, and induced
Euclidean norm $\norm{\cdot}$. For $n\ge1$, a mechanism is a deterministic
map $f:E^n\to E$. A profile $x=(x_1,\ldots,x_n)$ records the reported ideal
points; $x_{-i}$ records the reports other than $i$, and $(a,x_{-i})$ replaces
report $i$ by $a$. Strategyproofness is the condition
\begin{equation}\label{eq:strategyproofness}
 \norm{a-f(a,x_{-i})}\le\norm{a-f(b,x_{-i})}
 \qquad(i\in[n],\ a,b\in E,
         \ x_{-i}\in E^{n-1}).
\end{equation}
Anonymity means
$f(x_{\pi(1)},\ldots,x_{\pi(n)})=f(x_1,\ldots,x_n)$ for every
permutation $\pi$, and unanimity means $f(a,\ldots,a)=a$ for every
$a\in E$. Continuity refers to the Euclidean topologies. Write
$\ExtS_{n,d}$ for the deterministic strategyproof class,
$\ExtA_{n,d}$ for its anonymous subclass, and $\F_{n,d}$ for the continuous,
anonymous, unanimous members. The continuous class is used for the geometric
proof; the main comparison classes are $\ExtS_{n,d}$ and $\ExtA_{n,d}$.

The social cost and optimum are
\[
 \SC(x,y)=\sum_{i=1}^n\norm{x_i-y},\qquad
 \OPT(x)=\min_{y\in E}\SC(x,y).
\]
The minimum exists because the objective is continuous and coercive.
For $n\ge2$, set
\[
 \AR(f)=\sup_{x\in E^n}\frac{\SC(x,f(x))}{\OPT(x)},
\]
where the quotient is one if both costs vanish and infinite if only the
denominator vanishes. Finite ratio therefore gives
$\SC(x,f(x))\le\AR(f)\OPT(x)$ at every profile and implies unanimity by
applying the inequality to $(a,\ldots,a)$.

For an orthonormal basis \(\B=(e_1,\ldots,e_d)\), write
\(x_{i,j}=\ip{e_j}{x_i}\), and let \(x_{(r),j}\) be the \(r\)th
smallest value among \(x_{1,j},\ldots,x_{n,j}\). For a rank vector
\(k=(k_1,\ldots,k_d)\in[n]^d\), define
\[
    Q_{\B,k}(x)
    :=
    \sum_{j=1}^d x_{(k_j),j}e_j.
\]

Let \(\B_0\) be the standard basis of \(\mathbb R^d\). We write
\(\CM_{n,d}\) for the coordinate-wise median in this basis: it uses rank
\((n+1)/2\) in every coordinate when \(n\) is odd, and rank \(n/2\) in
every coordinate when \(n\) is even.\footnote{This standard choice does
not affect the approximation ratio. Orthogonal changes of coordinates
preserve all Euclidean distances and transform the median output in the
same way as the reports. For even \(n\), reflecting any coordinate
exchanges the lower and upper middle ranks in that coordinate. Hence
every orthonormal basis and every fixed choice of middle ranks have the
same approximation ratio.} We define 
$R_{n,d}:=\AR(\CM_{n,d})$.
Finally, set
\begin{equation}\label{ext:eq:uniform-definitions}
    a_{n,d}
    :=
    \inf_{f\in\ExtS_{n,d}}\AR(f),
    \qquad
    C_d
    :=
    \sup_{n\ge2}R_{n,d}.
\end{equation}

For $v\in E$ and $t\ge0$, use $x+v=(x_1+v,\ldots,x_n+v)$ and
$tx=(tx_1,\ldots,tx_n)$. Translation equivariance and positive homogeneity
are, respectively, $f(x+v)=f(x)+v$ and $f(tx)=tf(x)$.
A mechanism having both properties is called normalized.
\Cref{cor:automatic-homogeneity} shows that, within $\F_{n,d}$, translation
equivariance already implies positive homogeneity. Thus translation
equivariance is the independent spatial assumption in the final rigidity
theorem, while \emph{normalized} is shorthand used in the continuous
structural proof.

An orthogonal box is a product of nonempty closed intervals in an
orthonormal affine frame. Factors may be bounded, halflines, full lines,
or singletons. For a fixed frame, $\boxB(a,b)$ denotes the coordinate box
with opposite corners $a,b$. A scalar projection onto a closed interval
$J=[\ell,u]$ is the clamp $P_Js=\min\{u,\max\{\ell,s\}\}$, allowing
infinite endpoints. Projection onto an orthogonal box acts coordinatewise.
For a convex set in a proper affine subspace, dimension, interior, and
boundary are relative to its affine hull unless stated otherwise. An
intrinsic orthonormal frame extends to the ambient space by singleton
factors. Coordinate order statistics commute with translations, so these
affine descriptions are independent of the selected origin.

\subsection{Option sets and common reports}\label{sec:option-sets}
Fixing some reports leaves a set of outcomes attainable by the remaining
agents. This is the classical option-set approach of Laffond and
Barber\`a--Peleg, discussed in the Euclidean setting by Bordes, Laffond,
and Le Breton~\cite[Section~1]{BordesLaffondLeBreton2011}.
For a nonempty $S\subseteq[n]$ and fixed outside profile $z$, define $C_S(z)=\{f(r_S,z):r_S\in E^S\}$ and $T_{S,z}(a)=f(a^S,z)$.
Here $a^S$ places the same report in all positions of $S$; we also write
$a^k$ for $k$ repeated reports. Single-agent menus are denoted
$M_i(x_{-i})=C_{\{i\}}(x_{-i})$. The full range is
$\mathcal R(f)=f(E^n)$. We allow $E^0=\{()\}$ for restrictions with no
outside reports.

Strategyproofness makes each truthful unilateral outcome a nearest point
in its menu. Outcome replacement extends that interpretation to coalition
diagonals, because all intermediate incentive comparisons then use the
same true point.

\begin{lemma}\label{lem:raw-menus}\label{ext:lem:menus-no-continuity}
Let $f:E^n\to E$ be strategyproof. If $f(x)=y$, replacing any report by $y$
preserves the outcome. For every nonempty coalition $S$ and outside profile
$z$, the map $T_{S,z}$ satisfies the one-input strategyproof inequality,
has range $C_S(z)$, and chooses a nearest point of that set at every input.
The set $C_S(z)$ is closed and every point of it is fixed by $T_{S,z}$.
In particular,
\begin{equation}\label{eq:diagonal-full-range}
 \mathcal R(f)=\{f(a,\ldots,a):a\in E\}.
\end{equation}
\end{lemma}
\begin{proof}
If a report can obtain $y$, an agent with true point $y$ can obtain zero
cost and must do so truthfully. This proves outcome replacement. Replacing
all the reports in $S$ successively shows that every $y\in C_S(z)$ satisfies
$f(y^S,z)=y$. Thus the diagonal has the entire coalition range.

Change the common report of $S$ from $a$ to $b$, one position at a time.
At each step the report being changed is still $a$, so its incentive
inequality compares the successive outputs by distance to $a$. Chaining
these inequalities yields $\norm{a-T_{S,z}(a)}\le\norm{a-T_{S,z}(b)}$.
The range identity makes $T_{S,z}(a)$ a nearest point of $C_S(z)$. If
$y_k\in C_S(z)$ converges to $y$, comparison at input $y$ gives
$\norm{y-T_{S,z}(y)}\le\norm{y-y_k}\to0$. Hence $y=T_{S,z}(y)$ is in
$C_S(z)$, proving closedness. Taking $S=[n]$ gives
\eqref{eq:diagonal-full-range}.
\end{proof}

\subsection{Continuous responses and convex projection}\label{sec:projection-tools}
We now work with continuous responses. The following one-report
characterization goes back to Kim and Roush~\cite{kim1984nonmanipulability}
and is stated by Peters, van der Stel, and Storcken~\cite[Theorem~2]{peters1993range}
in arbitrary finite spatial dimension. We include the proof to identify
the projection properties used below.

\begin{proposition}\label{lem:projection}
Suppose $T:E\to E$ is continuous and
\begin{equation}\label{eq:one-input-sp}
 \norm{a-T(a)}\le\norm{a-T(b)}\qquad(a,b\in E).
\end{equation}
Then $C=T(E)$ is closed and convex, and $T=P_C$.
\end{proposition}
\begin{proof}
The incentive inequality makes $T(a)$ a nearest point of $C$ to $a$.
Every point of $C$ is fixed, and closedness follows by comparing the
response at the limit of a convergent sequence in $C$, as in
\cref{lem:raw-menus}.

Let $c\in C$ be any nearest point to $a$, and set
$a_t=(1-t)a+tc$ for $0<t<1$. For every $s\in C$,
\[
 \norm{a_t-s}^2-\norm{a_t-c}^2
 =(1-t)\bigl(\norm{a-s}^2-\norm{a-c}^2\bigr)+t\norm{s-c}^2.
\]
This is strictly positive for $s\ne c$. Thus $T(a_t)=c$, and continuity
at $a$ gives $T(a)=c$. The nearest point is unique at every input.
The finite-dimensional Bunt--Motzkin theorem makes such a Chebyshev set
convex; see~\cite{BauschkeWangYeYuan2009}. Hence $T=P_C$.
\end{proof}

We record standard convex-projection identities~\cite{Rockafellar1970}.
For nonempty closed convex $C$, the unique nearest point exists by
compactness of a bounded minimizing sequence. It satisfies
\begin{equation}\label{eq:projection-normal}
 p=P_Ca\quad\Longleftrightarrow\quad
 p\in C\ \text{ and }\ \ip{a-p}{b-p}\le0\quad(b\in C).
\end{equation}
Necessity follows by comparing $p$ with $p+t(b-p)$ and differentiating at
$t=0$; expansion of squared distance proves sufficiency. Applying this
inequality to $p=P_Ca$ and $q=P_Cb$ and adding gives firm nonexpansiveness,
\begin{equation}\label{eq:firm}
 \norm{P_Ca-P_Cb}^2\le\ip{a-b}{P_Ca-P_Cb}.
\end{equation}
In particular $P_C$ is $1$-Lipschitz. If $D\subseteq C$ is nonempty,
closed, and convex, a minimizer over $C$ that belongs to $D$ also minimizes
over $D$:
\begin{equation}\label{eq:restrict-minimizer}
 P_Ca\in D\quad\Longrightarrow\quad P_Da=P_Ca.
\end{equation}

For continuous mechanisms, the preceding one-report characterization
applies to every coalition diagonal. This is the projection form of
\cref{lem:raw-menus}.

\begin{corollary}\label{lem:coalition}
Let $f:E^n\to E$ be continuous and strategyproof. For each nonempty
coalition $S$ and fixed outside profile $z$, the set $C_S(z)$ is closed
and convex, and $f(a^S,z)=P_{C_S(z)}a$.
\end{corollary}
\begin{proof}
The diagonal response is continuous and satisfies
\eqref{eq:one-input-sp} by \cref{lem:raw-menus}. Apply
\cref{lem:projection} to its range $C_S(z)$.
\end{proof}

Equation~\eqref{eq:projection-normal} shows that moving a report farther
along the ray from its outcome leaves that outcome unchanged. Together with
nonexpansiveness, this gives the profile identities needed for absorption
and normalization. Write
$d_1(x,x')=\sum_i\norm{x_i-x_i'}$.

\begin{lemma}\label{cor:replacement}
Let $f:E^m\to E$ be continuous and strategyproof, and let $f(x)=y$.
For every agent $i$ and every $t\ge0$,
\begin{align}
 f(y,x_{-i})&=y,\label{eq:replacement}\\
 f\bigl(y+t(x_i-y),x_{-i}\bigr)&=y.\label{eq:radial-one}
\end{align}
Moreover,
\begin{align}
 \norm{f(x)-f(x')}&\le d_1(x,x'),\label{eq:global-lipschitz}\\
 f\bigl(y+t_1(x_1-y),\ldots,y+t_m(x_m-y)\bigr)&=y
 \quad(t_1,\ldots,t_m\ge0).\label{eq:radial-all}
\end{align}
\end{lemma}
\begin{proof}
Each unilateral response is projection onto its menu.
Equation~\eqref{eq:replacement} is the fixed-point property, and multiplying
the residual $x_i-y$ in \eqref{eq:projection-normal} by $t\ge0$ proves
\eqref{eq:radial-one}. Changing reports one at a time and using
nonexpansiveness proves \eqref{eq:global-lipschitz}. Applying
\eqref{eq:radial-one} successively proves \eqref{eq:radial-all}; the
outcome remains $y$ after each change.
\end{proof}

For translation-equivariant rules, the common radial change also specifies
the response to scaling the entire profile.

\begin{corollary}\label{cor:automatic-homogeneity}
A continuous, unanimous, strategyproof, translation-equivariant mechanism
is positively homogeneous.
\end{corollary}
\begin{proof}
Put $y=f(x)$. Equation~\eqref{eq:radial-all} gives
$f(tx+(1-t)y)=y$. Translation equivariance gives
$f(tx)+(1-t)y=y$, so $f(tx)=tf(x)$ for every $t\ge0$.
\end{proof}

\subsection{Bounded influence and agreement}
We say that a mechanism has \emph{bounded individual influence} when every
unilateral menu $M_i(x_{-i})$ is bounded. For normalized mechanisms, this
condition has an equivalent form in terms of agreement by all but one agent. We will use it to prove
unanimity of the smaller mechanism constructed in Section~\ref{sec:absorption}.

\begin{lemma}\label{rem:near-unanimity}
Let $m\ge3$ and let $f\in\F_{m,d}$ be normalized. Its unilateral menus are
bounded if and only if
\begin{equation}\label{eq:near-unanimity}
 f(a,b^{m-1})=b\qquad(a,b\in E).
\end{equation}
\end{lemma}
\begin{proof}
The menu of $a\mapsto f(a,0^{m-1})$ contains zero by unanimity and is
invariant under nonnegative scaling. If bounded, it is $\{0\}$. Translation
equivariance gives \eqref{eq:near-unanimity}. Conversely, that identity and
\eqref{eq:global-lipschitz} give
$\norm{f(a,z_1,\ldots,z_{m-1})}
 \le\sum_{j=1}^{m-1}\norm{z_j}$
for every $a$, by comparison with $(a,0^{m-1})$. Anonymity gives the bound
for every position.
\end{proof}

\subsection{A joint differential influence bound}\label{ext:sec:influence}
The next analytic estimate is shared by normalization and anonymous
extraction. Projection derivatives are positive symmetric matrices; radial
invariance also bounds their sum. The constant is independent of the
number of reports, which will allow us to sample large populations.

\begin{lemma}\label{ext:lem:influence}
Let $F:E^r\to E$ be continuous and strategyproof, with $r\ge1$.
At almost every profile, its derivative blocks $A_i=D_iF$ satisfy
\begin{equation}\label{ext:eq:influence-matrix}
 A_i=A_i^{\mathsf T}\succeq0,\qquad
 \sum_iA_i\preceq I_d,\qquad
 \sum_i\operatorname{tr}A_i\le d.
\end{equation}
For all profiles,
\begin{equation}\label{ext:eq:profile-max-lipschitz}
 \norm{F(x)-F(x')}\le\sqrt d\,\max_i\norm{x_i-x_i'}.
\end{equation}
\end{lemma}

\Cref{app:translation-derivative} proves the estimates. Fixing some reports
leaves a continuous strategyproof map of the remaining reports, so the
lemma applies to each restriction in its own profile space.

\section{The absorbing reduction and ternary cores}\label{sec:absorption}
The purpose of this section is to replace a mechanism on $n$ reports by
one on $n-2$ reports. Fixing the other reports leaves a symmetric binary
rule. Compact projection menus give it a unique absorbing point. As the
fixed reports vary, these points form a smaller strategyproof mechanism.
We then express each of its menus as a core of a ternary restriction,
which will make the geometric analysis in Section~\ref{sec:core-geometry}
applicable.

\subsection{An absorbing point for a binary restriction}
The projection inequalities turn the binary response into a partial
order on its range. Compactness then produces a greatest element.

\begin{lemma}\label{lem:binary-absorption}
Let $h:E^2\to E$ be symmetric, continuous, and strategyproof in each argument. Suppose that $C_a=\{h(a,b):b\in E\}$
is compact for every fixed $a\in E$. Then there is a unique $p\in E$ such that
$h(a,p)=p$ holds for all $a\in E$.
\end{lemma}
\begin{proof}
Let $D=\{h(a,b):a,b\in E\}$ be the full range. If $y=h(a,b)$, replacing both reports successively by $y$ gives $h(y,y)=y$. Define a relation on $D$ by
$a\preceq b\quad\Longleftrightarrow\quad h(a,b)=b$.
It is reflexive by the preceding observation. If both $a\preceq b$ and $b\preceq a$, symmetry gives $a=h(b,a)=h(a,b)=b$, so it is antisymmetric.

To check transitivity, suppose $h(a,b)=b$ and $h(b,c)=c$, and put $w=h(a,c)$. The response with first report $a$ is projection onto $C_a$, so $w=P_{C_a}c$ and $b\in C_a$. The response with second report $c$ is projection onto $D_c=\{h(x,c):x\in E\}$, so $c=P_{D_c}b$ and $w\in D_c$. Consequently, $\ip{c-w}{b-w}\le0$,
 and $\ip{b-c}{w-c}\le0$.
With $u=w-c$ and $v=b-c$, these are
$\norm{u}^2\le\ip{u}{v}$ and $\ip{u}{v}\le0$. Hence $u=0$, or $w=c$. The relation is therefore a partial order. 

For $a\in D$, the fixed-point characterization of a unilateral menu gives
$C_a=\{y\in D:a\preceq y\}$.
For any $a,b\in D$, the outcome $h(a,b)$ is an upper bound of both: replacing either report by that outcome proves both required relations. By repeated use of transitivity, every finite subset of $D$ has an upper bound.

Fix $a_0\in D$. The sets $C_{a_0}\cap C_a$, indexed by $a\in D$, are closed subsets of the compact set $C_{a_0}$. Every finite collection has nonempty intersection, because its finitely many indices together with $a_0$ have a common upper bound. Compactness therefore gives $p\in\bigcap_{a\in D}C_a$.
In the partial order, $p$ is a greatest element. If $y\in C_p$, then $p\preceq y$, while greatestness gives $y\preceq p$. Thus $C_p=\{p\}$. Since $C_p=\{h(p,a):a\in E\}$, we have $h(p,a)=p$ for every $a\in E$. Symmetry gives the stated identity. If $p$ and $q$ both absorb every report, then $h(p,q)=p=q$, proving uniqueness.
\end{proof}

\subsection{The smaller mechanism}
For each outside profile $z$, let $g(z)$ be the absorbing point of
$(a,b)\mapsto f(a,b,z)$. To use $g$ in the population induction, we
compare two such points while holding both of them as reports in the
original mechanism. This transfers the incentive and Lipschitz bounds
to the smaller rule.
\begin{proposition}\label{prop:absorber}
Let $n\ge3$, and let $f:E^n\to E$ be anonymous, continuous, and strategyproof with compact unilateral menus. There is a unique map $g:E^{n-2}\to E$ satisfying
\begin{equation}\label{eq:absorber}
 f(a,g(z),z)=g(z)\qquad(a\in E,\ z\in E^{n-2}).
\end{equation}
The map $g$ is anonymous, continuous, and strategyproof. If $f$ is also normalized and unanimous, then so is $g$.
\end{proposition}
\begin{proof}
For each outside profile $z$, consider $h_z(a,b)=f(a,b,z)$. It is symmetric by anonymity, and its unilateral menus are menus of $f$. Apply \cref{lem:binary-absorption}; its unique absorbing point defines $g(z)$.

To compare outputs of $g$, suppose $z$ and $z'$ differ in just one report, say report $i$ among the outside agents, and put $p=g(z)$ and $p'=g(z')$. By \eqref{eq:absorber} and symmetry of the first two positions,
$f(p,p',z)=p$ and $f(p,p',z')=p'$.
We may hold $p,p'$ fixed and regard these as two responses of the original outside agent. Strategyproofness of $f$ gives $\norm{z_i-p}\le\norm{z_i-p'}$,
which is precisely the required incentive inequality for $g$. The same original response is a projection, and firm nonexpansiveness gives $\norm{p-p'}^2\le\ip{z_i-z_i'}{p-p'}$.
Therefore $g$ is separately $1$-Lipschitz and hence jointly continuous, by changing its reports one at a time. A permutation of $z$ leaves the binary rule $h_z$ unchanged, so uniqueness of its absorber gives anonymity of $g$.

Now suppose $f$ is normalized and unanimous. By \cref{rem:near-unanimity}, $f(a,b^{n-1})=b$. Take $z=b^{n-2}$ and $a=b$ in \eqref{eq:absorber}. The left side has $n-1$ reports equal to $b$, so it equals $b$, proving $g(b^{n-2})=b$.

For a common translation $v$, the point $g(z)+v$ absorbs $h_{z+v}$, because $f(a,g(z)+v,z+v)=f(a-v,g(z),z)+v=g(z)+v$.
Uniqueness yields $g(z+v)=g(z)+v$. For $t>0$, similarly,
$f(a,tg(z),tz)=t f(a/t,g(z),z)=tg(z)$,
so $g(tz)=t g(z)$. At $t=0$, use unanimity $g(0)=0$.
\end{proof}

Absorption preserves continuity, strategyproofness, and normalization.
Its menus may be unbounded: when one report remains, unanimity gives the
identity with range $E$. We will therefore use the induction hypothesis
for mechanisms with box menus, including unbounded boxes. The next
construction supplies the requisite box shape through a ternary core.

\subsection{Realizing a menu as a ternary core}
Fix an outside profile $w\in E^{n-3}$ and define
$H_w(a,b,c)=f(a,b,c,w)$, $T_w(c)=g(c,w)$ and $C_w=T_w(E)$.
The restriction $H_w$ is symmetric, continuous, and strategyproof.
The defining identity of $g$ relates it to the menu $C_w$ as follows.

\begin{proposition}\label{prop:core-realization}
Every unilateral menu $C_w$ of $g$ is closed and convex, and
\begin{equation}\label{eq:core-realization}
 H_w(a,P_{C_w}c,c)=P_{C_w}c\qquad(a,c\in E).
\end{equation}
\end{proposition}
\begin{proof}
By \cref{prop:absorber}, $T_w$ is a continuous strategyproof unilateral response of $g$. Thus $T_w=P_{C_w}$ by \cref{lem:projection}. Substituting $z=(c,w)$ into \eqref{eq:absorber} gives $H_w(a,T_w(c),c)=T_w(c)$.
\end{proof}

We use this projection identity to define the convex sets studied in
the next section.

\begin{definition}\label{def:core}
Let $H:E^3\to E$ be symmetric, continuous, and strategyproof in each argument. A nonempty closed convex set $C\subseteq E$ is a \emph{core} of $H$ if
\begin{equation}\label{eq:core}
 H(a,P_Cc,c)=P_Cc\qquad(a,c\in E).
\end{equation}
\end{definition}
Cores may be unbounded or lower-dimensional. The identity
\eqref{eq:core} relates them to reports throughout the ambient space.
In particular, the presence of a core report forces the outcome to
remain in the core.

\begin{lemma}\label{lem:core-closure}
If $C$ is a core of $H$, then an outcome belongs to $C$ whenever at least one report belongs to $C$. Moreover, for $p\in C$ and $a\in E$,
\begin{equation}\label{eq:core-repeated}
 H(p,p,a)=p,\qquad H(a,a,p)=P_Ca.
\end{equation}
\end{lemma}
\begin{proof}
By symmetry it suffices to fix $b\in C$ and consider $y=H(a,b,c)$. Set $p=P_Cy$. Outcome replacement and the core identity give $H(a,b,y)=y$ and $H(a,p,y)=p$.
Strategyproofness for the middle agent with true point $b$ yields $\norm{b-y}\le\norm{b-p}$. On the other hand, the projection inequality at $p=P_Cy$, tested against $b\in C$, gives
\[
 \norm{b-y}^2
 =\norm{b-p}^2+\norm{y-p}^2-2\ip{b-p}{y-p}
 \ge\norm{b-p}^2+\norm{y-p}^2.
\]
The two bounds imply $y=p\in C$.

For $p\in C$, put $c=p$ in \eqref{eq:core}; this gives $H(a,p,p)=p$. Symmetry proves the first identity in \eqref{eq:core-repeated}. With the third report fixed at $p$, the diagonal coalition response $a\mapsto H(a,a,p)$ takes values in $C$ by the first part of the lemma. It fixes every $a\in C$ by the first repeated-report identity, so its range is exactly $C$. Apply \cref{lem:coalition} to get the second identity.
\end{proof}

\section{The geometry of ternary cores}\label{sec:core-geometry}
Section~\ref{sec:absorption} realizes every unilateral menu of the
absorbing mechanism as a ternary core $C$ of a symmetric strategyproof rule
$H:E^3\to E$. In this section we determine both the operation induced by
$H$ on $C$ and the geometry of $C$.

The proof has three steps. First, fixing two reports in $C$ and varying the
third produces sets $I_{ab}$ that satisfy the interval axioms of a median
algebra. Second, the Euclidean projection identities determine one
orthonormal frame in which $H|_{C^3}$ is the coordinate-wise median. Third,
reports outside $C$ determine the supporting directions of the core and
force $C$ to be an orthogonal box.

The interval characterization, halfspace separation, and gate projections
used in the first two steps are standard facts from median-algebra
theory~\cite[Section~2]{Bowditch2016} \cite[Section~4]{CHATTERJI2010882}.
The arguments below derive the interval structure from the mechanism and
then identify its Euclidean realization. Together they prove the following
theorem.

\begin{theorem}\label{thm:core-box}
Let $H:E^3\to E$ be symmetric, continuous, and strategyproof in each argument, and let $C$ be a core of $H$. Then $C$ is an orthogonal box. In a corresponding orthonormal affine frame $\B$ of $\aff C$,
\[
 H(a,b,c)=\CM_{\B}(a,b,c)\qquad(a,b,c\in C).
\]
\end{theorem}

\subsection{Option intervals and a median algebra}\label{sec:option-intervals}
We first identify the abstract operation carried by the core. For
$a,b\in C$, vary the third report and define
$I_{ab}=\{H(a,b,x):x\in E\}$. By
\cref{lem:projection,lem:core-closure}, the set $I_{ab}$ is a nonempty
closed convex subset of $C$, it contains $a$ and $b$, and
\begin{equation}\label{eq:interval-projection}
 H(a,b,x)=P_{I_{ab}}x.
\end{equation}
We call $I_{ab}$ the \emph{option interval} between $a$ and $b$.

The identities $I_{aa}=\{a\}$ and $I_{ab}=I_{ba}$ follow from outcome
replacement and symmetry. To apply the interval characterization of median
algebras, it remains to prove nesting and unique triple intersection. The
next lemma derives these two properties, together with the Euclidean
diameter bound that will be used later.

\begin{lemma}\label{lem:interval-geometry}
For $a,b\in C$, the following hold.
\begin{enumerate}[label=(\roman*)]
\item The option interval lies in the ball with diameter $[a,b]$:
\begin{equation}\label{eq:diameter-ball}
 I_{ab}\subseteq B\left(\frac{a+b}{2},\frac{\norm{a-b}}2\right),
 \qquad \diam I_{ab}=\norm{a-b}.
\end{equation}
Here $B(q,r)=\{z:\norm{z-q}\le r\}$.
\item If $c\in I_{ab}$, then $I_{ac}\subseteq I_{ab}$.
\item For all $a,b,c\in C$,
\begin{equation}\label{eq:triple-intersection}
 I_{ab}\cap I_{bc}\cap I_{ca}=\{H(a,b,c)\}.
\end{equation}
\end{enumerate}
\end{lemma}
\begin{proof}
For (i), take $w=H(a,b,x)$. Fixing the first and third reports, the
second-report response sends $b$ to $w$ and $a$ to $a$, by
\cref{lem:core-closure}. Firm nonexpansiveness therefore gives
$\norm{w-a}^2\le\ip{b-a}{w-a}$. Completing the square places $w$ in
the ball with center $(a+b)/2$ and radius $\norm{a-b}/2$. Since
$a,b\in I_{ab}$, the diameter of $I_{ab}$ is exactly $\norm{a-b}$.

For (ii), fix the first report $a$ and consider the symmetric binary rule
$(u,v)\mapsto H(a,u,v)$. The transitivity calculation from
\cref{lem:binary-absorption} applies to this rule. If $c\in I_{ab}$ and
$w\in I_{ac}$, then $H(a,b,c)=c$ and $H(a,c,w)=w$. Transitivity gives
$H(a,b,w)=w$, hence $w\in I_{ab}$.

For (iii), symmetry shows that $H(a,b,c)$ belongs to all three
intervals. To prove uniqueness, we first establish the following replacement
property:
\begin{equation}\label{eq:interval-replacement}
 p,q\in I_{ab},\quad w=H(a,p,q)
 \quad\Longrightarrow\quad p,q\in I_{wb}.
\end{equation}
Because $p\in I_{ab}$, symmetry and fixed points give
$p=H(a,p,b)=P_{I_{ap}}b$.
Also $w\in I_{ap}=I_{pa}$. By (ii), $I_{pw}\subseteq I_{pa}=I_{ap}$. The point $p$ belongs to $I_{pw}$, so restricting the feasible set preserves the known minimizer: $P_{I_{pw}}b=p$. Thus $H(w,b,p)=p$, proving $p\in I_{wb}$. The argument for $q$ is identical.

Suppose now that distinct points $p,q$ lie in the three intervals in
\eqref{eq:triple-intersection}. Set $w_a=H(a,p,q)$,
$w_b=H(b,p,q)$, and $w_c=H(c,p,q)$. Applying
\eqref{eq:interval-replacement} to the pair $a,b$, and then once more with
$b,w_a$, gives $p,q\in I_{w_aw_b}$. Cyclically,
$p,q\in I_{w_bw_c}$ and $p,q\in I_{w_cw_a}$. By part (i), each of the
three distances $\norm{w_a-w_b}$, $\norm{w_b-w_c}$, and
$\norm{w_c-w_a}$ is at least $\norm{p-q}$. On the other hand,
$w_a,w_b,w_c\in I_{pq}$, which lies in a ball of diameter
$\norm{p-q}$. Hence all three pairwise distances equal that diameter.
Two points at diameter distance in a Euclidean ball are antipodal, and
three distinct points cannot be pairwise antipodal. This contradiction
proves (iii).
\end{proof}

Together with $I_{aa}=\{a\}$ and $I_{ab}=I_{ba}$, the preceding lemma
gives $I(a,a)=\{a\}, I(a,b)=I(b,a)$,
 $c\in I(a,b)\Longrightarrow I(a,c)\subseteq I(a,b)$,
and the intersection $I(a,b)\cap I(b,c)\cap I(c,a)$ is a singleton for
every triple. These are the interval axioms of a median algebra; see
Bowditch's account of Sholander's characterization
~\cite[Section~2, axioms~(I1)--(I4), pp.~283--284]{Bowditch2016}.
The unique point in the triple intersection defines the median. By
\cref{eq:triple-intersection}, this operation is
$m(a,b,c)=H(a,b,c)$ for $a,b,c\in C$.

Every point of $I_{ab}$ is fixed by $x\mapsto H(a,b,x)$, so restricting the
third report from $E$ to $C$ does not change the interval. From now on,
$I(a,b)$ denotes this median interval, which is the same set as $I_{ab}$.

We will use two standard facts about median algebras. A subset is
\emph{median-convex} if it contains $I(a,b)$ with every $a,b$ it contains,
and a \emph{median halfspace} is a nonempty proper subset whose two sides
are median-convex.
\begin{enumerate}[label=(M\arabic*)]
\item Intervals are median-convex, and two disjoint nonempty median-convex
sets can be separated by a median halfspace
~\cite[Theorem~4.8]{CHATTERJI2010882}.
\item If $\pi:C\to Q\subseteq C$ satisfies
$m(x,\pi(x),y)=\pi(x)$ for every $x\in C$ and $y\in Q$, then $\pi$ is the
gate projection onto $Q$ and preserves medians:
$\pi m(a,b,c)=m(\pi a,\pi b,\pi c)$
~\cite[Section~2]{Fioravanti2021ConvexCF}.
\end{enumerate}
The first fact will produce median-convex coordinate cuts. The second will
allow us to project the operation to a lower-dimensional section. For
completeness, we derive (M2) from halfspace separation. First, the gate identity fixes $Q$: for $q\in Q$, it gives $\pi q=m(q,\pi q,q)=q$. If $p,p'\in Q$ both satisfy the gate identity for $x$, symmetry gives
$p=m(x,p,p')=m(x,p',p)=p'$,
so the gate is unique. For a median halfspace $A$, write $\chi_A$ for its indicator. Median-convexity of both sides gives
$\chi_A(m(a,b,c))
 =\operatorname{med}(\chi_A(a),\chi_A(b),\chi_A(c))$,
because the side containing at least two inputs contains their interval and hence the median. If both sides of $A$ meet $Q$, choose $y\in Q$ on the side opposite $\pi x$. Since $\pi x=m(x,\pi x,y)\in I(x,y)$, the points $x$ and $\pi x$ lie on the same side: otherwise the side containing $x,y$ would also contain $\pi x$ by median-convexity. Thus $\chi_A\circ\pi=\chi_A$. If only one side meets $Q$, then $\chi_A\circ\pi$ is constant. In either case,
$\chi_A(\pi m(a,b,c))
 =\chi_A(m(\pi a,\pi b,\pi c))$.
Halfspace separation of points therefore proves (M2). This argument uses the halfspace-coordinate representation of median algebras~\cite[Corollary~4.11]{CHATTERJI2010882}.

Two consequences will be used repeatedly. First, every median-convex subset
of $C$ is Euclidean convex, because $I(a,b)$ is Euclidean convex and contains
the segment $[a,b]$. Second, the Euclidean closure of a median-convex set is
median-convex. To prove the latter, take $a,b\in\overline A$ and
$z\in I(a,b)$, choose $a_n,b_n\in A$ converging to $a,b$, and use
$m(a_n,b_n,z)\to m(a,b,z)=z$.

\subsection{Straightening the median in one orthonormal frame}\label{sec:straightening}
We now use the Euclidean projection identity
$m(a,b,x)=P_{I(a,b)}x$ to determine the median operation in coordinates.
The proof proceeds by induction on $\dim C$. A median halfspace provides a
median-convex hyperplane section of dimension one less; the next elementary
convex-separation lemma realizes that section as an affine hyperplane.
We write $\overline A$ for Euclidean closure.

\begin{lemma}\label{lem:convex-partition}
Let \(C\subseteq\mathbb R^r\) be closed, convex, and
full-dimensional. Suppose that \(A\subseteq C\), and that both
\(A\) and \(C\setminus A\) are convex and have nonempty interior.
Then there are a nonzero vector \(\nu\in\mathbb R^r\) and a scalar
\(t\in\mathbb R\) such that, with $P=\{x:\langle \nu,x\rangle=t\}$ and
\[
    P^-=\{x:\langle \nu,x\rangle\le t\},
    \qquad
    P^+=\{x:\langle \nu,x\rangle\ge t\},
\]
we have
\[
    \overline A=C\cap P^-,
    \qquad
    \overline{C\setminus A}=C\cap P^+.
\]
Moreover, \(P\) meets \(\operatorname{int}C\).
\end{lemma}
\begin{proof}
The separating-hyperplane theorem for the two disjoint open convex interiors gives a nonzero linear functional $\ell$ and a level $t$ with
$A\subseteq\{\ell\le t\}$ and $C\setminus A\subseteq\{\ell\ge t\}$.
These inclusions extend from the interiors to the sets because a convex set with nonempty interior is contained in the closure of its interior. Each interior lies strictly on its own side: an open ball cannot be contained in a closed halfspace and have its center on the bounding hyperplane. Therefore the segment between interior points on opposite sides meets $P=\{\ell=t\}$ inside $\inter C$.

Every point of $C$ with $\ell<t$ must belong to $A$, since it cannot belong to the complement. Choose $a_0\in\inter A$, so $\ell(a_0)<t$. For any $z\in C\cap P^-$, the points $(1-s)z+sa_0$, $0<s<1$, belong to $C$ and have $\ell<t$, hence lie in $A$. Letting $s\downarrow0$ proves $C\cap P^-\subseteq\overline A$. The other inclusion and the argument for the other side are immediate or symmetric.
\end{proof}

After applying the induction hypothesis on the hyperplane section, we must
add its normal direction without changing the tangential coordinates. The
following projection fact does exactly this: once one output coordinate uses
the full displacement allowed by firm nonexpansiveness, the perpendicular
output cannot move.

\begin{lemma}\label{lem:clamp-saturation}
Let \(T:E\to E\) be firmly nonexpansive, and fix a unit coordinate
vector \(e_j\). Let \(I\subseteq\mathbb R\) be an interval, and
consider inputs of the form $x(s)=x_\perp+s e_j$,
where $s\in I$ and \(x_\perp\perp e_j\) is fixed.

Suppose that the \(j\)-th coordinate of the output is obtained by
clamping \(s\) to a fixed nonempty closed interval \(J\): $\langle e_j,T(x(s))\rangle
    =
    q(s)
    :=
    P_Js$.
Then the component of \(T(x(s))\) orthogonal to \(e_j\) is independent
of \(s\). Equivalently, changing input coordinate \(j\) cannot change
any other output coordinate.
\end{lemma}

\begin{proof}
Write the output as $T(x(s))=q(s)e_j+r(s)$, where $r(s)\in e_j^\perp$.
We will show that \(r(s)\) is constant. Take two parameters \(s,s'\in I\). Since only input coordinate \(j\)
changes, $x(s)-x(s')=(s-s')e_j$.
The corresponding output difference is
$T(x(s))-T(x(s'))=\bigl(q(s)-q(s')\bigr)e_j+
    \bigl(r(s)-r(s')\bigr)$.
The two terms on the right are orthogonal. Therefore,
\[
\begin{aligned}
    \norm{T(x(s))-T(x(s'))}^2
    &=
    \bigl(q(s)-q(s')\bigr)^2
    +
    \norm{r(s)-r(s')}^2,\\
    \ip{x(s)-x(s')}{T(x(s))-T(x(s'))}
    &=
    (s-s')\bigl(q(s)-q(s')\bigr).
\end{aligned}
\]

Firm nonexpansiveness now gives
\begin{equation}\label{eq:clamp-saturation}
    \bigl(q(s)-q(s')\bigr)^2
    +
    \norm{r(s)-r(s')}^2
    \le
    (s-s')\bigl(q(s)-q(s')\bigr).
\end{equation}

Let \(J=[\ell,u]\), allowing either endpoint to be infinite. The clamp
has at most three regions:
\[
    q(s)=
    \begin{cases}
        \ell, & s\le \ell,\\
        s,    & \ell\le s\le u,\\
        u,    & s\ge u.
    \end{cases}
\]

First suppose that \(s\) and \(s'\) lie in a region where the clamp is
constant. Then $q(s)-q(s')=0$.
Equation \eqref{eq:clamp-saturation} reduces to $\norm{r(s)-r(s')}^2\le0$,
so \(r(s)=r(s')\).

Next suppose that \(s\) and \(s'\) lie in the middle region, where the
clamp is the identity. Then $q(s)-q(s')=s-s'$.
Equation \eqref{eq:clamp-saturation} becomes $(s-s')^2+\norm{r(s)-r(s')}^2
    \le
    (s-s')^2$. Again, $r(s)=r(s')$.
Thus \(r\) is constant on each of the at most three clamping regions.
Firm nonexpansiveness implies that \(T\) is \(1\)-Lipschitz and hence
continuous. The constants on adjacent regions therefore agree at the
clamping points \(\ell\) and \(u\). Consequently, \(r(s)\) is constant
on all of \(I\).
\end{proof}

We can now combine the abstract median structure with Euclidean projection.
The next theorem first constructs a coordinate-median formula near one
interior point and then propagates the same frame across $C$. The restriction
$H|_{C^3}$ satisfies its hypotheses by \cref{lem:interval-geometry}.

\begin{theorem}\label{thm:straightening}
Let \(C\) be a nonempty closed convex subset of a finite-dimensional
Euclidean space, and let $m:C^3\to C$
be a continuous median operation. For \(a,b\in C\), write $I(a,b)=\{m(a,b,z):z\in C\}$
for the corresponding median interval. Assume that every \(I(a,b)\)
is closed and Euclidean convex and that
\begin{equation}\label{eq:straightening-assumption}
    m(a,b,x)=P_{I(a,b)}x
    \qquad
    (a,b,x\in C).
\end{equation}
Then there is one orthonormal basis \(\mathcal B\) of the direction
space of \(\operatorname{aff}C\) such that
\[
    m(a,b,c)
    =
    \operatorname{CM}_{\mathcal B}(a,b,c)
    \qquad
    (a,b,c\in C).
\]
\end{theorem}

\begin{proof}
We work inside \(\operatorname{aff}C\), after choosing an origin
there. Thus we may regard $C\subseteq\mathbb R^r$ where $r=\dim C$,
as a full-dimensional convex set. All interiors and affine
hyperplanes below are taken in this intrinsic Euclidean space.
We first record a consequence of
\eqref{eq:straightening-assumption} that will be used throughout the
proof. Fix \(a,b,x\in C\), and put $w=m(a,b,x)$.
Because a median operation is symmetric, the map obtained by varying
the second input can be written as $u\longmapsto m(a,u,x)
    =
    m(a,x,u)
    =
    P_{I(a,x)}u$.
It is therefore a Euclidean projection and is firmly nonexpansive.
At the inputs \(b\) and \(a\), its outputs are respectively
$m(a,b,x)=w$
    and
   $m(a,a,x)=a$.
Applying \eqref{eq:firm} gives $\norm{w-a}^2
    \le
    \ip{b-a}{w-a}$.
Completing the square, we obtain
\begin{equation}\label{eq:straightening-ball}
    I(a,b)
    \subseteq
    B\left(
        \frac{a+b}{2},
        \frac{\norm{a-b}}{2}
    \right).
\end{equation}
Thus every median interval lies in the Euclidean ball having
\([a,b]\) as a diameter.

We prove the theorem by induction on \(r\). If \(r=0\), then \(C\)
consists of one point and there is nothing to prove. Suppose \(r=1\).
Since \(I(a,b)\) is convex and contains \(a\) and \(b\), it contains
the ordinary segment between them. Conversely,
\eqref{eq:straightening-ball} places it inside that segment. Hence
\(I(a,b)\) is exactly the ordinary interval between \(a\) and \(b\).
By \eqref{eq:straightening-assumption}, \(m(a,b,x)\) is projection
onto this interval, and therefore is the scalar median of \(a,b,x\).

Assume that \(r\ge2\). Choose a hyperplane \(P\) that meets the interior
of \(C\), and set $C'=C\cap P$.
We choose \(P\) so that \(C'\) is median-convex. Since \(C'\) has
dimension \(r-1\), the induction hypothesis gives an orthonormal
coordinate frame for the median operation on \(C'\).

We next extend this frame from \(C'\) to a neighborhood in \(C\). The
gate projection onto \(C'\) determines the coordinates parallel to
\(P\), and a unit vector perpendicular to \(P\) provides the remaining
coordinate. Firm nonexpansiveness then gives the coordinate-wise median
formula near an interior point of \(C'\).

From this local formula, we construct median-convex coordinate cuts of
the whole set \(C\). These cuts allow us to propagate the same coordinate
frame throughout the interior of \(C\). Finally, they determine the
median of an arbitrary triple one coordinate at a time, and continuity
extends the formula to boundary points.

The dependencies are
\[
\begin{aligned}
&C'=C\cap P
 \xrightarrow{\text{induction}}
 \text{coordinate frame on }C'\\
&\qquad\xrightarrow{\text{projection and firm nonexpansiveness}}
 \text{local formula near a point of }C'\\
&\qquad\xrightarrow{\text{median-convex cuts}}
 \text{global coordinate cuts in }C\\
&\qquad\xrightarrow{\text{movement along coordinate directions}}
 \text{the same local frame throughout }\operatorname{int}C\\
&\qquad\xrightarrow{\text{coordinate bounds and continuity}}
 m=\operatorname{CM}_{\mathcal B}\text{ on }C^3.
\end{aligned}
\]

\proofstage{stage:straightening-section}{A median-convex hyperplane section}
We begin by finding a median-convex hyperplane section to which the
induction hypothesis applies. Choose distinct points $x^-,x^+\in\inter C$.
By the median-halfspace separation property recorded in
\cref{sec:option-intervals}, there is a median halfspace $A\subseteq C$ with
$x^-\in A$ and $x^+\in C\setminus A$.
Both \(A\) and \(C\setminus A\) are Euclidean convex. Indeed, a
median-convex set contains \(I(a,b)\) whenever it contains \(a,b\),
and the Euclidean convexity of \(I(a,b)\) then gives the ordinary
segment between \(a\) and \(b\).

Both sides have nonempty interior in $\aff C$. We verify this for
\(A\); the argument for \(C\setminus A\) is identical. Suppose that
\(\inter A=\varnothing\). Since \(A\) is convex, it is contained in
a proper affine hyperplane \(L\). Because \(x^-\in\inter C\), we can
choose a vector \(v\) transverse to \(L\) and \(\varepsilon>0\)
small enough that $x^-+\varepsilon v$ and $x^--\varepsilon v$
both belong to \(C\). Neither point lies in \(L\), so neither belongs
to \(A\). Both therefore belong to \(C\setminus A\). Convexity of
\(C\setminus A\) would then place their midpoint \(x^-\) in
\(C\setminus A\), contradicting \(x^-\in A\).

By \cref{lem:convex-partition}, there is an affine hyperplane \(P\)
meeting \(\inter C\) such that the closures of \(A\) and
\(C\setminus A\) are the two closed traces of \(P\) on \(C\).
The closure of a median-convex set is median-convex, as observed at
the end of \cref{sec:option-intervals}. Hence $C':=C\cap P$
is median-convex. Since \(P\) meets \(\inter C\), the section \(C'\)
has dimension \(r-1\).

We next show that Euclidean projection onto \(C'\) is also a gate
projection in the median algebra. Fix \(x\in C\), and put $p=P_{C'}x$.
For every \(y\in C'\), median-convexity gives $I(p,y)\subseteq C'$,
and \(p\in I(p,y)\). Since \(p\) minimizes the distance from \(x\)
over \(C'\), the restriction principle
\eqref{eq:restrict-minimizer} gives $P_{I(p,y)}x=p$.
Using symmetry and \eqref{eq:straightening-assumption}, $m(x,p,y)
    =
    m(p,y,x)
    =
    P_{I(p,y)}x
    =
    p$. 
Thus \(P_{C'}\) satisfies the gate identity. By the gate-morphism
property recorded in \cref{sec:option-intervals}, it preserves the
median operation:
\begin{equation}\label{eq:gate-morphism}
    P_{C'}m(a,b,c)
    =
    m(P_{C'}a,P_{C'}b,P_{C'}c).
\end{equation}

The restriction of \(m\) to \((C')^3\) satisfies the hypotheses of
the theorem. To see that its median intervals are unchanged, fix
\(a,b\in C'\). Median-convexity gives $I(a,b)\subseteq C'$.
The interval generated by restricting the third input to \(C'\) is
therefore contained in \(I(a,b)\). Conversely, if \(z\in I(a,b)\),
then \(z\in C'\) and $m(a,b,z)=z$,
so \(z\) is also produced by the restricted operation. Thus the two
intervals coincide. Since the inputs and the feasible interval both
lie in \(P\), Euclidean projection computed inside \(P\) is the same
nearest-point problem as in the ambient space.

The induction hypothesis therefore gives an orthonormal frame $e_1,\ldots,e_{r-1}$ of the direction space of \(P\) in which the restriction of \(m\) to
\(C'\) is the coordinate-wise median.

\proofstage{stage:straightening-local}{A local orthonormal frame}
The section constructed in \cref{stage:straightening-section} determines
the tangential coordinates nearby. We extend its frame by a normal vector and determine the
remaining coordinate using firm nonexpansiveness. Choose $p_0\in P\cap\inter C$,
and let \(\nu\) be a unit vector perpendicular to \(P\). Set $e_r:=\nu$ and $\mathcal B=(e_1,\ldots,e_{r-1},e_r)$.
For the rest of this step, coordinates are taken relative to the
origin \(p_0\) and the frame \(\mathcal B\). Since \(p_0\in\inter C\), choose \(\rho>0\) such that $B(p_0,2\rho)\subseteq C$,
and let $V=\{x:\norm{x-p_0}<\rho\}$.
Let \(P_P\) denote orthogonal projection onto the affine hyperplane
\(P\). Since \(p_0\in P\) and orthogonal projection is nonexpansive, $P_P(V)
    \subseteq
    V\cap P
    \subseteq
    C'$.
For \(x\in V\), the point \(P_Px\) is the nearest point to \(x\) in
the entire hyperplane \(P\), and it belongs to \(C'\). It is therefore
also the nearest point to \(x\) in \(C'\). Hence $P_{C'}x=P_Px$ for $x\in V$.

Since $m(p_0,p_0,p_0)=p_0$,
continuity of \(m\) allows us to choose a sufficiently small open
\(\mathcal B\)-coordinate cube \(U\), centered at \(p_0\), such that
$U\subseteq V$, and $m(U^3)\subseteq V$.
For \(a,b,c\in U\), the inputs and their output all lie in \(V\).
Equation \eqref{eq:gate-morphism} therefore becomes $P_Pm(a,b,c)
    =
    m(P_Pa,P_Pb,P_Pc)$.
The three projected inputs belong to \(C'\), where the induction
hypothesis applies. Consequently, for every tangential coordinate
\(j<r\),
\[
\begin{aligned}
    \ip{e_j}{m(a,b,c)-p_0}
    =
    \operatorname{med}\bigl(
        \ip{e_j}{a-p_0},
        \ip{e_j}{b-p_0},
        \ip{e_j}{c-p_0}
    \bigr).
\end{aligned}
\]
Thus all tangential output coordinates are already determined. It
remains to identify the output coordinate in the normal direction
\(\nu\). Fix \(b,c\in U\), and vary only one tangential coordinate of the
first report. Write $a(s)=p_0+a_\perp+s e_j$, $a_\perp\perp e_j$,
where \(j<r\) and \(s\) ranges over an interval for which
\(a(s)\in U\). By symmetry and
\eqref{eq:straightening-assumption}, $a\longmapsto m(a,b,c)
    =
    m(b,c,a)
    =
    P_{I(b,c)}a$ is firmly nonexpansive.

The tangential formula already proved shows that the \(j\)-th output
coordinate is $q(s)
    =
    \operatorname{med}(s,b_j,c_j)$,
where \(b_j=\ip{e_j}{b-p_0}\) and
\(c_j=\ip{e_j}{c-p_0}\). This is the clamp of \(s\) to the interval
between \(b_j\) and \(c_j\). The hypotheses of
\cref{lem:clamp-saturation} are therefore satisfied, so the entire
output component orthogonal to \(e_j\), and in particular the normal
output coordinate, is independent of \(s\).

Repeat this argument for every tangential coordinate of each of the
three reports. Because \(U\) is a coordinate cube, the coordinates
may be changed one at a time without leaving \(U\). We conclude that
the normal output coordinate depends only on the three normal input
coordinates.

For \(a,b,c\in U\), write $\alpha=\ip{\nu}{a-p_0},
    \beta=\ip{\nu}{b-p_0}$ and
    $\gamma=\ip{\nu}{c-p_0},$
and define the corresponding points on the normal line through
\(p_0\): $\widehat a=p_0+\alpha\nu,
    \widehat b=p_0+\beta\nu$, and
    $\widehat c=p_0+\gamma\nu$.
These points remain in \(U\). Replacing the tangential coordinates
of the reports by those of \(p_0\) does not change the normal output,
so it is enough to determine $y:=m(\widehat a,\widehat b,\widehat c)$.
Using \eqref{eq:gate-morphism} and $P_P\widehat a
    =
    P_P\widehat b
    =
    P_P\widehat c
    =
    p_0$,
we get $P_Py=m(p_0,p_0,p_0)=p_0$.
Thus $y=p_0+\eta\nu$
for some \(\eta\in\mathbb R\).
The median \(y\) belongs to each of the three pairwise median
intervals. In particular, $y\in I(\widehat a,\widehat b)$.
By \eqref{eq:straightening-ball}, $\left|
        \eta-\frac{\alpha+\beta}{2}
    \right|
    \le
    \frac{|\alpha-\beta|}{2}$.
This is equivalent to saying that \(\eta\) lies between \(\alpha\)
and \(\beta\). Applying the same argument to the other two pairwise
intervals shows that \(\eta\) lies between each pair among
\(\alpha,\beta,\gamma\). The unique scalar with this property is
their median: $\eta
    =
    \operatorname{med}(\alpha,\beta,\gamma)$.

Combining the tangential and normal coordinates, we obtain
\begin{equation}\label{eq:local-median}
    m(a,b,c)
    =
    \operatorname{CM}_{\mathcal B}(a,b,c)
    \qquad
    (a,b,c\in U).
\end{equation}

\proofstage{stage:straightening-cuts}{Global coordinate cuts}
The formula from \cref{stage:straightening-local} produces coordinate cuts
that are median-convex throughout $C$. We construct them in a smaller box and
extend them using halfspace separation. Choose \(\delta>0\) such that $B(p_0,2\delta)\subseteq U$,
and choose a small full-dimensional closed \(\mathcal B\)-box \(Q\), centered at
\(p_0\), such that $Q\subseteq B(p_0,\delta)$.
If \(a,b\in Q\), then the center of the ball with diameter
\([a,b]\) belongs to \(B(p_0,\delta)\), and its radius is at most
\(\delta\). Hence the entire diameter ball lies in
\(B(p_0,2\delta)\subseteq U\).

We claim that, for every \(a,b\in Q\),
\begin{equation}\label{eq:local-interval-box}
    I(a,b)
    =
    \boxB(a,b).
\end{equation}
For the forward inclusion, let \(z\in I(a,b)\). By
\eqref{eq:straightening-ball}, \(z\in U\). The fixed-point
characterization of a median interval gives $m(a,b,z)=z$.
Together with \eqref{eq:local-median}, this says that every coordinate
of \(z\) lies between the corresponding coordinates of \(a\) and
\(b\). Thus $z\in\boxB(a,b)$.
Conversely, suppose \(z\in\boxB(a,b)\). Since \(Q\) is itself a
coordinate box and \(a,b\in Q\), $\boxB(a,b)\subseteq Q\subseteq U$.
The local formula \eqref{eq:local-median} then gives $m(a,b,z)=z$,
so \(z\in I(a,b)\). This proves
\eqref{eq:local-interval-box}. In particular, \(Q\) is
median-convex as a subset of the full median algebra \(C\).

Fix a coordinate \(j\) and a level \(t\) strictly between the two
\(j\)-faces of \(Q\). Define $A_-:=Q\cap\{x_j\le t\}$ and $A_+:=Q\cap\{x_j>t\}$.
By \eqref{eq:local-interval-box}, both sets are median-convex, and
they are disjoint. By the median-halfspace separation property
recorded in \cref{sec:option-intervals}, there is a median halfspace
\(J\subseteq C\), oriented so that $A_-\subseteq J$, and $A_+\subseteq C\setminus J$. Both \(J\) and \(C\setminus J\) are Euclidean convex, and both have
nonempty interior because they contain the corresponding open
half-boxes of \(Q\). By \cref{lem:convex-partition}, there is an
affine hyperplane \(R\) whose two closed traces on \(C\) are
\(\overline J\) and \(\overline{C\setminus J}\). We now identify \(R\). Let
$S=\inter Q\cap\{x_j=t\}$.
For every \(z\in S\) and every sufficiently small
\(\varepsilon>0\),
$z-\varepsilon e_j\in A_-\subseteq J$, and
    $z+\varepsilon e_j\in A_+\subseteq C\setminus J$.
Letting \(\varepsilon\downarrow0\) shows that \(z\) lies in both
closures, and hence in \(R\). Therefore \(R\) contains the relatively
open \((r-1)\)-dimensional set \(S\) of the coordinate hyperplane
\(\{x_j=t\}\). Two affine hyperplanes containing the same relatively
open \((r-1)\)-dimensional set must coincide. Thus $R=\{x_j=t\}$.
The orientation is determined by \(A_-\subseteq J\). Consequently,
\begin{equation}\label{eq:global-cuts}
    C\cap\{x_j\le t\}
    \quad\text{and}\quad
    C\cap\{x_j\ge t\}
    \quad\text{are median-convex.}
\end{equation}
Thus each level in the local box determines a median-convex cut of
all of $C$.

\proofstage{stage:straightening-interior}{Propagation through the interior}
The cuts from \cref{stage:straightening-cuts} allow us to carry the local
frame through the interior. We first move the center of a local coordinate
cube along a single axis.
Call \(p\in\inter C\) \emph{\(\mathcal B\)-local} if there is an
open \(\mathcal B\)-coordinate cube \(U_p\) around \(p\) such that $m(a,b,c)
    =
    \operatorname{CM}_{\mathcal B}(a,b,c)$ for all $a,b,c\in U_p$.
Equation~\eqref{eq:local-median} shows that $p_0$ is
$\mathcal B$-local.
Suppose \(p\) is \(\mathcal B\)-local and
\(p'\in\inter C\) differs from \(p\) only in coordinate \(k\).
Since \(\inter C\) is convex, the coordinate segment joining \(p\)
and \(p'\) lies in \(\inter C\).

Apply the construction of \eqref{eq:global-cuts} in a local box
around $p$. For
every \(j\ne k\), it supplies global median-convex cuts of the form $C\cap\{x_j\le t\}$, and $C\cap\{x_j\ge t\}$,
for all \(t\) in some open interval around \(p_j\). Since $p'_j=p_j$ for
    $j\ne k$,
we may choose a small open \(\mathcal B\)-coordinate cube
\(U'\subseteq\inter C\) around \(p'\) whose \(j\)-coordinate range
lies inside this interval for every \(j\ne k\).
Take arbitrary \(a,b,c\in U'\), and fix \(j\ne k\). Relabel the
three points according to their \(j\)-th coordinates so that $a_j\le b_j\le c_j$.
Choose an available cut level \(t>b_j\). The points \(a,b\) belong
to the lower cut $C\cap\{x_j\le t\}$.
Since that cut is median-convex, $I(a,b)\subseteq C\cap\{x_j\le t\}$.
The median \(m(a,b,c)\) belongs to \(I(a,b)\), and therefore $m(a,b,c)_j\le t$.
Letting \(t\downarrow b_j\) gives $m(a,b,c)_j\le b_j$.
Similarly, for an available level \(t<b_j\), the points \(b,c\)
belong to the upper cut $C\cap\{x_j\ge t\}$.
Since \(m(a,b,c)\in I(b,c)\), we get $m(a,b,c)_j\ge t$.
Letting \(t\uparrow b_j\) gives $m(a,b,c)_j\ge b_j$.
Thus $m(a,b,c)_j=b_j$.
We have shown that every output coordinate except possibly coordinate
\(k\) is the scalar median of the corresponding inputs on
\((U')^3\).

Now vary one input coordinate \(j\ne k\) of one report while holding
all other inputs fixed. The corresponding \(j\)-th output coordinate
is a scalar clamp. By \cref{lem:clamp-saturation}, the \(k\)-th output
coordinate is independent of this variation. Repeating this for every
coordinate \(j\ne k\) of every report shows that the \(k\)-th output
depends only on the three \(k\)-th input coordinates.

We may therefore replace all other coordinates of the three reports
by those of \(p'\). The three reports then lie on the \(k\)-coordinate
line through \(p'\), and all output coordinates other than \(k\)
equal the corresponding coordinates of \(p'\). The median belongs to
each pairwise interval, and \eqref{eq:straightening-ball} implies that
its \(k\)-th coordinate lies between each pair of scalar inputs. It
must therefore equal their scalar median. Hence \(p'\) is
\(\mathcal B\)-local.

It remains to propagate this local formula throughout \(\inter C\). Let
\(p,q\in\inter C\). The segment \([p,q]\) is a compact subset of the
open set \(\inter C\), so there is \(\varepsilon>0\) such that its
\(\varepsilon\)-neighborhood is contained in \(\inter C\). Subdivide
\([p,q]\) into sufficiently short segments. If \(v\) is the
displacement of one such segment, its coordinate-by-coordinate path
stays within distance at most $\norm{v}_1\le\sqrt r\,\norm{v}$
of its starting point. Choosing the subdivision sufficiently fine
keeps every such path in the \(\varepsilon\)-neighborhood of
\([p,q]\), and hence in \(\inter C\).

We have therefore connected \(p\) to \(q\) by a finite
\(\mathcal B\)-axis-parallel polygonal path contained in
\(\inter C\). Applying the preceding propagation argument along its
edges, starting with $p=p_0$, shows that every point of \(\inter C\) is
\(\mathcal B\)-local. In particular, the same frame
\(\mathcal B\) works locally throughout the interior.

\proofstage{stage:straightening-global}{The formula on all of the core}
By \cref{stage:straightening-interior}, every interior point has a
neighborhood with the same median frame. The cuts of
\cref{stage:straightening-cuts} identify the outcome for arbitrary interior
triples, and continuity extends the formula to the boundary.
Take arbitrary $a,b,c\in\inter C$,
and fix a coordinate \(j\). Relabel the points according to that
coordinate so that $a_j\le b_j\le c_j$.
Thus \(b_j\) is the scalar median. Since \(b\) is
\(\mathcal B\)-local, the construction of \eqref{eq:global-cuts} supplies global
median-convex coordinate cuts at every level sufficiently close to
\(b_j\).

For \(t>b_j\) sufficiently close to \(b_j\), the points \(a,b\)
belong to the lower cut $C\cap\{x_j\le t\}$.
The median belongs to \(I(a,b)\), and this interval lies in the lower
cut. Hence $m(a,b,c)_j\le t$.
Letting \(t\downarrow b_j\), we get $m(a,b,c)_j\le b_j$.
For \(t<b_j\) sufficiently close to \(b_j\), the points \(b,c\)
belong to the upper cut $C\cap\{x_j\ge t\}$.
Since \(m(a,b,c)\in I(b,c)\), $m(a,b,c)_j\ge t$.
Letting \(t\uparrow b_j\), we obtain $m(a,b,c)_j\ge b_j$.
Therefore $m(a,b,c)_j=b_j$.
Repeating this argument for every coordinate proves $m(a,b,c)
    =
    \operatorname{CM}_{\mathcal B}(a,b,c)$
    for $a,b,c\in\inter C$.
Finally, fix \(p_*\in\inter C\). For \(a\in C\) and \(0<s<1\), set $a_s=(1-s)a+sp_*$.
Then \(a_s\in\inter C\). The reason is that if
\(B(p_*,\rho)\subseteq C\), convexity gives $B(a_s,s\rho)
    =
    (1-s)a+sB(p_*,\rho)
    \subseteq C$.
Define \(b_s,c_s\) in the same way. The identity already proved in
the interior gives $m(a_s,b_s,c_s)
    =
    \operatorname{CM}_{\mathcal B}(a_s,b_s,c_s)$.
Letting \(s\downarrow0\), continuity of \(m\) and of the scalar median
yields $m(a,b,c)
    =
    \operatorname{CM}_{\mathcal B}(a,b,c)$
for all \(a,b,c\in C\).
\end{proof}

\subsection{Determining the boundary of the core}\label{sec:boundary}
The coordinate formula leaves the shape of $C$ open. For example, the oblique halfspace $\{x_1\le x_2\}$ is closed under
coordinate-wise median. We now use reports outside $C$ to show that every
supporting direction is a coordinate direction. It follows that $C$ is an
orthogonal box.

We need three standard facts about regular points of a closed convex set.
They follow from finite-dimensional separation and Rademacher's theorem
~\cite{Rockafellar1970,EvansGariepy2015}. Let $C$ be full-dimensional in a
finite-dimensional Euclidean space $V$. A unit vector $\nu$ is an
\emph{outward supporting normal} at $y\in\partial C$ if
$\ip{\nu}{z-y}\le0$ for all $z\in C$, and $y$ is \emph{regular} if this
normal is unique.

\begin{lemma}\label{lem:regular-boundary}
For a full-dimensional closed convex set $C$ the following hold.
\begin{enumerate}[label=(\roman*)]
\item If $y$ is regular with normal $\nu$, all unit supporting normals at boundary points sufficiently close to $y$ are arbitrarily close to $\nu$.
\item If $y$ is regular and $\ip{\nu}{d}<0$, then $y+td\in\inter C$ for all sufficiently small $t>0$.
\item $C$ is the intersection of its supporting closed halfspaces at regular boundary points.
\end{enumerate}
\end{lemma}
\begin{proof}
For (i), otherwise take $y_q\to y$ and unit normals $\nu_q$ that stay a fixed positive distance from $\nu$. A subsequence converges on the compact unit sphere to $\nu'$. Passing to the limit in $\ip{\nu_q}{z-y_q}\le0$ for each $z\in C$ makes $\nu'$ a unit supporting normal at $y$, contrary to uniqueness.

For (ii), suppose instead that there are $t_q\downarrow0$ with $y+t_qd\notin\inter C$. Separate each such point from the open convex set $\inter C$ to obtain a unit vector $w_q$ with
$\ip{w_q}{z-(y+t_qd)}\le0$ for $z\in C$.
A convergent subsequence has limit $\nu$ by regularity at $y$. Taking $z=y$ gives $\ip{w_q}{d}\ge0$, so $\ip{\nu}{d}\ge0$, a contradiction. The separating-hyperplane result used here is
finite-dimensional~\cite{Rockafellar1970}.

For (iii), the claim is immediate if $C=V$. Otherwise translate so that $0\in\inter C$, and choose $r>0$ with $B(0,r)\subseteq C$. Define the polar set and its support function by
\[
 T=\{u\in V:\ip{u}{c}\le1\text{ for all }c\in C\},\qquad
 h_T(x)=\max_{u\in T}\ip{u}{x}.
\]
The set $T$ is closed and contained in $B(0,1/r)$, hence compact. Moreover,
\begin{equation}\label{eq:polar-representation}
 C=\{x:h_T(x)\le1\}.
\end{equation}
For $x\notin C$, let $p=P_Cx$ and $w=x-p$. Then $\ip{w}{c}\le\ip{w}{p}$ for $c\in C$, and $\ip{w}{p}\ge r\norm{w}>0$. Thus $u=w/\ip{w}{p}$ belongs to $T$ and satisfies $\ip{u}{x}>1$, proving \eqref{eq:polar-representation}.

The support function $h_T$ is Lipschitz, with constant $M=\max_{u\in T}\norm{u}$. By Rademacher's theorem it is differentiable almost everywhere~\cite{EvansGariepy2015}. At a differentiability point $v$, its maximizer is unique: if $u$ maximizes at $v$, the inequalities $h_T(v+tw)\ge h_T(v)+t\ip{u}{w}$ for both signs of $t$ force $u=\nabla h_T(v)$.

Fix $x\notin C$. Choose a differentiability point $v$ so close to $x$ that $h_T(v)>1$ and, for its unique maximizer $u$, $\ip{u}{x}=h_T(v)+\ip{u}{x-v}
 \ge h_T(x)-2M\norm{x-v}>1$.
Such a choice is possible since $h_T(x)>1$ and differentiability points have full measure. Set $y=v/h_T(v)$. Then $y\in\partial C$ by homogeneity of $h_T$. For any nonzero supporting normal $w$ of $C$ at $y$, the ball $B(0,r)$ implies $\ip{w}{y}>0$. Consequently $\widehat w=w/\ip{w}{y}$ belongs to $T$, and $\ip{\widehat w}{v}=h_T(v)$.
Uniqueness of the maximizer forces $\widehat w=u$. Thus $y$ is regular, and its supporting halfspace $\{z:\ip{u}{z}\le1\}$ excludes $x$. Every exterior point is therefore excluded by a regular supporting halfspace, proving (iii).
\end{proof}

\begin{proof}[Proof of \cref{thm:core-box}]
By \cref{lem:interval-geometry}, the operation $H|_{C^3}$ satisfies the
hypotheses of \cref{thm:straightening}. Fix the resulting orthonormal frame
$\B$ in $\aff C$. For $p,q\in C$, the fixed points of
$x\mapsto H(p,q,x)$ form its range $I_{pq}$. Since the operation on $C$ is
coordinate-wise median,
\begin{equation}\label{eq:core-interval-box}
 I_{pq}=C\cap\boxB(p,q).
\end{equation}
In particular, $C$ is closed under coordinate-wise median.

Translate to the direction space $V$ of $\aff C$, so that $C$ is
full-dimensional. Points, one-dimensional closed convex sets, and $C=V$ are
already boxes. In the remaining cases we prove that every regular supporting
normal is a signed coordinate vector. By
\cref{lem:regular-boundary}(iii), the corresponding supporting halfspaces
will then determine $C$.

We first show that a regular normal has at most two nonzero coordinates.
Suppose a regular point $y$ has normal $\nu$ with at least three nonzero components, relabelled $1,2,3$. Choose $M>2$. Define directions $d^{(1)},d^{(2)},d^{(3)}$ by prescribing the products $\nu_kd_k^{(s)}$ in these coordinates to be $(-M,1,1),(1,-M,1),(1,1,-M)$,
and setting all other components of the directions to zero. Each has $\ip{\nu}{d^{(s)}}=2-M<0$, so by \cref{lem:regular-boundary}(ii), all $y+\eps d^{(s)}$ belong to $C$ for sufficiently small $\eps>0$. In each of the first three coordinates, two of the direction values equal $1/\nu_k$, so their coordinate median has value $1/\nu_k$. Hence the median displacement $q$ satisfies $\ip{\nu}{q}=3>0$. The median of the three points lies outside the supporting halfspace at $y$, contradicting closure of $C$ under coordinate median.

It remains to exclude a regular normal supported on exactly two
coordinates. Suppose $y$ has a regular normal $\nu$ supported on
$\{i,j\}$. Reverse the signs of these axes so that $\nu_i,\nu_j>0$.
We first construct a nearby regular point $z$ that differs from $y$ only in
coordinates $i,j$.

By \cref{lem:regular-boundary}(i), after shrinking a neighborhood of $y$,
every supporting unit normal there has nonzero $i$- and $j$-components.
At every regular point in this neighborhood, the support bound forces all
remaining normal components to vanish. We now express the nearby boundary as a graph. For small $\eta>0$,
\cref{lem:regular-boundary}(ii) gives $y-\eta e_i\in\inter C$, while the
supporting inequality gives $y+\eta e_i\notin C$. Shrink a rectangle of
base coordinates $x_{-i}$ around $y_{-i}$ so that its bottom slice at
height $y_i-\eta$ lies in $C$ and its top slice at height $y_i+\eta$ lies
outside the supporting halfspace at $y$. Each vertical section of $C$ is an
interval. Its upper endpoint defines a finite concave locally Lipschitz
function $\varphi(x_{-i})$, and the local boundary is
$x_i=\varphi(x_{-i})$~\cite{Rockafellar1970}.

At almost every base point, Rademacher's theorem gives a unique normal
proportional to $e_i-\sum_{k\ne i}(\partial_k\varphi)e_k$. The support
bound therefore implies $\partial_k\varphi=0$ almost everywhere for every
$k\notin\{i,j\}$. Since $\varphi$ is Lipschitz, Fubini's theorem and
absolute continuity on coordinate lines show that $\varphi$ is independent
of those coordinates in a smaller rectangle. Hence
$\varphi(x_{-i})=\psi(x_j)$ there. Choose a differentiability point of
$\psi$ distinct from and arbitrarily close to $y_j$, keeping all other
base coordinates fixed. The corresponding boundary point $z$ is regular,
$z\ne y$, $z_k=y_k$ for $k\notin\{i,j\}$, and its unit normal
$\nu_z$ tends to $\nu$ by \cref{lem:regular-boundary}(i).

Choose $\delta>0$ so small that $p=y-\delta\nu\in\inter C$, and fix
$\eps>0$. Put $\eta_0=\min\{\nu_i,\nu_j\}>0$. Choose the point $z$ above
close enough that, for $k=i,j$,
\[
 |z_k-y_k|<\frac{\eta_0}4\min\{\delta,\eps\},\qquad
 |(\nu_z)_k-\nu_k|<\frac{\eta_0}2.
\]
Define $a=y+\eps\nu$ and $b=z+\eps\nu_z$. Since $\nu$ and
$\nu_z$ are outward supporting normals, for every $c\in C$ we have
$\langle a-y,c-y\rangle\le0$ and
$\langle b-z,c-z\rangle\le0$. Equation~\eqref{eq:projection-normal}
therefore gives $P_Ca=y$ and $P_Cb=z$. In coordinates $i,j$, the displayed
bounds give $p_k<z_k<a_k$ and $p_k<y_k<b_k$. All other coordinates of
$p,y,z,a,b$ agree. Therefore clamping onto the coordinate boxes gives
$P_{\boxB(p,z)}a=z$ and $P_{\boxB(p,y)}b=y$. Because $z,y\in C$,
restriction of these known minimizers to the sets in
\eqref{eq:core-interval-box} yields $H(a,p,z)=z$ and $H(b,p,y)=y$.

Now let $D_{ap}$ be the full unilateral menu with reports $a,p$ fixed. It
is a closed convex subset of $C$, by \cref{lem:core-closure}, and it
contains $z$. Since $P_Cb=z$, \eqref{eq:restrict-minimizer} gives
$H(a,p,b)=z$; symmetry gives $H(a,b,p)=z$. Interchanging the roles, the
menu with reports $b,p$ fixed contains $y$, and $P_Ca=y$ gives
$H(b,p,a)=y$. Symmetry now gives $H(a,b,p)=y$, contradicting $y\ne z$.

We have shown that every regular normal is a signed coordinate vector.
By \cref{lem:regular-boundary}(iii), $C$ is the intersection of its regular
supporting halfspaces. Each such halfspace imposes an upper or lower bound on
one coordinate, so their intersection is a Cartesian product of closed
intervals in frame $\B$, with infinite endpoints allowed. If $C$ is
lower-dimensional in $E$, extend the frame from $V$ to $E$; the additional
factors are singletons. The formula for $H$ on $C^3$ was established by
\cref{thm:straightening}.
\end{proof}

When the majority identity holds, the whole space is a core. The box theorem
then determines the ternary rule on every profile.

\begin{corollary}\label{cor:ternary-majority}
If $H:E^3\to E$ is symmetric, continuous, and strategyproof, and $H(a,b,b)=b$ for all $a,b\in E$, then $H=\CM_{\B}$ for some orthonormal basis $\B$.
\end{corollary}
\begin{proof}
The whole space $E$ is a core. Apply \cref{thm:core-box}.
\end{proof}

Finally, apply the box theorem to the cores constructed in
\cref{prop:core-realization}. This gives the local geometry needed for the
common-frame argument in Section~\ref{sec:classification}.

\begin{corollary}\label{cor:absorber-boxes}
Every unilateral menu of the absorbing mechanism in \cref{prop:absorber} is an orthogonal box. The frame may depend on the outside profile.
\end{corollary}
\begin{proof}
Apply \cref{thm:core-box} to the ternary restriction and core in \cref{prop:core-realization}.
\end{proof}

\section{Recovering the original mechanism from its absorber}\label{sec:reconstruction}

Section~\ref{sec:absorption} associates with an $N$-report mechanism $f$
an absorbing mechanism $g$ on $N-2$ reports. Section~\ref{sec:core-geometry}
shows that every unilateral menu of $g$ is an orthogonal box. In
Section~\ref{sec:classification}, we will show that these boxes use one
common frame and that $g$ selects a fixed order statistic in each
coordinate. To complete the reduction, suppose the induction hypothesis
identifies the smaller mechanism $g$ as a fixed-rank rule. We must recover
the form of the original mechanism $f$ from this representation.

This section proves the required reconstruction. Fixing $N-3$ reports of
$f$ leaves a symmetric ternary rule $H$. The range of the corresponding
absorbing response is a core box $C$ of $H$. We first show that projecting
any outcome of $H$ onto $C$ gives the coordinate-wise median of the three
projected reports. We then choose the fixed reports so that the faces of
$C$ identify the coordinate frame. If $g$ selects rank $k_j$ in coordinate
$j$, the clipping formula then shows that $f$ selects rank $k_j+1$ in the
same coordinate and the same frame.

\subsection{Clipping through a core box}

Let $H:E^3\to E$ be symmetric, continuous, and strategyproof, and suppose
that it has a full-dimensional core
$C=\prod_{j=1}^d[\ell_j,u_j]$ in an orthonormal frame $\mathcal B$. The
endpoints may be infinite. By \cref{thm:core-box},
$H|_{C^3}=\operatorname{CM}_{\mathcal B}$. Let $\pi=P_C$.

The next proposition determines the part of an arbitrary outcome that lies
inside $C$.

\begin{proposition}\label{prop:clipping}
For every $a,b,c\in E$,
\begin{equation}\label{eq:clipping}
    \pi H(a,b,c)
    =
    \operatorname{CM}_{\mathcal B}(\pi a,\pi b,\pi c).
\end{equation}
\end{proposition}

\begin{proof}
We first compute the response when one report belongs to the core. Fix
$a\in E$ and $p\in C$, and let
$D(a,p)=\{H(a,p,x):x\in E\}$ be the menu generated by the third report.
By \cref{lem:core-closure}, $D(a,p)\subseteq C$.

For $y\in C$, the fixed-point property of a projection response gives
$y\in D(a,p)$ if and only if $H(a,p,y)=y$. Since $p,y\in C$, symmetry and
\eqref{eq:core-interval-box} give
$H(a,p,y)=H(p,y,a)=P_{\operatorname{box}_{\mathcal B}(p,y)}a$. This
projection equals the corner $y$ exactly when each coordinate $y_j$ lies
between $a_j$ and $p_j$. Hence
$D(a,p)
    =C\cap\operatorname{box}_{\mathcal B}(a,p)
    =\operatorname{box}_{\mathcal B}(\pi a,p)$.
Projecting a second report $b$ onto this box gives
\begin{equation}\label{eq:one-core-formula}
    H(a,b,p)
    =
    \operatorname{CM}_{\mathcal B}(\pi a,\pi b,p)
    \qquad(a,b\in E,\ p\in C).
\end{equation}

We now use core reports as test points for a menu generated by two arbitrary
reports. Fix $a,b\in E$, and write
$D_{ab}:=\{H(a,b,x):x\in E\}$
and
    $E_{ab}:=\operatorname{box}_{\mathcal B}(\pi a,\pi b)\subseteq C$.
For every $p\in C$, the third-report response is projection onto $D_{ab}$,
while \eqref{eq:one-core-formula} is projection onto $E_{ab}$. Therefore
\begin{equation}\label{eq:core-test-projections}
    P_{D_{ab}}p=P_{E_{ab}}p
    \qquad(p\in C).
\end{equation}

Equation~\eqref{eq:core-test-projections} bounds each coordinate of
$D_{ab}$. Fix coordinate $j$, and let $v_j$ be the upper endpoint of the
$j$th factor of $E_{ab}$. Suppose $v_j<u_j$. Choose a point $p$ on the
upper $j$-face of $E_{ab}$ and $\varepsilon>0$ such that
$p+\varepsilon e_j\in C$. Projection onto $E_{ab}$ sends
$p+\varepsilon e_j$ to $p$, so \eqref{eq:core-test-projections} gives
$P_{D_{ab}}(p+\varepsilon e_j)=p$. The property of projection then
yields $\langle \varepsilon e_j,z-p\rangle\le0$ for $z\in D_{ab}$,
so every $z\in D_{ab}$ satisfies $z_j\le v_j$. The same argument at a
lower face shows that if the lower endpoint of the $j$th factor of
$E_{ab}$ is larger than $\ell_j$, then every $z\in D_{ab}$ lies above
that endpoint.

Let $y=H(a,b,c)$, and fix coordinate $j$. Relabel the three reports as
$r,s,t$ so that
$(\pi r)_j=\alpha\le (\pi s)_j=\beta\le(\pi t)_j=\gamma$.
If $\beta<u_j$, apply the upper-face bound to the menu $D_{rs}$. Since
$y=H(r,s,t)\in D_{rs}$, we obtain $y_j\le\beta$. If
$\beta>\ell_j$, apply the lower-face bound to $D_{st}$ to obtain
$y_j\ge\beta$. Thus $y_j=\beta$ whenever $\beta$ lies in the interior of
$[\ell_j,u_j]$.

If $\beta=\ell_j$, the upper bound gives $y_j\le\ell_j$, and hence
$(\pi y)_j=\ell_j=\beta$. If $\beta=u_j$, the lower bound gives
$(\pi y)_j=u_j=\beta$. Infinite endpoints cause no additional case,
since the three projected values are finite. Therefore
$(\pi y)_j=\beta$ in every coordinate, which proves
\eqref{eq:clipping}.
\end{proof}

\subsection{Lifting the absorber's ranks}

Suppose now that the absorber $g$ has already been classified in one frame
$\mathcal B$. We use a profile with distinct coordinates to identify the
same frame for the relevant ternary core, and then apply
\cref{prop:clipping} to three consecutive ranks.

\begin{proposition}\label{prop:rank-lifting}
Let $N\ge4$, and let $f:E^N\to E$ be anonymous, continuous, and
strategyproof. Suppose that $g:E^{N-2}\to E$ is an absorber satisfying
\eqref{eq:absorber}, and that in one fixed frame $\mathcal B$, $g_j(z)=z_{(k_j),j}$ for $1\le k_j\le N-2$.
Then, in the same frame, $f_j(x)=x_{(k_j+1),j}$
for every profile $x$ and every coordinate $j$.
\end{proposition}

\begin{proof}
Put $M=N-2\ge2$. We first identify the frame of a ternary core. Fix
$M-1=N-3$ outside reports $z$ whose coordinates are pairwise distinct in
each coordinate of $\mathcal B$, and define
$H(a,b,c)=f(a,b,c,z)$. Its absorbing response is
$c\mapsto g(c,z)$.

In coordinate $j$, inserting the scalar $c_j$ among the $M-1$ outside
values and taking rank $k_j$ clamps $c_j$ to
\begin{equation}\label{eq:neighbor-interval}
    [z_{(k_j-1),j},z_{(k_j),j}],
    \qquad
    z_{(0),j}=-\infty,
    \quad
    z_{(M),j}=+\infty.
\end{equation}
The coordinates of $c$ can be chosen independently, so the range of this
response is the product of these intervals. Denote the resulting core box
by $C_z$.

Every factor in \eqref{eq:neighbor-interval} has positive length because
the outside coordinates are distinct. Since $M\ge2$, each factor has at
least one finite endpoint. Thus $C_z$ has a boundary face perpendicular
to every axis of $\mathcal B$. These face normals determine the axes of a
full-dimensional orthogonal box. Hence the frame supplied by
\cref{thm:core-box} agrees with $\mathcal B$ up to signs and permutation,
which do not change the coordinate-wise median. The clipping identity
\eqref{eq:clipping} is therefore valid in the frame $\mathcal B$.

We now recover the rank selected by $f$. Take a profile $x$ whose
coordinates are pairwise distinct in $\mathcal B$, and fix coordinate
$j$. Leave variable the three agents whose $j$-coordinates have ranks
\[
    k_j,\qquad k_j+1,\qquad k_j+2,
\]
and use the remaining $N-3$ reports as the outside profile $z$.
The $j$th factor of $C_z$ has lower endpoint $x_{(k_j-1),j}$ when
$k_j>1$, and upper endpoint $x_{(k_j+3),j}$ when $k_j<N-2$; the missing
endpoint is infinite at either extreme. The three selected values
$x_{(k_j),j}$, $x_{(k_j+1),j}$, and $x_{(k_j+2),j}$ all lie strictly
inside this interval.

Apply \cref{prop:clipping} to the ternary restriction. In coordinate $j$,
it gives
\[
    (P_{C_z}f(x))_j
    =
    \operatorname{med}
    \bigl(x_{(k_j),j},x_{(k_j+1),j},x_{(k_j+2),j}\bigr)
    =x_{(k_j+1),j}.
\]
The right-hand side lies strictly inside the $j$th factor of $C_z$.
Projection of a scalar onto an interval can equal a strictly interior point
only when the scalar already equals that point. Hence
$f_j(x)=x_{(k_j+1),j}$.

Repeating the argument for every coordinate gives the desired formula for
profiles with pairwise distinct coordinates. Such profiles are dense, and
both $f$ and each coordinate order statistic are continuous. The formula
therefore extends to every profile.
\end{proof}

When $N=3$, the absorber has one report, so unanimity makes it the identity.
This one-report mechanism does not determine a frame. The three-report base
case is supplied directly by \cref{cor:ternary-majority}. For $N\ge4$, the
finite face in every coordinate direction of the generic core $C_z$
identifies the absorber's frame and carries it to the original mechanism.

\section{Rigidity for continuous mechanisms}
\label{sec:classification}

Section~\ref{sec:core-geometry} shows that every unilateral menu of the
absorbing mechanism is an orthogonal box. The frame of that box may still
depend on the agent and on the outside reports. Section~\ref{sec:reconstruction}
shows how to recover the original mechanism once the absorbing mechanism
has been identified as a fixed-rank rule. The remaining task is therefore
to prove that a mechanism with box-shaped menus uses one orthonormal frame
and one fixed rank in each coordinate.

Recall that \(\mathcal F_{m,d}\) is the class of continuous, anonymous,
unanimous, and strategyproof mechanisms \(F:E^m\to E\). Throughout this
section, the mechanisms under consideration are also \emph{normalized},
meaning that they are translation-equivariant and positively homogeneous.

The proof is governed by the directions in which the unilateral menus are
unbounded. We first show that all menus have the same recession cone \(K\).
If \(K=\{0\}\), every menu is bounded, and the absorbing reduction lowers
the population by two. If \(K\ne\{0\}\), a ray of \(K\) is an axis of every
menu, and separating that axis lowers the dimension by one. We first prove
the common-cone statement and the one-dimensional and two-report base cases.

\subsection{A recession cone shared by all menus}

For a nonempty closed convex set \(C\), define
\[
    \operatorname{rec}C
    =\{v:c+tv\in C\text{ for every }c\in C\text{ and }t\ge0\}.
\]
Thus \(\operatorname{rec}C\) records the directions in which \(C\) extends
indefinitely. For nonempty sets \(C,D\), their Hausdorff distance is
\[
    d_H(C,D)
    =\max\left\{
        \sup_{c\in C}\operatorname{dist}(c,D),
        \sup_{d\in D}\operatorname{dist}(d,C)
      \right\},
\]
allowing the value \(+\infty\).

The next lemma uses the profile Lipschitz bound to compare menus at two
outside profiles. Finite Hausdorff distance then forces the two menus to
have the same recession cone.

\begin{lemma}\label{lem:common-cone}
Let \(F\in\mathcal F_{m,d}\) be normalized. Fix one agent, and let \(C(z)\)
be its unilateral menu when the other reports form the profile
\(z=(z_1,\ldots,z_{m-1})\). Then
\[
    d_H(C(z),C(z'))
    \le
    \sum_{j=1}^{m-1}\|z_j-z'_j\|.
\]
Moreover, every unilateral menu of \(F\), for every agent and every outside
profile, has the same recession cone
\[
    K=\{F(a,0^{m-1}):a\in E\}.
\]
\end{lemma}

\begin{proof}
For fixed outside reports \(z\), the response of the chosen agent is
\(a\mapsto P_{C(z)}a\). Hence \eqref{eq:global-lipschitz} gives
\[
    \|P_{C(z)}a-P_{C(z')}a\|
    \le
    \sum_{j=1}^{m-1}\|z_j-z'_j\|
    \qquad(a\in E).
\]
If \(a\in C(z)\), then \(P_{C(z)}a=a\), so the distance from \(a\) to
\(C(z')\) is at most the right-hand side. Taking the supremum over
\(a\in C(z)\), and then interchanging \(z\) and \(z'\), proves the
Hausdorff bound.

We next show that nonempty closed convex sets at finite Hausdorff distance
have the same recession cone. Suppose \(d_H(C,D)\le L\), take
\(v\in\operatorname{rec}C\), and fix \(c_0\in C\) and \(d_0\in D\).
For every \(T>0\), choose \(d_T\in D\) with
\(\|d_T-(c_0+Tv)\|\le L+T^{-1}\). Fix \(t>0\) and take \(T>t\).
Convexity gives
\[
    \left(1-\frac{t}{T}\right)d_0+\frac{t}{T}d_T\in D.
\]
As \(T\to\infty\), these points converge to \(d_0+tv\). Since \(D\) is
closed, \(d_0+tv\in D\). The choices of \(d_0\) and \(t\) were arbitrary,
so \(v\in\operatorname{rec}D\). The reverse inclusion follows by symmetry.

It follows that all menus of the chosen agent have the same recession cone.
At the zero outside profile, the menu is
\(C(0)=\{F(a,0^{m-1}):a\in E\}\). By \cref{lem:projection}, \(C(0)\) is
closed and convex. Unanimity gives \(0\in C(0)\), and positive homogeneity
makes \(C(0)\) a cone. Therefore
\(\operatorname{rec}C(0)=C(0)\), which gives the displayed formula for
\(K\). Anonymity gives the same cone for every agent.
\end{proof}

\subsection{The one-dimensional and two-report base cases}

When the dimension is one, every unilateral menu is an interval and every
response is nondecreasing. Translation equivariance and homogeneity then
reduce the rule to one fixed order statistic. This is the normalized form
of the classical generalized-median description
\cite{moulin1980strategy,border1983straightforward}.

\begin{lemma}\label{lem:scalar}
Let \(s\in\mathcal F_{m,1}\) be normalized. There is a fixed
\(k\in\{1,\ldots,m\}\) such that
\[
    s(t)=t_{(k)}
    \qquad(t\in\mathbb R^m).
\]
\end{lemma}

\begin{proof}
Each unilateral response is projection onto a closed interval, so \(s\) is
nondecreasing in every report. For \(1\le r\le m-1\), let
\[
    K_r=\{s(a^r,0^{m-r}):a\in\mathbb R\}.
\]
By \cref{lem:coalition} and positive homogeneity, \(K_r\) is a closed
convex cone containing zero. Hence it is one of
\(\{0\}\), \([0,\infty)\), \(( -\infty,0]\), or \(\mathbb R\).

The common-report response of the first \(r\) agents is projection onto
\(K_r\). Translation equivariance gives
\[
    s(0^r,1^{m-r})
    =1+s((-1)^r,0^{m-r})
    =1+P_{K_r}(-1)
    \in\{0,1\}.
\]
For \(r=0\) and \(r=m\), unanimity gives the values \(1\) and \(0\),
respectively. As \(r\) increases, one report changes from \(1\) to \(0\),
so monotonicity makes these values nonincreasing. Let \(k\) be the first
index \(r\) for which the value is zero. For \(a<b\), translation
equivariance and positive homogeneity give
\[
    s(a^r,b^{m-r})
    =
    \begin{cases}
      b,&r<k,\\
      a,&r\ge k.
    \end{cases}
\]

Now sort an arbitrary profile as
\(t_{(1)}\le\cdots\le t_{(m)}\). Choose
\(A<t_{(1)}\) and \(B>t_{(m)}\). By anonymity and coordinatewise
monotonicity,
\[
    s(t)
    \le s(t_{(k)}^k,B^{m-k})
    =t_{(k)},
\]
while
\[
    s(t)
    \ge s(A^{k-1},t_{(k)}^{m-k+1})
    =t_{(k)}.
\]
The same formulas cover \(k=1\) and \(k=m\) with an empty block. Hence
\(s(t)=t_{(k)}\).
\end{proof}

The population induction ends with one or two reports. One report gives the
identity by unanimity. With two reports, translation equivariance determines
the mechanism from the menu against the zero report. Anonymity then forces
each factor of that box menu to be a halfline.

\begin{lemma}\label{lem:even-binary-box}
Let \(F\in\mathcal F_{2,d}\) be normalized, and suppose every unilateral
menu of \(F\) is an orthogonal box. Then there is an orthonormal basis
\(\mathcal B=(e_1,\ldots,e_d)\) such that, for each coordinate \(j\), the
\(j\)-th output coordinate is either
\(\min\{a_j,b_j\}\) for every \(a,b\in E\), or
\(\max\{a_j,b_j\}\) for every \(a,b\in E\). The choice is fixed in each
coordinate.
\end{lemma}

\begin{proof}
Fix the second report at zero and put
\(K=\{F(a,0):a\in E\}\). By \cref{lem:projection},
\(F(a,0)=P_Ka\). Translation equivariance gives
\begin{equation}\label{eq:even-binary-translation}
    F(a,b)=b+F(a-b,0)=b+P_K(a-b).
\end{equation}
Thus the menu \(K\) determines the entire mechanism.

Unanimity gives \(0\in K\), and positive homogeneity makes \(K\) a cone.
By assumption, \(K\) is also an orthogonal box. Choose a box frame
\(\mathcal B\) and write
\[
    K=\left\{\sum_{j=1}^d t_je_j:t_j\in J_j\right\}.
\]
Each \(J_j\) is a nonempty closed interval containing zero and invariant
under multiplication by positive scalars. Therefore
\[
    J_j\in\bigl\{\{0\},\mathbb R,[0,\infty),(-\infty,0]\bigr\}.
\]

Anonymity gives \(F(a,0)=F(0,a)\). Using
\eqref{eq:even-binary-translation}, we obtain
\begin{equation}\label{eq:even-binary-symmetry}
    P_Ka=a+P_K(-a).
\end{equation}
Take \(a=te_j\) with \(t>0\). If \(J_j=\{0\}\), the \(j\)-th
coordinates of the two sides of \eqref{eq:even-binary-symmetry} are \(0\)
and \(t\). If \(J_j=\mathbb R\), they are \(t\) and \(0\). Hence neither
case is possible, and every factor is a halfline.

If \(J_j=[0,\infty)\), the \(j\)-th coordinate of
\eqref{eq:even-binary-translation} is
\(b_j+\max\{a_j-b_j,0\}=\max\{a_j,b_j\}\). If
\(J_j=(-\infty,0]\), it is \(\min\{a_j,b_j\}\). The same frame and the
same choice in each coordinate apply to every profile because they come
from the fixed menu \(K\).
\end{proof}

\subsection{The common cone contains no line}

In the unbounded branch of the induction, we will choose a ray of the common
cone \(K\) and use it as a coordinate direction for every menu. We first
show that \(K\) contains no line. The argument uses coalition menus at the
zero profile.

For \(1\le r\le m-1\), define
\[
    K_r=\{F(a^r,0^{m-r}):a\in E\},
    \qquad
    K_r^*=\{u:\langle u,v\rangle\ge0\text{ for every }v\in K_r\}.
\]

\begin{lemma}\label{lem:coalition-cones}
Let \(F\in\mathcal F_{m,d}\) be normalized, with \(m\ge2\). Then
\(K_r\subseteq K_{r+1}\) whenever both cones are defined, and
\[
    K_{m-r}=K_r^*.
\]
In particular, the common unilateral recession cone \(K=K_1\) satisfies
\(K\cap(-K)=\{0\}\).
\end{lemma}

\begin{proof}
By \cref{lem:coalition} and positive homogeneity, each \(K_r\) is a closed
convex cone. If \(y\in K_r\), then
\(F(y^r,0^{m-r})=y\). Replacing one additional zero report by the current
outcome \(y\), using \cref{cor:replacement}, gives
\(F(y^{r+1},0^{m-r-1})=y\). Hence
\(K_r\subseteq K_{r+1}\).

For every \(a\in E\), translation equivariance and anonymity give
\begin{equation}\label{eq:coalition-duality-response}
    P_{K_r}a
    =F(a^r,0^{m-r})
    =a+F(0^r,(-a)^{m-r})
    =a+P_{K_{m-r}}(-a).
\end{equation}
We compare this identity with Moreau's decomposition for a closed convex
cone. With the positive-dual convention above, the decomposition is
\begin{equation}\label{eq:moreau}
    P_Ka=a+P_{K^*}(-a).
\end{equation}
To verify the sign convention, put \(p=P_Ka\). Substituting \(z=0\) and
\(z=2p\) into \eqref{eq:projection-normal} gives
\(\langle a-p,p\rangle=0\). Substituting an arbitrary \(z\in K\) then
gives \(\langle p-a,z\rangle\ge0\), so \(p-a\in K^*\). For \(u\in K^*\),
\[
    \langle -a-(p-a),u-(p-a)\rangle
    =-\langle p,u\rangle
    \le0.
\]
Applying \eqref{eq:projection-normal} to \(K^*\) shows that
\(P_{K^*}(-a)=p-a\), proving \eqref{eq:moreau}; see also the cone-projection
discussion in~\cite[Section~5.4]{BordesLaffondLeBreton2011}.

Comparing \eqref{eq:coalition-duality-response} with \eqref{eq:moreau}, and
varying \(a\), shows that the projections onto \(K_{m-r}\) and \(K_r^*\)
agree. Their ranges are therefore equal:
\(K_{m-r}=K_r^*\).

For \(r=1\), the inclusion already proved gives
\(K=K_1\subseteq K_{m-1}=K_1^*\). If \(v,-v\in K\), then
\(v\in K^*\) and \(\langle v,-v\rangle\ge0\). Hence \(v=0\), so
\(K\) contains no line.
\end{proof}

\subsection{Induction on population and dimension}

We now have the ingredients for the classification. The common cone \(K\)
determines which parameter decreases. If \(K=\{0\}\), the menus are bounded
and the absorber has two fewer reports. If \(K\ne\{0\}\), one ray of \(K\)
is a coordinate direction of every menu, and splitting off that direction
reduces the dimension.

A ray \(\mathbb R_{\ge0}e\subseteq K\) is called \emph{extreme} if, whenever
a nonzero point on the ray is written as \(u+v\) with \(u,v\in K\), both
\(u\) and \(v\) lie on the same ray. For an orthogonal-box cone that contains
no line, the extreme rays are exactly its nonzero coordinate halflines.

\begin{theorem}\label{thm:box-characterization}
Let \(m\ge1\), and let \(F\in\mathcal F_{m,d}\) be normalized. Suppose every
unilateral menu of \(F\) is an orthogonal box, where the frame may depend on
the menu. Then there are one orthonormal basis \(\mathcal B\) and fixed
ranks \(k_j\in\{1,\ldots,m\}\) such that
\[
    F=Q_{\mathcal B,k}.
\]
\end{theorem}

\begin{proof}
We use induction on the population \(m\). For each fixed population, we use
induction on the dimension \(d\). Thus, when proving the case \((m,d)\),
the theorem is available for every smaller population in every dimension,
and for population \(m\) in every smaller dimension.

If \(m=1\), unanimity makes \(F\) the identity. If \(m=2\), apply
\cref{lem:even-binary-box}. For any \(m\), the case \(d=1\) follows from
\cref{lem:scalar}. We may therefore assume \(m\ge3\) and \(d\ge2\). Let
\(K\) be the common recession cone from \cref{lem:common-cone}.

\paragraph{The case \(K=\{0\}\).}
We first show that every unilateral menu is bounded. Fix an agent, outside
reports \(z=(z_1,\ldots,z_{m-1})\), and an attainable outcome \(y\).
By \cref{cor:replacement}, \(F(y,z)=y\). The menu against the zero outside
profile is \(K=\{0\}\), so \(F(y,0^{m-1})=0\). Applying
\eqref{eq:global-lipschitz} to these two profiles gives
\[
    \|y\|\le\sum_{j=1}^{m-1}\|z_j\|.
\]
The menus are closed by \cref{lem:projection}, hence compact.

By \cref{prop:absorber}, \(F\) has a normalized absorber
\(G\in\mathcal F_{m-2,d}\). By \cref{cor:absorber-boxes}, every unilateral
menu of \(G\) is an orthogonal box. If \(m=3\), unanimity gives \(G(c)=c\),
and the absorbing identity becomes \(F(a,b,b)=b\). Hence
\cref{cor:ternary-majority} gives a coordinate-wise median representation
for \(F\).

Suppose \(m\ge4\). The population induction applies to \(G\), including the
two-report base case when \(m=4\). Thus \(G=Q_{\mathcal B,k}\) in one frame.
By \cref{prop:rank-lifting}, \(F\) uses the same frame and increases every
rank by one. This proves the theorem when \(K=\{0\}\).

\paragraph{The case \(K\ne\{0\}\).}
By \cref{lem:coalition-cones}, \(K\) contains no line. Since \(K\) is the
unilateral menu at the zero outside profile, it is also an orthogonal box
and a cone. In a box representation of such a cone, each interval factor is
\(\{0\}\) or a coordinate halfline; a full-line factor would give a line in
\(K\). Choose an extreme ray \(\mathbb R_{\ge0}e\subseteq K\), with
\(\|e\|=1\).

Every unilateral menu has recession cone \(K\). In any orthogonal-box
representation of such a menu, its unbounded interval factors give the
extreme rays of its recession cone. Therefore \(e\) is a coordinate axis of
every unilateral menu. Each menu can consequently be written as
\[
    C=Je+D,
\]
where \(J\subseteq\mathbb R\) is an interval and
\(D\subseteq e^\perp\) is an orthogonal box. For \(t\in\mathbb R\) and
\(z\in e^\perp\), projection onto this menu splits as
\begin{equation}\label{eq:product-projection}
    P_C(te+z)=P_J(t)e+P_Dz.
\end{equation}

We next show that the mechanism itself separates into a scalar rule along
\(e\) and a mechanism on \(e^\perp\). Fix one agent and the reports of all
other agents. If only the \(e^\perp\) component of that agent's report
changes, \eqref{eq:product-projection} shows that the \(e\)-component of the
outcome is unchanged. Changing the agents' perpendicular components one at
a time to zero shows that the scalar output depends only on
\((\langle e,x_1\rangle,\ldots,\langle e,x_m\rangle)\). The same argument,
with scalar and perpendicular components interchanged, shows that the
perpendicular output depends only on
\((P_{e^\perp}x_1,\ldots,P_{e^\perp}x_m)\).

Define
\[
    s(t_1,\ldots,t_m)
    =\langle e,F(t_1e,\ldots,t_me)\rangle
\]
and
\[
    F_\perp(z_1,\ldots,z_m)
    =P_{e^\perp}F(z_1,\ldots,z_m).
\]
Then
\begin{equation}\label{eq:global-splitting}
    F(x_1,\ldots,x_m)
    =s(\langle e,x_1\rangle,\ldots,\langle e,x_m\rangle)e
     +F_\perp(P_{e^\perp}x_1,\ldots,P_{e^\perp}x_m).
\end{equation}

Unanimity gives \(s(0)=0\) and \(F_\perp(0)=0\). Continuity, anonymity,
unanimity, translation equivariance, and positive homogeneity pass to both
factors. To verify strategyproofness of \(s\), restrict all true points and
reports in the incentive inequality for \(F\) to the line \(\mathbb Re\); the
perpendicular component of the outcome is then \(F_\perp(0)=0\). Restricting
all true points and reports to \(e^\perp\) proves strategyproofness of
\(F_\perp\).

The unilateral menus of \(F_\perp\) are the perpendicular factors \(D\) in
\eqref{eq:product-projection}. Indeed, fix all scalar report coordinates at
zero and vary one report through \(e^\perp\). The scalar output remains
\(s(0)=0\), while the perpendicular response is \(z\mapsto P_Dz\), whose
range is \(D\). Thus every unilateral menu of \(F_\perp\) is an orthogonal
box.

By \cref{lem:scalar}, \(s\) is a fixed scalar order statistic. The dimension
induction applies to \(F_\perp\) on the \((d-1)\)-dimensional space
\(e^\perp\), so \(F_\perp\) is a fixed-frame coordinate-wise order-statistic
rule. Adjoining \(e\) to that perpendicular frame gives one frame and fixed
ranks for \(F\).
\end{proof}

\subsection{From bounded individual influence to fixed ranks}

The preceding theorem starts from box-shaped menus. We now apply it to a
mechanism whose unilateral menus are only assumed to be bounded. The
absorbing reduction gives a smaller mechanism; Section~\ref{sec:core-geometry}
shows that its menus are boxes; the theorem above aligns their frames and
fixes their ranks; and Section~\ref{sec:reconstruction} lifts those ranks
back to the original mechanism.

\begin{theorem}\label{thm:rigidity}
Let \(n\ge3\), \(d\ge1\), and let
\(f\in\mathcal F_{n,d}\) be translation-equivariant. If every unilateral
menu of \(f\) is bounded, then there are one orthonormal basis \(\mathcal B\)
and fixed ranks \(k_j\in\{2,\ldots,n-1\}\) such that
\[
    f=Q_{\mathcal B,k}.
\]
\end{theorem}

\begin{proof}
By \cref{cor:automatic-homogeneity}, translation equivariance implies
positive homogeneity, so \(f\) is normalized. Its unilateral menus are
closed by \cref{lem:projection} and bounded by assumption, hence compact.
By \cref{prop:absorber}, there is a normalized mechanism
\(g\in\mathcal F_{n-2,d}\) satisfying
\[
    f(a,g(z),z)=g(z).
\]
By \cref{cor:absorber-boxes}, every unilateral menu of \(g\) is an
orthogonal box.

If \(n=3\), then \(g\) has one report and unanimity gives \(g(c)=c\). The
absorbing identity becomes \(f(a,b,b)=b\). Hence
\cref{cor:ternary-majority} identifies \(f\) as a coordinate-wise median,
with rank \(2\) in every coordinate.

If \(n\ge4\), \cref{thm:box-characterization} gives
\(g=Q_{\mathcal B,k}\) for one frame and ranks
\(k_j\in\{1,\ldots,n-2\}\). By \cref{prop:rank-lifting}, the original
mechanism uses the same frame and rank \(k_j+1\) in coordinate \(j\). Hence
all ranks lie in \(\{2,\ldots,n-1\}\).
\end{proof}

\begin{remark}\label{rem:converse}
For a fixed-frame order-statistic rule, the response to one report in
coordinate \(j\) clamps that report to
\([z_{(k_j-1),j},z_{(k_j),j}]\), where \(z\) is the outside profile and a
missing extreme endpoint is interpreted as \(-\infty\) or \(+\infty\).
Thus every unilateral response is projection onto an orthogonal box. The
rule is strategyproof, anonymous, continuous, unanimous, and normalized.
Its unilateral menus are bounded exactly when
\(2\le k_j\le n-1\) for every \(j\).

Combining this observation with
\cref{cor:automatic-homogeneity,rem:near-unanimity} gives an equivalent
agreement formulation: a continuous, anonymous, translation-equivariant,
strategyproof mechanism with \(n\ge3\) and
\(f(a,b^{n-1})=b\) is a fixed-frame interior-rank rule. The agreement
identity implies unanimity.
\end{remark}

The next section proves continuity from the remaining hypotheses of the
final characterization.

\section{Automatic continuity and the final rigidity theorem}
\label{sec:automatic-regularity}
\label{ext:sec:regularity}

Section~\ref{sec:classification} proves the rigidity theorem for
continuous mechanisms. We now derive continuity from the remaining
assumptions. The relevant geometric object is the full range $\mathcal R(f):=\{f(x):x\in E^n\}$,
the set of all outcomes selected by the mechanism.

Peters, van der Stel, and Storcken~\cite{peters1993range} established
two results that form the starting point of this section. For a
one-report strategyproof map, continuity is equivalent to convexity of
its range in every finite dimension, and the map is Euclidean projection
onto that range. For mechanisms with several reports, they proved the
equivalence between continuity and convexity of the full range in the
Euclidean plane~\cite[Theorem~5]{peters1993range}.

We extend the multi-agent equivalence to every finite dimension. The
implication from continuity to convex range follows directly from the
one-report result. The converse is the new step: convexity of the full
range forces every unilateral menu to be convex and every response to
be continuous.

\subsection{Range convexity and continuity}
\label{sec:range-equivalence}

\begin{theorem}\label{intro:thm:range-continuity}
Let $f:E^n\to E$ be strategyproof. Then $f$ is continuous if and only if
its full range $\mathcal R(f)$ is convex.
\end{theorem}

\begin{proof}
Suppose first that $f$ is continuous. By \cref{lem:raw-menus}, the
common-report map $T(a)=f(a,\ldots,a)$ is strategyproof and has range
$\mathcal R(f)$. It is continuous, so \cref{lem:projection} makes its
range convex.

Conversely, suppose that $\mathcal R(f)$ is convex. The all-dimensional
regularity theorem \cref{ext:thm:regularity}, proved in
\cref{ext:app:regularity}, shows that every unilateral response is
projection onto a closed convex menu. Changing the reports one at a time
and using nonexpansiveness gives
\begin{equation}\label{ext:eq:regularity-lipschitz}
    \norm{f(x)-f(x')}
    \le
    \sum_{i=1}^n \norm{x_i-x_i'}.
\end{equation}
Hence $f$ is continuous.
\end{proof}

\paragraph{Proof idea for the converse.}
The new part begins with a two-report rule. Strategyproofness makes each
outcome a nearest attainable point, but the attainable set may initially be
nonconvex and the selected nearest point may jump. The incentive inequalities
define one convex scalar function for each report. Comparing their slopes as
the two reports move toward the diagonal bounds both response maps by the
derivative of projection onto the convex full range. This rules out jumps
and makes every binary menu convex. For a general population, several agents
are grouped at a common report. Starting from the grand coalition and
removing one agent at a time transfers the binary conclusion to every
unilateral menu. Appendix~\ref{ext:app:regularity} gives the complete
argument.
\subsection{Unanimity and the final rigidity theorem}
\label{sec:regularity-consequences}

A unanimous mechanism selects every point of \(E\), so its full range is
\(E\). The range-convexity theorem therefore supplies continuity
automatically.

\begin{corollary}\label{ext:cor:finite-regularity}
Every unanimous deterministic strategyproof mechanism is
continuous. For \(n\ge2\), every
\(f\in\ExtS_{n,d}\) with \(\AR(f)<\infty\) is unanimous and satisfies
\eqref{ext:eq:regularity-lipschitz}. In particular, every finite-ratio
member of \(\ExtA_{n,d}\) belongs to \(\F_{n,d}\).
\end{corollary}

\begin{proof}
If \(f\) is unanimous, then $\mathcal R(f)=E$,
which is convex. Apply \cref{intro:thm:range-continuity}.

Under the approximation-ratio convention in \cref{sec:model}, a profile
with all reports equal has optimal social cost zero. A finite-ratio
mechanism must also have cost zero at that profile, and hence is
unanimous. The first part then gives continuity and
\eqref{ext:eq:regularity-lipschitz}.
\end{proof}

We are now ready to state and prove the final characterization.

\begin{theorem}\label{intro:thm:rigidity}
Let $n\ge3$ and $d\ge1$. A mechanism $f:E^n\to E$ is anonymous,
unanimous, translation-equivariant, strategyproof, and has bounded
individual influence if and only if there are an orthonormal basis
$\mathcal B$ and fixed ranks $k_j\in\{2,\ldots,n-1\}$ such that
$f=Q_{\mathcal B,k}$.
\end{theorem}

\begin{proof}
Let \(f\) satisfy the hypotheses of
\cref{intro:thm:rigidity}. Unanimity and
\cref{ext:cor:finite-regularity} imply that \(f\) is continuous.
Therefore \(f\in\F_{n,d}\), and the continuous rigidity theorem
\cref{thm:rigidity} gives one orthonormal frame and fixed ranks $k_j\in\{2,\ldots,n-1\}$
such that $f=Q_{\mathcal B,k}$.

Conversely, let \(f=Q_{\mathcal B,k}\) with
\(2\le k_j\le n-1\) for every coordinate \(j\). Fixing the other reports,
the response to one report is coordinate-wise clamping to a bounded
orthogonal box. It is therefore Euclidean projection onto that box and
is strategyproof. Anonymity, unanimity, and translation equivariance
follow directly from the coordinate order statistics. Thus every such
rule satisfies the hypotheses of \cref{intro:thm:rigidity}.
\end{proof}

\section{Median optimality}
\label{sec:approximation}

Fix a population size \(n\ge3\) and a dimension \(d\). Recall that
\(R_{n,d}\) is the approximation ratio of the coordinate-wise median. We
prove that no anonymous strategyproof mechanism has a smaller ratio.

The proof is by contradiction. Suppose that an anonymous mechanism \(f\)
satisfies \(\AR(f)<R_{n,d}\). We first replace \(f\) by a normalized
mechanism \(\bar f\) with no larger ratio. The approximation guarantee then
bounds every unilateral menu of \(\bar f\), so the rigidity theorem writes
\(\bar f\) as a fixed-rank rule \(Q_{\B,k}\). Finally, coordinate
reflections show that every fixed-rank rule has approximation ratio at least
\(R_{n,d}\).

\subsection{From an anonymous mechanism to a fixed-rank rule}

The rigidity theorem applies to translation-equivariant mechanisms. The
following normalization theorem produces such a mechanism without increasing
the approximation ratio.

\begin{theorem}\label{thm:normalization}
For every \(n\ge2\), \(d\ge1\), and \(f\in\F_{n,d}\), there is a normalized
mechanism \(\bar f\in\F_{n,d}\) such that
\(\AR(\bar f)\le\AR(f)\).
\end{theorem}

\paragraph{Proof idea.}
For \(v\in E\), write $x+v:=(x_1+v,\ldots,x_n+v)$,
and define $f^v(x):=f(x+v)-v$.
Thus \(f^v\) moves every report by the same vector \(v\), applies \(f\),
and then moves the outcome back by \(v\). A common shift preserves all
distances, so \(f^v\) is strategyproof, anonymous, and unanimous, and
\(\AR(f^v)=\AR(f)\).

We choose vectors \(v\) along which the mechanisms \(f^v\) converge. Fix
a profile \(x\), and define the displacement
$F_x(v):=v-f(x+v)=-f^v(x)$. Equation~\eqref{eq:common-translation-lipschitz}
controls the output when the common shift changes. The calculation in
\eqref{eq:monotone-displacement} shows that
\[
    \ip{F_x(v)-F_x(w)}{v-w}\ge0,
\]
so \(F_x\) is monotone.

The displacement is also bounded independently of \(v\). Indeed,
unanimity gives \(f(v,\ldots,v)=v\), and the profile Lipschitz bound gives
$\norm{F_x(v)}
    =
    \norm{f(x+v)-f(v,\ldots,v)}
    \le
    \sum_{i=1}^n\norm{x_i}$.
For each \(x\), let \(K_x\) be the closed convex hull of the bounded set
\(\{F_x(v):v\in E\}\). An exposed-point argument chooses one unit vector
\(u\) such that, for every profile in a fixed countable dense set, the
linear function \(z\mapsto\ip{u}{z}\) has a unique maximizer on \(K_x\).
The monotonicity of \(F_x\) then implies that \(F_x(tu)\) converges as
\(t\to\infty\). It also implies that, for every fixed \(w\in E\),
\(F_x(tu+w)\) has the same limit.

The mechanisms \(f^{tu}\) have a common Lipschitz bound as functions of
the profile. Their convergence on the dense set therefore extends to
locally uniform convergence on \(E^n\). Denote the limit by \(\bar f\).

The independence of the limit from the fixed offset \(w\) gives
translation equivariance. For every profile \(x\) and vector \(w\),
\[
\begin{aligned}
    \bar f(x+w)
    &=\lim_{t\to\infty}\bigl(f(x+w+tu)-tu\bigr)\\
    &=w+\lim_{t\to\infty}
        \bigl(f(x+(tu+w))-(tu+w)\bigr)=w+\bar f(x).
\end{aligned}
\]
Positive homogeneity follows from \cref{cor:automatic-homogeneity}.
Finally, strategyproofness, anonymity, unanimity, and the approximation
inequality pass to the locally uniform limit. Hence
\(\AR(\bar f)\le\AR(f)\). The complete construction of the direction
\(u\) and the convergence argument are given in \cref{app:normalization}.

We next show that a ratio below \(n-1\) bounds the outcomes that one agent can
induce while the other reports are fixed.

\begin{lemma}\label{lem:approx-menu}
Let \(n\ge2\), and let \(f\in\ExtS_{n,d}\) have approximation ratio
\(R<n-1\). Fix the reports \(z_1,\ldots,z_{n-1}\) of all agents except one.
If \(y\) belongs to the remaining agent's unilateral menu, then, for every
\(b\in E\),
\begin{equation}\label{eq:approx-menu-bound}
    \norm{y-b}
    \le
    \frac{R+1}{n-1-R}
    \sum_{j=1}^{n-1}\norm{z_j-b}.
\end{equation}
Consequently, every unilateral menu of \(f\) is compact.
\end{lemma}

\begin{proof}
Since \(y\) is attainable against \(z=(z_1,\ldots,z_{n-1})\), outcome
replacement gives \(f(y,z)=y\). For any \(b\in E\), the triangle inequality
and comparison with the facility location \(b\) give
\[
    \SC((y,z),y)
    \ge
    (n-1)\norm{y-b}
    -\sum_{j=1}^{n-1}\norm{z_j-b},\quad
    \OPT(y,z)
    \le
    \norm{y-b}
    +\sum_{j=1}^{n-1}\norm{z_j-b}.
\]
Combining these inequalities with
\(\SC((y,z),y)\le R\OPT(y,z)\) and rearranging proves
\eqref{eq:approx-menu-bound}. For fixed outside reports, the right-hand side
is finite, so the unilateral menu is bounded. It is closed by
\cref{lem:raw-menus}, and hence compact.
\end{proof}

Normalization, the menu bound, and rigidity now reduce every anonymous
mechanism with ratio below \(n-1\) to a fixed-rank rule.

\begin{corollary}\label{ext:cor:representative}
Let \(n\ge3\), and let \(f\in\ExtA_{n,d}\) satisfy \(\AR(f)<n-1\). Then there
are an orthonormal frame \(\B\) and fixed ranks
\(k_j\in\{2,\ldots,n-1\}\) such that $\AR(Q_{\B,k})\le\AR(f)$.
The rule \(Q_{\B,k}\) can be chosen as the normalized representative from
\cref{thm:normalization}.
\end{corollary}

\begin{proof}
By \cref{ext:cor:finite-regularity}, \(f\in\F_{n,d}\). Apply
\cref{thm:normalization} to obtain a normalized mechanism \(\bar f\) with
\(\AR(\bar f)\le\AR(f)<n-1\). Its unilateral menus are compact by
\cref{lem:approx-menu}. The rigidity theorem \cref{thm:rigidity} therefore
gives \(\bar f=Q_{\B,k}\) for one orthonormal frame \(\B\) and fixed ranks
\(k_j\in\{2,\ldots,n-1\}\).
\end{proof}

\subsection{Comparing fixed ranks with the median}

It remains to compare a fixed-rank rule \(Q_{\B,k}\) with the coordinate-wise
median. Reflecting coordinate \(j\) changes rank \(k_j\) to the complementary
rank \(n+1-k_j\). Across all coordinate reflections, the resulting outputs
are the corners of a coordinate box containing the median outcome.

\begin{lemma}\label{lem:reflection}
For every \(n\ge2\) and every fixed-frame order-statistic rule
\(Q_{\B,k}\), $\AR(Q_{\B,k})\ge R_{n,d}$.
\end{lemma}

\begin{proof}
For \(\varepsilon=(\varepsilon_1,\ldots,\varepsilon_d)\in\{-1,1\}^d\), let
\(D_\varepsilon\) be the orthogonal map satisfying
\(D_\varepsilon e_j=\varepsilon_j e_j\), and define
\[
    Q^\varepsilon(x_1,\ldots,x_n)
    :=
    D_\varepsilon
    Q_{\B,k}(D_\varepsilon x_1,\ldots,D_\varepsilon x_n).
\]
Since \(D_\varepsilon\) preserves Euclidean distances,
\(\AR(Q^\varepsilon)=\AR(Q_{\B,k})\).

In coordinate \(j\), the reflected rule selects rank \(k_j\) or
\(n+1-k_j\). The median rank fixed for that coordinate lies between these
two ranks. Therefore, for every profile \(x\),
$\CM_{\B}(x)
    \in
    \operatorname{conv}
    \{Q^\varepsilon(x):\varepsilon\in\{-1,1\}^d\}$.
For fixed \(x\), the function \(y\mapsto\SC(x,y)\) is convex. Hence
\[
    \SC(x,\CM_{\B}(x))
    \le
    \max_{\varepsilon\in\{-1,1\}^d}
    \SC(x,Q^\varepsilon(x))\le
    \AR(Q_{\B,k})\,\OPT(x).
\]
Taking the supremum over profiles proves
\(R_{n,d}\le\AR(Q_{\B,k})\).
\end{proof}

This is the \(2^d\)-reflection extension of the planar comparison of Goel and
Hann-Caruthers~\cite[Lemma~6]{GoelHC23}. The rigidity theorem supplies the
fixed-frame representation needed to apply this comparison to an arbitrary
anonymous competitor.

\subsection{Proof of anonymous median optimality}

Gravin and Jia prove that the coordinate-wise median satisfies, for every
population size and every finite dimension,
\begin{equation}\label{eq:median-bound}
    R_{n,d}
    \le
    \rho:=\sqrt{6\sqrt3-8}
    <2.
\end{equation}
See~\cite[Theorem~1, spatial exponent \(2\)]{jia1}.

A mechanism improving on the median therefore has ratio below \(n-1\) for
every \(n\ge3\). Its normalized representative is a fixed-rank rule, and the
reflection lemma gives the contradiction.

We are now ready to state and prove anonymous median optimality.

\begin{theorem}\label{ext:thm:fixed-optimality}
For every $n\ge3$ and $d\ge1$,
\[
    \inf_{f\in\ExtA_{n,d}}\AR(f)=R_{n,d}.
\]
\end{theorem}

\begin{proof}
Fix \(n\ge3\) and \(d\ge1\). Suppose that an anonymous deterministic
strategyproof mechanism \(f\) satisfies \(\AR(f)<R_{n,d}\). By
\eqref{eq:median-bound},
\(\AR(f)<R_{n,d}\le\rho<2\le n-1\).

The representative corollary gives a fixed-rank rule \(Q_{\B,k}\) with
\(\AR(Q_{\B,k})\le\AR(f)\). Thus
\(\AR(Q_{\B,k})<R_{n,d}\), contradicting \cref{lem:reflection}. Hence every
anonymous deterministic strategyproof mechanism has approximation ratio at
least \(R_{n,d}\).

The coordinate-wise median attains \(R_{n,d}\). Fix all reports except one.
In each coordinate, the remaining report is clamped to the interval between
the neighboring outside order statistics. The full response is Euclidean
projection onto the product of these intervals, so truthful reporting
minimizes the agent's distance to the outcome. The mechanism is anonymous
because scalar order statistics are symmetric in the reports.

Therefore the coordinate-wise median minimizes the approximation ratio among
anonymous deterministic strategyproof mechanisms.
\end{proof}


\section{Anonymous extraction and large-population optimality}
\label{ext:sec:extraction}

We now remove anonymity from the large-population comparison. The input is
a sequence of deterministic strategyproof mechanisms whose population sizes
grow and whose approximation ratios share one bound. The output is an
anonymous mechanism on any prescribed finite number of agents.

We use two kinds of reports. The \(m\) reports of the mechanism that we want
to construct are called \emph{virtual reports}. The reports submitted to a
large source mechanism are called \emph{physical reports}. We assign the
physical reports randomly to the virtual reports. This makes the sampled
output anonymous in distribution. The main task is to make that random output
concentrate near one deterministic point, so that strategyproofness passes to
a deterministic limit. 

\begin{theorem}\label{ext:thm:extraction}
Fix finite \(d\ge1\), \(m\ge2\), and \(1\le R<\infty\). Suppose that an
unbounded set of population sizes \(N\) admits mechanisms
\(f_N\in\ExtS_{N,d}\) with \(\AR(f_N)\le R\). Then there exists a continuous,
anonymous, unanimous, deterministic strategyproof mechanism
\(h:E^m\to E\) with \(\AR(h)\le R\).
\end{theorem}

The analytic input is the following finite-profile concentration statement.
For a virtual profile \(x=(x_1,\ldots,x_m)\), each live physical agent first
chooses one of the \(m\) virtual reports uniformly. With probability
\(1-\delta\), it reports that point plus smooth noise of size at most
\(\sigma\). With probability \(\delta\), it submits an independent Gaussian
background report. The background component makes the sampling density
positive and smooth.

\begin{lemma}\label{ext:lem:finite-profile-concentration}
Fix a finite set \(\mathcal T\subseteq E^m\) of virtual profiles, parameters
\(0<\sigma,\delta<1\), and an accuracy \(\eta>0\). There is a number
\(B=B(\mathcal T,\sigma,\delta,\eta,d,m,R)\) with the following property.

Let \(N>B\), and let \(f_N\in\ExtS_{N,d}\) satisfy \(\AR(f_N)\le R\).
One can fix a set \(K\subseteq[N]\) of at most \(B\) physical labels at the
origin so that the same set \(K\) works for every \(x\in\mathcal T\).
If \(Y_N(x)\) denotes the random output obtained by sampling all remaining
physical reports from \(x\) as described above, then
\[
    \mathbb E\bigl\|Y_N(x)-\mathbb E Y_N(x)\bigr\|^2
    \le \eta
    \qquad(x\in\mathcal T).
\]
The bound \(B\) is independent of the source population \(N\).
\end{lemma}

The proof has two ingredients. First, the sum of the physical agents'
differential influences is bounded independently of \(N\). Fixing a physical
report whose influence is large decreases a bounded squared-output potential
by a definite amount, so only \(B\) labels need to be fixed. Once every
remaining influence is small, a product-measure concentration inequality
bounds the variance of the output. The appendix defines differential influence
and proves the lemma.

\paragraph{Proof idea for \cref{ext:thm:extraction}.}
Choose finite sets
\(\mathcal T_1\subseteq\mathcal T_2\subseteq\cdots\) whose union is the set
of rational virtual profiles in \(E^m\). At stage \(\ell\), choose smoothing
parameters \(\sigma_\ell,\delta_\ell\downarrow0\) and apply
\cref{ext:lem:finite-profile-concentration} to \(\mathcal T_\ell\) with
variance bound \(\ell^{-2}\). We then choose an available source population
\(N_\ell\) much larger than the resulting bound \(B_\ell\), and fix the
corresponding set \(K_\ell\) of physical labels. The fraction
\(|K_\ell|/N_\ell\) tends to zero.

For an arbitrary virtual profile \(x\), let \(Y_\ell(x)\) be the random
output of the source mechanism after the labels in \(K_\ell\) have been
fixed and the remaining physical reports have been sampled from \(x\).
Define
\[
    h_\ell(x):=\mathbb E Y_\ell(x),
    \qquad
    V_\ell(x):=
    \mathbb E\bigl\|Y_\ell(x)-h_\ell(x)\bigr\|^2.
\]
Thus \(h_\ell(x)\) is the deterministic mean of the sampled output, and
\(V_\ell(x)\) measures how far the sampled output typically lies from that
mean. Since the sampling rule depends only on the multiset
\(\{x_1,\ldots,x_m\}\), every \(h_\ell\) is anonymous.

We couple the samples at two virtual profiles \(x,x'\) by using the same
physical labels, the same chosen virtual indices, the same bounded noises,
and the same Gaussian background reports. Under this coupling, every live
physical report moves by at most
\(\max_j\norm{x_j-x'_j}\). The profile stability estimate
\eqref{ext:eq:profile-max-lipschitz} therefore gives
\[
    \norm{h_\ell(x)-h_\ell(x')}
    \le
    \sqrt d\,\max_j\norm{x_j-x'_j}.
\]
The same estimate and the approximation guarantee bound \(h_\ell\) on every
bounded set of virtual profiles. A subsequence therefore converges locally
uniformly to a continuous deterministic anonymous map \(h:E^m\to E\).

The concentration lemma gives \(V_\ell(q)\to0\) for every rational profile
\(q\). The coupling estimate extends this convergence to every fixed profile
\(x\): choose a rational \(q\) close to \(x\), compare
\(Y_\ell(x)\) and \(Y_\ell(q)\), and then let \(q\to x\). Hence
\begin{equation}\label{ext:eq:variance-every-profile}
    V_\ell(x)\longrightarrow0
    \qquad\text{for every fixed }x\in E^m.
\end{equation}

We next prove strategyproofness of the limit. Let \(x\) and \(x'\) differ
only in one virtual report, from a true point \(a\) to a deviation \(b\).
Use the same physical sampling for the two profiles. Remove the small bounded
noises while keeping the Gaussian and fixed reports unchanged. The physical
agents assigned to the changed virtual report now form a block whose reports
are all \(a\) in the truthful profile and all \(b\) in the deviating profile.
Changing these physical reports one at a time and chaining the source
mechanism's incentive inequalities gives the truthful comparison for the two
noise-free outputs.

Restoring the bounded noises changes each output by at most
\(\sqrt d\,\sigma_\ell\), by
\eqref{ext:eq:profile-max-lipschitz}. Taking expectations and using
\eqref{ext:eq:variance-every-profile} gives
\[
    \norm{h_\ell(x)-a}
    \le
    \norm{h_\ell(x')-a}
    +\sqrt{V_\ell(x')}
    +2\sqrt d\,\sigma_\ell.
\]
Letting \(\ell\to\infty\) proves
\(\norm{h(x)-a}\le\norm{h(x')-a}\), so \(h\) is strategyproof.

Finally, apply the source mechanism's approximation inequality to each
sampled physical profile. Only \(|K_\ell|\) physical reports are fixed, and
their fraction tends to zero. The distribution of one sampled live report
converges to the uniform distribution on \(x_1,\ldots,x_m\), while
\eqref{ext:eq:variance-every-profile} allows the random output
\(Y_\ell(x)\) to be replaced by its mean \(h_\ell(x)\). Passing to the limit
gives, for every comparison point \(z\in E\),
\[
    \frac1m\sum_{j=1}^m\norm{x_j-h(x)}
    \le
    \frac Rm\sum_{j=1}^m\norm{x_j-z}.
\]
Minimizing over \(z\) proves \(\AR(h)\le R\). The complete proof, including
the sampling law, the influence calculation, the pinning argument, and the
concentration estimate, is given in \cref{ext:app:extraction}.

\subsection{Large-population optimality}
\label{sec:extraction-application}

Recall from \eqref{ext:eq:uniform-definitions} that
$C_d=\sup_{m\ge2}R_{m,d}$. The two-agent term does not change this
supremum. To see this, fix a two-agent profile $x$, let $v$ be its chosen
coordinate-wise median, and let $z^*$ be an optimal location. Repeat each
report $T$ times and add one report at $v$. The new population has size
$2T+1$, its coordinate-wise median is still $v$, and its ratio is at least
\[
    \frac{T\,\SC(x,v)}{T\,\OPT(x)+\norm{v-z^*}}.
\]
For profiles with positive optimum, this converges to the original ratio as
$T\to\infty$; zero-optimum profiles have ratio one. Hence
$C_d=\sup_{m\ge3}R_{m,d}$.

The extraction theorem turns any improvement on \(C_d\) at arbitrarily large
populations into an anonymous improvement at one fixed population, contradicting
\cref{ext:thm:fixed-optimality}.

We are now ready to state and prove the unrestricted large-population result.

\begin{theorem}\label{ext:thm:uniform}
For every fixed finite $d\ge1$,
\[
    \lim_{n\to\infty}a_{n,d}=C_d.
\]
Equivalently, every sequence $f_n\in\ExtS_{n,d}$ satisfies
$\liminf_{n\to\infty}\AR(f_n)\ge C_d$, and
\[
    \inf_{(f_n):\,f_n\in\ExtS_{n,d}}
    \ \sup_{n\ge2}\AR(f_n)=C_d.
\]
\end{theorem}

\begin{proof}
Suppose that a sequence \(f_n\in\ExtS_{n,d}\) satisfies
\(\liminf_{n\to\infty}\AR(f_n)<C_d\). Choose a number \(R<C_d\) strictly
above this limit inferior. Then arbitrarily large population sizes \(N\)
satisfy \(\AR(f_N)\le R\).

By the preceding observation and the definition of \(C_d\), choose
\(m\ge3\) with \(R_{m,d}>R\).
The extraction theorem gives an anonymous deterministic strategyproof
mechanism on \(m\) agents with approximation ratio at most \(R\). This
contradicts the fixed-population optimality theorem
\cref{ext:thm:fixed-optimality}. Therefore every sequence satisfies
 $\liminf_{n\to\infty}\AR(f_n)\ge C_d$.

Let \(a_{n,d}\) denote the optimal ratio among deterministic strategyproof
mechanisms with \(n\) agents in dimension \(d\). Choose \(f_n\) with
\(\AR(f_n)\le a_{n,d}+1/n\). The preceding paragraph gives
\(\liminf_n a_{n,d}\ge C_d\). Conversely, the coordinate-wise median is an
allowed deterministic strategyproof mechanism, so
\(a_{n,d}\le R_{n,d}\le C_d\) for every \(n\ge3\). Hence
$\lim_{n\to\infty}a_{n,d}=C_d$.

Every family of mechanisms has worst-case ratio at least its limit inferior,
and hence at least \(C_d\). The median family has worst-case ratio \(C_d\)
by definition.
\end{proof}

Taking the supremum over dimensions and using the matching median bounds of
Gravin and Jia gives the joint guarantee.

\begin{corollary}\label{ext:cor:all-dimensions}
For families of deterministic strategyproof mechanisms,
\[
    \inf_{(f_{n,d}):\,f_{n,d}\in\ExtS_{n,d}}
    \ \sup_{\substack{n\ge2\\d\ge1}}
    \AR(f_{n,d})
    =
    \sqrt{6\sqrt3-8}.
\]
\end{corollary}

\begin{proof}
For every fixed \(d\), \cref{ext:thm:uniform} gives
\(\sup_n\AR(f_{n,d})\ge C_d\). Gravin and Jia's upper bound and matching
examples imply
\(\sup_d C_d=\sqrt{6\sqrt3-8}\). Therefore every mechanism family has joint
worst-case ratio at least this value. The coordinate-wise median family
attains the matching upper bound.
\end{proof}


\section{Conclusion and future work}
\label{sec:conclusion}

We studied deterministic strategyproof facility location in
finite-dimensional Euclidean spaces. At every fixed population of at least
three agents, the coordinate-wise median is optimal among anonymous
mechanisms. Without anonymity, agent identities can improve the ratio at a
fixed population size, but this advantage disappears in the
large-population limit. In each fixed dimension, the optimal ratio over all
deterministic strategyproof mechanisms converges to the supremum of the
median's ratios over population sizes.

The proof also gives structural results for strategyproof aggregation.
Translation-equivariant anonymous mechanisms with bounded individual
influence use one orthonormal frame and fixed coordinate ranks. Convexity of
the full range is equivalent to continuity, so unanimity supplies the
regularity needed by the characterization. Anonymous extraction then
transfers the fixed-population lower bound to mechanisms that may depend on
agent identities.

A natural direction is to characterize broader classes of strategyproof
mechanisms under Euclidean-distance preferences, including mechanisms
without translation equivariance or bounded individual influence. It is
also open to determine the exact fixed-population optimum without anonymity
in higher dimensions. Another direction is to study different objectives,
such as maximum cost, weighted social cost, and \(L_p\)-aggregations of the
agents' costs. These objectives may favor different coordinate ranks or
mechanisms beyond the median, and may require new links between structural
characterization and approximation.
\appendix
\section{Normalization by common translations}\label{app:normalization}
We prove \cref{thm:normalization} by translating all reports along
one direction and taking a locally uniform limit. The projection
identities bound the output change under common translation. The
resulting displacement map is bounded and monotone, so exposed points
select a limit along a direction shared by a dense set of profiles.
Its independence of fixed offsets yields translation equivariance.

\subsection{Finite-dimensional analytic tools}\label{app:analytic-tools}
We use Rademacher's theorem: a Lipschitz map between finite-dimensional Euclidean spaces is differentiable outside a set of Lebesgue measure zero. We also use Fubini's theorem and the fact that a one-dimensional Lipschitz function is absolutely continuous, so the difference of its endpoint values equals the integral of its derivative along the interval. These standard facts can be found in~\cite{EvansGariepy2015}. ``Almost everywhere'' below always refers to the appropriate finite-dimensional Lebesgue measure.

The next two standard analytic facts connect projection geometry with
weak and pointwise derivative bounds. We include the elementary proofs
for use in both normalization and extraction.

\begin{lemma}\label{lem:projection-derivative}
If projection $P_C:E\to E$ onto a nonempty closed convex set is differentiable at $a$, then its derivative is a symmetric matrix $A$ with $0\preceq A\preceq I$. Here $A\preceq B$ means $\ip{v}{Av}\le\ip{v}{Bv}$ for every $v$.
\end{lemma}
\begin{proof}
The projection is the gradient of the classical convex potential
\cite{Rockafellar1970}
\[
 \Phi_C(a)=\frac12\norm{a}^2-\frac12\dist(a,C)^2
 =\sup_{z\in C}\left(\ip{a}{z}-\frac12\norm{z}^2\right).
\]
The supremum is attained uniquely at $P_Ca$; completing the square verifies this. It is a supremum of affine functions of $a$, so $\Phi_C$ is convex. To compute its gradient, put $p=P_Ca$ and $p'=P_C(a+h)$. Comparing the two maximizers gives
$\ip{p}{h}\le\Phi_C(a+h)-\Phi_C(a)\le\ip{p'}{h}$.
Since $\norm{p'-p}\le\norm{h}$, the difference between the bounds is at most $\norm{h}^2$. Thus $\nabla\Phi_C(a)=P_Ca$, and the gradient is continuous.

Where this gradient is differentiable, its derivative is symmetric. One direct verification is to integrate it around a small rectangle with side vectors $tu,tv$. The integral is zero because it is a gradient. Its first-order expansion is
$ t^2(\ip{Au}{v}-\ip{Av}{u})+o(t^2)$,
which gives symmetry after division by $t^2$ and passage to the limit. Finally, apply firm nonexpansiveness to $a+tv,a$, divide by $t^2$, and let $t\to0$: $\norm{Av}^2\le\ip{v}{Av}$.
For an eigenvector of the symmetric matrix $A$, this says $\lambda^2\le\lambda$, hence $0\le\lambda\le1$.
\end{proof}

An almost-everywhere derivative bound extends to every pair of points
by integrating along nearby parallel lines.

\begin{lemma}\label{lem:ae-lipschitz}
Let $G:\R^r\to\R^s$ be Lipschitz. If the operator norm of $DG$ is at most $L$ almost everywhere, then $G$ is $L$-Lipschitz everywhere.
\end{lemma}
\begin{proof}
Fix two points and consider lines parallel to the segment joining them. For almost every parallel line, Fubini's theorem makes the set where $G$ is not differentiable or where $\lVert DG\rVert_{\mathrm{op}}>L$ have one-dimensional measure zero on that line. The restriction of $G$ is absolutely continuous, and its derivative has norm at most $L$ almost everywhere. Integrating along that line gives the $L$-Lipschitz inequality between any two of its points. Approximate the chosen segment by segments on such parallel lines and use continuity to pass to the limit. When $r=1$, apply absolute continuity directly. Since the original pair was arbitrary, the bound holds globally.
\end{proof}

For a nonempty compact convex set $K\subseteq E$, the support function
$h_K(u)=\max_{p\in K}\ip{u}{p}$ is Lipschitz. At a differentiability point $u$, the maximizer is unique: every maximizer $p$ gives the supporting inequalities $h_K(u+tv)\ge h_K(u)+t\ip{p}{v}$, and differentiability for both signs of $t$ forces $p=\nabla h_K(u)$. Thus almost every direction exposes a unique maximizing point of $K$. This conclusion applies to singleton and lower-dimensional sets as well.

\subsection{The joint derivative bound}\label{app:translation-derivative}
\label{ext:app:influence-proof}
We prove \cref{ext:lem:influence} here. Radial invariance controls the
sum of the derivative blocks. We then derive the maximum-report
Lipschitz estimate and the sharper bound for common translations.

\begin{proof}[Proof of \cref{ext:lem:influence}]
Let $f:E^n\to E$ be continuous and strategyproof.
By \eqref{eq:global-lipschitz}, $f$ is Lipschitz. At a differentiability point $x\in E^n$, write its derivative in blocks as
\[
 Df(x)=[A_1(x),\ldots,A_n(x)],\qquad
 S(x)=\sum_{i=1}^n A_i(x).
\]
The block $A_i$ is the derivative of the corresponding unilateral projection, so \cref{lem:projection-derivative} makes each $A_i$ symmetric positive semidefinite. Thus $S$ is symmetric positive semidefinite. It describes the derivative when all reports move by the same small vector. Radial invariance gives the bound $S\preceq I$.

For $t>0$, define $R_t:E^n\to E^n$ by $(R_tx)_i=tx_i+(1-t)f(x)$.
Equation~\eqref{eq:radial-all} gives
\begin{equation}\label{eq:radial-map-invariance}
 f(R_tx)=f(x).
\end{equation}
Moreover, $R_t^{-1}=R_{1/t}$: substitute \eqref{eq:radial-map-invariance} into the composition to verify that it returns each $x_i$. Both maps are Lipschitz by \eqref{eq:global-lipschitz}, so $R_t$ is bi-Lipschitz.

Let $N$ be the null set where $f$ is not differentiable. A Lipschitz map between Euclidean spaces of the same dimension sends null sets to null sets; this follows, for example, by covering a null set with small cubes and bounding the volumes of their images by a constant multiple of the cube volumes. Since $R_t^{-1}$ is Lipschitz, $R_t^{-1}(N)$ is null. Remove the single null set $N_* = N\ \cup\!\bigcup_{t\in\mathbb Q_{>0}} R_t^{-1}(N)$.
At $x\notin N_*$, $f$ is differentiable both at $x$ and at $R_tx$ for every positive rational $t$. Differentiate \eqref{eq:radial-map-invariance} in a common-translation direction. Under a common infinitesimal displacement $v$, every report of $R_tx$ has displacement $tv+(1-t)S(x)v$. The chain rule therefore gives, simultaneously for all positive rational $t$,
\begin{equation}\label{eq:translation-matrix-identity}
 S(R_tx)\bigl[tI+(1-t)S(x)\bigr]=S(x).
\end{equation}

Suppose $S(x)v=\lambda v$ for some $v\ne0$ and $\lambda>1$. Choose a positive rational $t>\lambda/(\lambda-1)$. Applying \eqref{eq:translation-matrix-identity} to $v$ yields
$S(R_tx)v=\frac{\lambda}{t+(1-t)\lambda}\,v$.
The denominator is negative, contradicting positive semidefiniteness of $S(R_tx)$. Thus
\begin{equation}\label{eq:translation-derivative-bound}
 0\preceq S(x)\preceq I\qquad\text{almost everywhere in }E^n.
\end{equation}
The common exceptional set $N_*$ makes this choice of rational $t$
valid for every $x\notin N_*$.

Taking traces gives the last inequality in
\eqref{ext:eq:influence-matrix}. For a unit vector $w$ and report increments
$\Delta_i$, positive-semidefinite Cauchy--Schwarz yields
\begin{align*}
 \left|\sum_i w^{\mathsf T}A_i\Delta_i\right|
 &\le
 \left(\sum_iw^{\mathsf T}A_iw\right)^{1/2}
 \left(\sum_i\Delta_i^{\mathsf T}A_i\Delta_i\right)^{1/2}\\
 &\le\left(\sum_i\operatorname{tr}(A_i)\norm{\Delta_i}^2\right)^{1/2}
 \le\sqrt d\,\max_i\norm{\Delta_i}.
\end{align*}
To integrate this almost-everywhere estimate on arbitrary segments,
mollify $f$ with nonnegative compactly supported approximate identities.
Its weak derivative equals its almost-everywhere derivative because $f$
is Lipschitz. Each mollification preserves positivity of the blocks and
the bound on their sum. The displayed estimate therefore holds everywhere
for it, and integration on a segment gives
\eqref{ext:eq:profile-max-lipschitz}. Local uniform convergence returns
that estimate for $f$.
\end{proof}

We next specialize the derivative bound to common translations, where the
constant is one.  Use the invertible linear parameterization
$(r_1,\ldots,r_{n-1},v)\longmapsto
 (r_1+v,\ldots,r_{n-1}+v,v)$.
Differentiation in $v$ is exactly differentiation in a common-translation direction. By Fubini and \eqref{eq:translation-derivative-bound}, for almost every relative profile $r$, the Lipschitz map $G_r(v)=f(r_1+v,\ldots,r_{n-1}+v,v)$
has derivative of operator norm at most $1$ almost everywhere in $v$. By \cref{lem:ae-lipschitz}, $G_r$ is $1$-Lipschitz for each such $r$. These relative profiles form a dense set. Approximate any other $r$ by them and use \eqref{eq:global-lipschitz} to pass to the limit at any fixed pair of translation vectors. We obtain
\begin{equation}\label{eq:common-translation-lipschitz}
 \norm{f(x_1+v,\ldots,x_n+v)-f(x_1+w,\ldots,x_n+w)}
 \le\norm{v-w}
\end{equation}
for every profile $x$ and every $v,w\in E$. General profiles are covered by taking $r_i=x_i-x_n$ and shifting the translation variable by $x_n$.

\subsection{A bounded monotone displacement and a limit at infinity}\label{app:monotone-limit}
We now take $f\in\F_{n,d}$ and, for a fixed profile $x$, define
$F_x(v)=v-f(x_1+v,\ldots,x_n+v)$. A map $F:E\to E$ is
\emph{monotone} if $\ip{F(v)-F(w)}{v-w}\ge0$ for all $v,w$. By \eqref{eq:common-translation-lipschitz} and Cauchy--Schwarz,
\begin{equation}\label{eq:monotone-displacement}
 \ip{F_x(v)-F_x(w)}{v-w}
 =\norm{v-w}^2-\ip{f(x+v)-f(x+w)}{v-w}\ge0.
\end{equation}
Here and below, $x+v$ means adding $v$ to every report. The displacement is also bounded: comparing with the unanimous profile $(v,\ldots,v)$ gives
\begin{equation}\label{eq:bounded-displacement}
 \norm{F_x(v)}=\norm{f(v,\ldots,v)-f(x+v)}
 \le\sum_{i=1}^n\norm{x_i}.
\end{equation}
Let $K_x=\overline{\conv}(F_x(E))$, which is compact and convex by this bound and finite dimensionality.

An exposed direction determines the limit of this displacement map
along every parallel ray.

\begin{lemma}\label{lem:exposed-limit}
Suppose $u$ exposes a unique maximizing point $p_x$ of $K_x$. Then for every fixed $v\in E$, when $t\to\infty$, $F_x(tu+v)\longrightarrow p_x$.
The limit is independent of $v$.
\end{lemma}
\begin{proof}
Fix $v,z\in E$. Monotonicity gives $\ip{F_x(tu+v)-F_x(z)}{tu+v-z}\ge0$.
Take any sequence $t\to\infty$ along which the first displacement converges to a cluster point $p$; such subsequences exist by boundedness. Divide the inequality by $t$ and pass to the limit. The terms containing $(v-z)/t$ vanish because the displacements are bounded. Hence $\ip{u}{p}\ge\ip{u}{F_x(z)}$ for $z\in E$.
Since $p\in K_x$, this says that $p$ maximizes the linear functional over $K_x$: the range has the same support supremum as its closed convex hull. Uniqueness gives $p=p_x$. Every cluster point is therefore $p_x$, so the bounded family converges to $p_x$. The argument works for every fixed offset $v$ and selects the same point.
\end{proof}

\subsection{One direction for all profiles and local uniform convergence}\label{app:common-direction}
Choose a countable dense subset $\mathcal D\subseteq E^n$. For each $x\in\mathcal D$, almost every $u\in E$ exposes a unique point of $K_x$, by the support-function fact in \cref{app:analytic-tools}. A countable union of exceptional null sets is null. Excluding also $u=0$, choose one nonzero direction that works for every $x\in\mathcal D$.

Define $f_t(x)=f(x+tu)-tu$ for $t>0$. By \cref{lem:exposed-limit}, $f_t(x)=-F_x(tu)$ converges for each $x\in\mathcal D$. Every $f_t$ has the same $1$-Lipschitz bound in the profile metric $d_1$: $\norm{f_t(x)-f_t(x')}\le d_1(x,x')$.
This promotes convergence on the dense set to locally uniform convergence. Explicitly, let $L\subseteq E^n$ be compact and fix $\delta>0$. Choose finitely many points $x^1,\ldots,x^q\in\mathcal D$ forming a $\delta$-net for $L$. For $x\in L$, choose $x^j$ with $d_1(x,x^j)<\delta$. Then for $s,t>0$, $\norm{f_t(x)-f_s(x)}
 \le2\delta+\max_{1\le j\le q}\norm{f_t(x^j)-f_s(x^j)}$.
The finite maximum tends to zero as $s,t\to\infty$. Since $\delta$ is arbitrary, the family is uniformly Cauchy on $L$. Completeness of $E$ gives a limit $\bar f$, uniform on every compact set. It retains the same global Lipschitz bound and is continuous.

Every $f_t$ is anonymous, unanimous, and strategyproof: it is just a translated copy of $f$. Passing to the limit in each fixed incentive inequality proves strategyproofness of $\bar f$; anonymity and unanimity pass to the limit in the same way. Thus $\bar f\in\F_{n,d}$. Translation of profiles and outcomes preserves both social cost and the optimum, so $\AR(f_t)=\AR(f)$. If this ratio is finite, pass to the limit in
$\SC(x,f_t(x))\le\AR(f)\OPT(x)$
for each fixed profile with positive optimum. Continuity of the finite sum of distances gives $\AR(\bar f)\le\AR(f)$. If $\AR(f)=\infty$, the same comparison is automatic.

\subsection{Equivariance of the limit}\label{app:limit-equivariance}
Fix $x\in\mathcal D$ and $v\in E$. The offset-independent limit in \cref{lem:exposed-limit} gives
\begin{align*}
 \bar f(x+v)=\lim_{t\to\infty}\bigl(f(x+v+tu)-tu\bigr)=v-\lim_{t\to\infty}F_x(tu+v)
 =v-p_x=\bar f(x)+v.
\end{align*}
Local uniform convergence gives the limit at $x+v$. Approximate an arbitrary profile by points of $\mathcal D$ and use continuity to obtain the same identity for every $x$ and every fixed $v$. Thus $\bar f$ is translation-equivariant.

By \cref{cor:automatic-homogeneity}, translation equivariance also gives positive homogeneity. Thus $\bar f$ is normalized, completing the proof of \cref{thm:normalization}.\qed

\section{Regularity from convex range}
\label{ext:app:regularity}
We prove the convex-range implication used in
\cref{intro:thm:range-continuity}. The first stage is a binary theorem.
Its proof combines dimension-free anchored geometry with a weak-derivative
argument. The second stage descends through coalition menus to obtain
convex projection responses for every finite population. The analytic
stage replaces the planar continuity step in
Peters, van der Stel, and Storcken~\cite[Lemma~11 and Theorem~5]{peters1993range}.

The menu identities of \cref{lem:raw-menus} hold before continuity.
We use them throughout this appendix. The star-shapedness and anchored
continuity below are the binary geometry developed by Peters, van der Stel,
and Storcken~\cite[Lemmas~4, 8, and~9]{peters1993range}; their Lemma~10
also identifies the equality of the two agents' costs across a discontinuity.
We record the forms needed for the incentive envelopes, then use their weak
derivatives to establish continuity in arbitrary finite dimension.

\subsection{The binary theorem}\label{review:app:binary}
The binary statement supplies the induction step for general
populations. We prove it by establishing anchored continuity, constructing incentive
envelopes, and bounding their derivative measures.

\begin{theorem}
\label{review:thm:binary-regularity}
Let $h:E^2\to E$ be deterministic and strategyproof in both arguments, and
suppose its full range $C=h(E^2)$ is convex. Then every unilateral menu is
closed and convex, its response is Euclidean projection onto that menu, and
\[
 \norm{h(a,b)-h(a',b')}\le\norm{a-a'}+\norm{b-b'}
 \qquad(a,a',b,b'\in E).
\]
\end{theorem}

Fix a binary rule $h:E^2\to E$ that is strategyproof in both arguments and
has convex full range $C$. By \cref{lem:raw-menus}, $C$ is
closed and its diagonal response is nearest-point selection onto $C$.
Consequently, with $\pi=P_C$,
\begin{equation}\label{ext:eq:binary-diagonal}
 h(a,a)=\pi(a).
\end{equation}
For $b\in E$, write $M_b=\{h(a,b):a\in E\}$.

At a report belonging to the full range, the corresponding menu
contains every segment from that report to one of its attainable points.

\begin{lemma}\label{ext:lem:star}
For $c\in C$, the menu $M_c$ is star-shaped about $c$.
\end{lemma}
\begin{proof}
Fix $v\in M_c$, put $e=v-c$, and set
\[
 z_t=c+te,\qquad p_t=h(z_t,c),\qquad D(t)=\norm{z_t-p_t},\qquad 0\le t\le1.
\]
Convexity of $C$ gives $h(z_t,z_t)=z_t$. Strategyproofness of the second
argument, whose true point is $c$, gives
$\norm{c-p_t}\le t\norm e$. For $t>0$, this implies
$\ip{e}{z_t-p_t}\ge0$.
For $0\le s<t\le1$, the first argument's incentive inequality gives
\begin{align*}
 D(s)^2&\le\norm{z_s-p_t}^2=D(t)^2-2(t-s)\ip{e}{z_t-p_t}+(t-s)^2\norm e^2\le D(t)^2+(t-s)^2\norm e^2.
\end{align*}
Apply this inequality to the $N$ equal subintervals of $[s,t]$ and sum: $D(s)^2\le D(t)^2+\frac{(t-s)^2\norm e^2}{N}$.
Let $N\to\infty$. Since $D(1)=0$, every $D(s)$ is zero. Thus
$[c,v]\subseteq M_c$.
\end{proof}

An anchored nearest point is therefore unique: if $c$ is a nearest point of
$M_c$ to $a$, a distinct equally near point would have a midpoint with $c$
that lies in $M_c$ and is strictly closer to $a$. It follows that the response
$a'\mapsto h(a',c)$ is continuous at each $a$ satisfying $h(a,c)=c$.
Indeed, along $a_k\to a$, comparison with $c$ bounds the outputs. Every
cluster point belongs to the closed menu and is nearest to $a$, so equals
$c$ by uniqueness. We use this continuity to compare the other agent's costs at tied
nearest points.

\begin{lemma}\label{ext:lem:cross}
If $p=h(a,b)$ and $q$ is any nearest point of $M_b$ to $a$, then
\begin{equation}\label{ext:eq:cross}
 \norm{b-p}=\norm{b-q}.
\end{equation}
The same assertion holds after interchanging the two arguments.
\end{lemma}
\begin{proof}
Assume $p\ne q$ and set $a_t=(1-t)a+tq$ for $0<t<1$. For $r\in M_b$,
\[
 \norm{a_t-r}^2-\norm{a_t-q}^2
 =(1-t)\bigl(\norm{a-r}^2-\norm{a-q}^2\bigr)+t\norm{r-q}^2.
\]
Thus $q$ is the unique nearest point at $a_t$, so $h(a_t,b)=q$.
Outcome replacement gives $h(a_t,q)=q$. Passing to the limit in the
nearest-point inequalities makes $q$ a nearest point of $M_q$ to $a$.
Anchored uniqueness gives $h(a,q)=q$.
Likewise $h(a,p)=p$, and continuity of that anchored response gives
$h(a_t,p)\to p$. The second argument's incentive inequalities now give
\[
 \norm{b-q}\le\norm{b-h(a_t,p)}\longrightarrow\norm{b-p},
 \qquad
 \norm{b-p}\le\norm{b-h(a,q)}=\norm{b-q}.
\]
This proves the equality. Applying the proof to $(a,b)\mapsto h(b,a)$ proves
the exchanged assertion.
\end{proof}

\subsection{Jointly continuous scalar envelopes and measurability}\label{regularity:envelopes}
Define
\begin{equation}\label{ext:eq:envelopes}
 u(a,b)=\ip{a}{h(a,b)}-\tfrac12\norm{h(a,b)}^2,
 \qquad
 v(a,b)=\ip{b}{h(a,b)}-\tfrac12\norm{h(a,b)}^2.
\end{equation}
For a finite convex function $\phi:E\to\mathbb R$, its subdifferential at
$a$ is
\[
 \partial\phi(a)=\{g\in E:\phi(a')\ge\phi(a)+\ip{g}{a'-a}
 \text{ for every }a'\in E\}.
\]
We write $\partial_a u(a,b)=\partial(u(\cdot,b))(a)$ for the
subdifferential in the first variable, and define $\partial_b v(a,b)$
similarly. Strategyproofness gives the pointwise envelope representation
$u(a,b)=\sup_{r\in E}
 \left\{\ip{a}{h(r,b)}-\tfrac12\norm{h(r,b)}^2\right\}$.
Hence $u$ is finite and convex in $a$, with $h(a,b)\in\partial_a u(a,b)$.
Similarly $v$ is convex in $b$ and $h(a,b)\in\partial_b v(a,b)$.

To bound the outcomes locally, report $b$ in the first position and
use \eqref{ext:eq:binary-diagonal}:
\begin{equation}\label{ext:eq:binary-local-bound}
 \norm{a-h(a,b)}\le\norm{a-\pi(b)}.
\end{equation}
The two supporting inequalities $\ip{h(a,b)}{a'-a}\le u(a',b)-u(a,b)\le\ip{h(a',b)}{a'-a}$
therefore make $u$ locally Lipschitz in its first variable, uniformly on
compact product sets. 

For cross-coordinate continuity, fix $a$ and let $b_k\to b$. Every cluster
point of the bounded outputs $h(a,b_k)$ is a nearest point to $b$ in the
closed row menu $\{h(a,r):r\in E\}$. By the exchanged form of
\cref{ext:lem:cross}, every such cluster point has the same distance to $a$
as $h(a,b)$. Since $u(a,b_k)=\tfrac12\norm a^2-\tfrac12\norm{a-h(a,b_k)}^2$,
we have $u(a,b_k)\to u(a,b)$. The uniform own-coordinate Lipschitz bound then
makes $u$ jointly continuous. The same reasoning applies to $v$.

Joint continuity of the envelopes now gives a measurable description
of their partial gradients.
For each coordinate $i$, the limits of
\[
 m\bigl(u(a+m^{-1}e_i,b)-u(a,b)\bigr),\qquad
 m\bigl(u(a,b)-u(a-m^{-1}e_i,b)\bigr)
\]
exist and are finite by convexity. The set where the two limits agree for
every $i$ is Borel, because these are limits of continuous difference
quotients. For a finite convex function on $E$, existence of all two-sided coordinate
derivatives makes its subdifferential a singleton: any two subgradients have
the same inner product with every coordinate vector, and hence are equal.
The function is therefore differentiable. Thus this is exactly the partial
differentiability set
in $a$, and the gradient there is Borel measurable.
For each fixed $b$ its exceptional section has Lebesgue measure zero.
Tonelli's theorem on bounded product boxes, followed by a countable
exhaustion, makes the exceptional set null in $E^2$. At every partial
differentiability point, the subgradient inclusion gives $h=\nabla_a u$.
Hence $h$ is measurable for completed Lebesgue measure and belongs to
$L^\infty_{\mathrm{loc}}$ by \eqref{ext:eq:binary-local-bound}.

The identities $\nabla_a u=h$ and $\nabla_b v=h$ also hold weakly. For a
compactly supported smooth test function, integrate by parts on the convex
slices and then apply Fubini; local boundedness makes all integrals finite.

\subsection{Positive derivative measures and a transport identity}\label{regularity:transport}
All derivatives in this subsection are distributional. Set
\[
 A=D_a h=D^2_{aa}u,\qquad B=D_b h=D^2_{bb}v,
 \qquad \gamma=\tfrac12\norm h^2,
\]
with the index conventions $A_{ij}=\partial_{a_j}h_i$ and
$B_{ij}=\partial_{b_j}h_i$. Weak partial derivatives commute, so $A,B$ are
symmetric. For a constant vector $\xi$ and a nonnegative test function,
convexity on the $a$-slices makes $\xi^{\mathsf T}A\xi$ a positive
distribution. A positive distribution is a Radon measure; polarization
therefore makes $A$ a positive semidefinite matrix-valued Radon measure.
The same argument gives this property for $B$.

Consequently $h\in BV_{\mathrm{loc}}(E^2;E)$: it is locally
integrable, and every entry of its full weak derivative $Dh=(A\ B)$
is a Radon measure of locally finite total variation. We will bound
these measures by Lebesgue measure through a transport identity.

From \eqref{ext:eq:envelopes} and the weak gradient identities, the
product rule with the smooth coordinate functions gives
\begin{equation}\label{ext:eq:weak-energy}
 \nabla_a\gamma=Aa,\qquad \nabla_b\gamma=Bb.
\end{equation}
The undifferentiated terms cancel by $\nabla_a u=h$ and
$\nabla_b v=h$; all products involving $A$ or $B$ use smooth
coordinate functions.

Symmetry of $B$ and commutation of weak derivatives give, componentwise,
\[
 \partial_{b_k}A_{ij}
 =\partial_{a_j}B_{ik}
 =\partial_{a_j}B_{ki}
 =\partial_{b_i}A_{kj}.
\]
Taking mixed derivatives in \eqref{ext:eq:weak-energy} gives
\begin{align*}
 \partial_{b_i}\partial_{a_j}\gamma
 =\sum_k a_k\partial_{b_i}A_{kj}, \qquad
 \partial_{a_j}\partial_{b_i}\gamma
 =\sum_k b_k\partial_{a_j}B_{ki}
 =\sum_k b_k\partial_{b_i}A_{kj}.
\end{align*}
Since $a_k$ is independent of $b_i$ and $b_k$ of $a_j$, subtracting
these identities yields
\begin{equation}\label{ext:eq:transport}
 \sum_k(a_k-b_k)\partial_{b_k}A_{ij}=0,
 \qquad\text{or}\qquad (a-b)\cdot\nabla_b A=0.
\end{equation}
This identity holds on the full product space, including its diagonal.

The transport identity makes the measure invariant under contraction
of the second report toward the first.
For $\varepsilon>0$, put
$\Phi_\varepsilon(a,b)=(a,a+\varepsilon(b-a))$.
Its Jacobian determinant is $\varepsilon^d$. The pullback is applied to each
matrix entry as a scalar distribution, with convention
\begin{equation}\label{ext:eq:pullback-definition}
 \langle\Phi_\varepsilon^*T,\varphi\rangle
 =\varepsilon^{-d}\langle T,\varphi\circ\Phi_\varepsilon^{-1}\rangle.
\end{equation}
This scalar-distribution pullback preserves positivity; see
\cite[Section~10.1.3 and Proposition~10.10]{ext:Dyatlov}.

In this notation, \eqref{ext:eq:transport} gives
\begin{equation}\label{ext:eq:pullback-invariance}
 \Phi_\varepsilon^*A=A.
\end{equation}
Indeed, for a fixed test function $\varphi$, set
\[
 \psi_\varepsilon(a,b)
 =\varepsilon^{-d}\varphi\bigl(a,a+\varepsilon^{-1}(b-a)\bigr).
\]
Direct differentiation gives
\[
 \partial_\varepsilon\psi_\varepsilon
 =-\varepsilon^{-1}
 \bigl(d\psi_\varepsilon+(b-a)\cdot\nabla_b\psi_\varepsilon\bigr).
\]
On a compact interval of positive $\varepsilon$, these test functions have
one common compact support and depend smoothly on $\varepsilon$ in the
test-function topology. The weak meaning of \eqref{ext:eq:transport} is
\[
 \langle A_{ij},d\psi+(b-a)\cdot\nabla_b\psi\rangle=0.
\]
Differentiating the pairing in \eqref{ext:eq:pullback-definition} therefore
gives zero. At $\varepsilon=1$ it is the identity, proving
\eqref{ext:eq:pullback-invariance}. The Jacobian factor supplies the dimension term in the preceding
calculation.

\subsection{Domination of the derivative measures}\label{regularity:domination}
Set $h_\varepsilon=h\circ\Phi_\varepsilon$. For this invertible linear
map, the weak chain rule follows directly by testing against a compactly
supported smooth function and changing variables. It gives, for
$0<\varepsilon<1$,
\begin{equation}\label{ext:eq:derivative-budget}
 D_a h_\varepsilon
 =\Phi_\varepsilon^*A+(1-\varepsilon)\Phi_\varepsilon^*B
 =A+(1-\varepsilon)\Phi_\varepsilon^*B\succeq A\succeq0.
\end{equation}

We next show that the contracted outcomes converge to the diagonal
projection. Write
$b_\varepsilon=a+\varepsilon(b-a)$ and $p=\pi(a)$. Since
$h(a,b_\varepsilon)\in C$, projection geometry and
\eqref{ext:eq:binary-local-bound} imply
\begin{align*}
 \norm{h(a,b_\varepsilon)-p}^2
 &\le\norm{a-h(a,b_\varepsilon)}^2-\norm{a-p}^2\\
 &\le\norm{a-\pi(b_\varepsilon)}^2-\norm{a-\pi(a)}^2\\
 &\le2\dist(a,C)\,\varepsilon\norm{b-a}
       +\varepsilon^2\norm{b-a}^2.
\end{align*}
The last estimate uses nonexpansiveness of $\pi$. Therefore
$h_\varepsilon\to\pi(a)$ locally uniformly and in $L^1_{\mathrm{loc}}$.
This $L^1_{\mathrm{loc}}$ convergence gives convergence of the
weak derivatives as distributions.

Let $J_\pi(a)$ be the almost-everywhere Jacobian of the Lipschitz projection
$\pi$. By \cref{lem:projection-derivative},
$0\preceq J_\pi\preceq\ExtId$ almost everywhere, and its weak derivative is
this absolutely continuous matrix measure. For every constant $\xi$ and
nonnegative test function $\varphi$, \eqref{ext:eq:derivative-budget} gives
\begin{align*}
 0\le\langle\xi^{\mathsf T}A\xi,\varphi\rangle
 \le\lim_{\varepsilon\downarrow0}
       \langle\xi^{\mathsf T}D_a h_\varepsilon\xi,\varphi\rangle=\int\varphi(a,b)\,\xi^{\mathsf T}J_\pi(a)\xi\,da\,db
 \le\norm\xi^2\int\varphi.
\end{align*}
Write $\ExtLeb^{2d}$ for Lebesgue measure on $E^2$. Thus, as inequalities
of entire matrix-valued measures,
\begin{equation}\label{ext:eq:measure-domination}
 0\preceq A\preceq J_\pi(a)\,\ExtLeb^{2d}
 \preceq\ExtId\,\ExtLeb^{2d},
 \qquad
 0\preceq B\preceq\ExtId\,\ExtLeb^{2d}.
\end{equation}
The assertion for $B$ follows by applying the same argument to the
rule $(a,b)\mapsto h(b,a)$.

Taking traces in \eqref{ext:eq:measure-domination} gives
$0\le\Exttr A\le d\ExtLeb^{2d}$ and the same for $B$. Every entry of a
positive semidefinite matrix measure is absolutely continuous with respect
to its trace measure: on a relatively compact set with trace zero, its
positive semidefinite value has zero trace and is therefore the zero matrix.
Exhaustion handles all null sets. Thus all entries of $A,B$ are absolutely
continuous with respect to Lebesgue measure, and their density matrices lie
between zero and $\ExtId$ almost everywhere. These bounds exclude both jump
and singular-continuous derivative parts.

\subsection{Recovering every value of the binary mechanism}\label{regularity:pointwise}
\label{review:app:pointwise}
Mollify $h$ with nonnegative smooth approximate identities. The $a$- and
$b$-derivative matrices of every mollification lie between zero and $\ExtId$,
so each mollification satisfies $\norm{H(a,b)-H(a',b')}\le\norm{a-a'}+\norm{b-b'}$.
Local boundedness and this common estimate give a locally uniform
subsequential limit $H$ satisfying the same estimate. Convergence in
$L^1_{\mathrm{loc}}$ identifies $H=h$ almost everywhere. 

We now identify $h$ with $H$ at every profile through the supporting
inequalities of $u$. Approximate any $(a,b)$ by
points $(a_k,b_k)$ where $h=H$. For every $z\in E$, the supporting inequality
at those points gives $u(z,b_k)\ge u(a_k,b_k)+\ip{H(a_k,b_k)}{z-a_k}$.
Joint continuity of $u,H$ permits passage to the limit. Hence
$H(a,b)\in\partial_a u(a,b)$ everywhere. For a displacement $w$,
$\ip{H(a,b)}w
 \le u(a+w,b)-u(a,b)
 \le\ip{H(a+w,b)}w$.
The difference between these bounds is at most $\norm w^2$. Thus
$u(\cdot,b)$ is differentiable everywhere and
$\partial_a u(a,b)=\{H(a,b)\}$. The original outcome $h(a,b)$ also belongs
to this subdifferential, so $h=H$ at every profile.

Every unilateral menu of $h$ is now convex. It equals the fixed-point set
of its nonexpansive response $T$. For fixed points $p,q$ and
$z=(1-t)p+tq$, $0\le t\le1$, nonexpansiveness gives
\[
 \norm{T(z)-p}\le t\norm{p-q},\qquad
 \norm{T(z)-q}\le(1-t)\norm{p-q}.
\]
Equality in the Euclidean triangle inequality forces $T(z)=z$ (the case
$p=q$ is immediate). Menus were already closed, so the nearest-point
property identifies every response with Euclidean projection. We have
proved \cref{review:thm:binary-regularity}.

\subsection{Descent through coalition menus}\label{regularity:coalitions}
We now pass from two reports to any finite population by grouping agents
at common reports. This follows the block-restriction organization in
Peters, van der Stel, and Storcken~\cite[Lemma~13]{peters1993range}.
The binary theorem transfers convexity from a coalition menu to its
subcoalition menus in every finite dimension, giving the following result.
\begin{theorem}\label{ext:thm:regularity}
Let $f:E^n\to E$ be deterministic and strategyproof, with convex full
range. Every unilateral menu is closed and convex, every unilateral
response is its metric projection, and
$\norm{f(x)-f(x')}\le\sum_i\norm{x_i-x_i'}$.
\end{theorem}
\begin{proof}
For $n=1$, the nearest-point property on the convex closed full range
immediately gives a convex projection. Let $n\ge2$.

We descend on coalition size. Menus of size $n$ are convex by the full-range
hypothesis. Suppose every size-$k$ coalition menu is convex, where $k\ge2$.
Fix $S$ of size $k-1$, choose $j\notin S$, and fix all reports $z$ outside
$S\cup\{j\}$. Define $h(a,b)=f(a^S,b^j,z)$.
It is strategyproof in $b$ by individual strategyproofness, and in $a$ by
chaining the inequalities for the members of $S$ at common true point $a$.
Its full range is exactly $C_{S\cup\{j\}}(z)$: every outcome in that
coalition menu is attained by replacing all those reports by the outcome,
and is therefore attained at a diagonal input of $h$.

The binary full range is convex by induction.
\Cref{review:thm:binary-regularity} makes every unilateral menu of $h$ convex. For fixed $b$, its first-report
menu is precisely $C_S(b^j,z)$ by
\cref{ext:lem:menus-no-continuity}. Hence every size-$(k-1)$ coalition menu is
convex. Descending to size one makes every unilateral response of $f$ a
closed convex projection. Changing reports one at a time and using
nonexpansiveness proves \eqref{ext:eq:regularity-lipschitz}.
\end{proof}

\section{Analytic proof of anonymous extraction}
\label{ext:app:extraction}

This appendix proves \cref{ext:lem:finite-profile-concentration} and
\cref{ext:thm:extraction}. We first define the sampling rule and the
differential influence of a physical report. We then prove that only a
bounded number of physical labels can have large influence and apply a
product-measure concentration inequality to the remaining labels.

\subsection{Sampling law and notation}
\label{ext:app:sampling-laws}

Fix a nonnegative smooth probability density \(b\) supported in the unit
ball of \(E=\mathbb R^d\). For \(\sigma>0\), define
$b_\sigma(z):=\sigma^{-d}b(z/\sigma)$.
Thus a random vector with density \(b_\sigma\) has norm at most \(\sigma\).

Let \(\gamma_d\) be the standard Gaussian probability measure on \(E\).
For a virtual profile \(x=(x_1,\ldots,x_m)\) and parameters
\(0<\sigma,\delta<1\), define the one-report sampling law
\begin{equation}\label{ext:eq:sampling-law}
    \mu_x^{\sigma,\delta}
    :=
    \frac{1-\delta}{m}
    \sum_{j=1}^m b_\sigma(\,\cdot-x_j)\,dz
    +\delta\gamma_d.
\end{equation}
A report sampled from \(\mu_x^{\sigma,\delta}\) is generated as follows.
With probability \(1-\delta\), choose \(j\) uniformly from \([m]\) and
report \(x_j+\xi\), where \(\xi\) has density \(b_\sigma\). With probability \(\delta\), submit an independent standard Gaussian report.
For a probability law $\mu$ on $E$, write $\mu^{\otimes r}$ for the joint
law of $r$ independent samples from $\mu$.

We use the following notation for a source mechanism \(f_N:E^N\to E\).

\begin{itemize}
\item \(K\subseteq[N]\) is the set of physical labels whose reports are
fixed at the origin.
\item \(r:=N-|K|\) is the number of remaining, or \emph{live}, physical
labels.
\item \(X=(X_1,\ldots,X_r)\) is the vector of sampled live reports.
\item $F_K(X):=f_N(0_K,X)$ is the source output after the labels in
\(K\) have been fixed; $0_K$ means that every label in $K$ reports the
origin.
\item For a virtual profile \(x\), \(Y_K(x)\) denotes the random vector
\(F_K(X)\) when the live reports are sampled independently from
\(\mu_x^{\sigma,\delta}\).
\item $h_K(x):=\mathbb E Y_K(x),
    V_K(x):=
    \mathbb E\bigl\|Y_K(x)-h_K(x)\bigr\|^2$
are the mean and variance of the sampled output.
\end{itemize}

Let \(q\) be the density of a probability law \(\mu\) on \(E\). Its Fisher
information is $J(\mu):=
    \int_E \norm{\nabla\log q(z)}^2q(z)\,dz$.
This quantity appears when integration by parts relates a derivative of the
mechanism to the output displacement caused by fixing one report.

The laws in \eqref{ext:eq:sampling-law} have the analytic properties needed
below.

\begin{lemma}\label{ext:lem:sampling-regularity}
Fix \(m,d<\infty\), \(0<\sigma,\delta<1\), and a compact set of virtual
profiles. The corresponding laws \(\mu_x^{\sigma,\delta}\) have uniformly
bounded second moments and Fisher information. Their densities are positive
and smooth. They satisfy a logarithmic Sobolev inequality and the lower
curvature bound required by the product-measure inequality used in
\cref{ext:lem:concentration}, with constants uniform over the compact set.
\end{lemma}

\begin{proof}
Let \(\varphi_d\) be the standard Gaussian density and let \(q_x\) be the
density of \(\mu_x^{\sigma,\delta}\). The translated bump supports remain
in one compact set when \(x\) ranges over a compact set of profiles.
Outside that compact set, \(q_x=\delta\varphi_d\). Inside it, the bump
densities and their derivatives are uniformly bounded, while
\(\varphi_d\) has a positive minimum. Hence there is a finite constant
\(M\) such that $\delta
    \le
    \frac{q_x}{\varphi_d}
    \le
    M$.
The mixture representation gives a uniform second-moment bound. The score
\(\nabla\log q_x\) is uniformly bounded on the compact set containing the
bump supports and equals \(-z\) outside it, so the Fisher information is
uniformly finite.

Write \(q_x=e^{-V_x}\). On the compact set, the Hessians
\(\nabla^2V_x\) are uniformly bounded below; outside it,
\(\nabla^2V_x=I\). Thus \(\nabla^2V_x\succeq-\kappa I\) for a uniform
finite \(\kappa\).

Finally, the Gaussian logarithmic Sobolev inequality transfers through the
density comparison. For every smooth \(u\),
$\operatorname{Ent}_{\mu_x}(u^2)
    \le
    \frac{2M}{\delta}
    \int_E\norm{\nabla u}^2\,d\mu_x$.
This gives the required hypercontractivity and completes the proof.
\end{proof}

\subsection{Differential influence and fixing influential labels}
\label{ext:app:pinning-proof}

Let \(F:E^r\to E\) be continuous and strategyproof. At a differentiability
point, \(D_iF\) denotes the \(d\times d\) Jacobian matrix describing the
first-order change in the output when only physical report \(i\) changes.
By \cref{ext:lem:influence}, the matrices \(D_iF\) are symmetric and
positive semidefinite and satisfy
$\sum_{i=1}^rD_iF\preceq I$.
The nonnegative quantity \(\operatorname{tr}D_iF\) measures the total
infinitesimal sensitivity of the output to report \(i\).

For a one-report sampling law \(\mu\), define the average influence of
physical label \(i\) by
$I_i^\mu(F)
    :=
    \mathbb E_{\mu^{\otimes r}}
    \operatorname{tr}D_iF$.

We now show that a bounded number of labels can be fixed so that all
remaining average influences are small for a finite family of sampling laws.

\begin{lemma}\label{ext:lem:pinning}
Let \(\mu_1,\ldots,\mu_L\) be laws of the form
\eqref{ext:eq:sampling-law}. Let \(J\) be their maximum Fisher information
and \(M_2\) their maximum second moment:
\[
    J:=\max_{1\le\ell\le L}J(\mu_\ell),
    \qquad
    M_2:=\max_{1\le\ell\le L}\mathbb E_{\mu_\ell}\norm X^2.
\]
Let \(f_N\) be continuous and strategyproof with \(\AR(f_N)\le R\).
For every \(\varepsilon>0\), there is a set \(K\subseteq[N]\) with
\begin{equation}\label{ext:eq:pinning-count}
    |K|
    \le
    B:=
    \left\lceil
        \frac{L(R+1)^2M_2J}{\varepsilon^2}
    \right\rceil
\end{equation}
such that, for every live label \(i\notin K\) and every
\(\ell\in[L]\),
\begin{equation}\label{ext:eq:small-influences}
    I_i^{\mu_\ell}(F_K)
    \le
    \varepsilon.
\end{equation}
The same set \(K\) works for all \(L\) laws, and \(B\) is independent of
the source population \(N\).
\end{lemma}

\begin{proof}
For a set \(K\) of fixed labels, let \(r=N-|K|\) and define the potential $\mathcal E(K)
    :=
    \sum_{\ell=1}^L
    \mathbb E_{\mu_\ell^{\otimes r}}
    \norm{F_K(X)}^2$.
This is the total expected squared output over the finite family of
sampling laws.

We first bound the potential. At the profile \((0_K,X)\), compare the
mechanism's outcome \(F_K(X)\) with facility location \(0\). The
approximation guarantee and the triangle inequality give
\[
    N\norm{F_K(X)}
    \le
    (R+1)\sum_{i=1}^r\norm{X_i}.
\]
By Cauchy--Schwarz, $\mathbb E\norm{F_K(X)}^2
    \le
    (R+1)^2M_2$.
Consequently,
\begin{equation}\label{ext:eq:pinning-potential-bound}
    \mathcal E(K)
    \le
    L(R+1)^2M_2.
\end{equation}

Fix a live physical label \(i\). With all other reports fixed, let \(C\)
be the unilateral menu of label \(i\). The projection representation gives $F_K=P_CX_i$, and
    $F_{K\cup\{i\}}=P_C0$.
Using \eqref{eq:projection-normal} with \(x=0\),
\(p=P_C0\), and \(z=P_CX_i\), we obtain
$\norm{F_K}^2\ge
    \norm{F_{K\cup\{i\}}}^2+
    \norm{F_K-F_{K\cup\{i\}}}^2$.
Thus fixing label \(i\) decreases each term of the potential by at least
the expected squared output displacement.

Let \(q_\ell\) be the density of \(\mu_\ell\). Conditional on all
other reports, nonexpansiveness of the projection response gives
\(\norm{F_K-F_{K\cup\{i\}}}\le\norm{X_i}\). Thus the difference has
at most linear growth. Choose a smooth cutoff \(\chi_T\) that equals one
on the ball of radius \(T\), vanishes outside the ball of radius \(2T\),
and satisfies \(\norm{\nabla\chi_T}\le2/T\). Apply integration by parts
to the compactly supported vector field
\(\chi_T(F_K-F_{K\cup\{i\}})\). The term containing
\(\nabla\chi_T\) tends to zero by the second-moment bound and the
Gaussian tail. The score term converges by Cauchy--Schwarz, using the finite
second moment and \(J(\mu_\ell)<\infty\). Letting \(T\to\infty\) gives
\[
    I_i^{\mu_\ell}(F_K)
    =
    -\mathbb E
      \left[
        \bigl(F_K-F_{K\cup\{i\}}\bigr)
        \cdot
        \nabla\log q_\ell(X_i)
      \right].
\]
Cauchy--Schwarz yields
\begin{equation}\label{ext:eq:influence-displacement}
    \bigl(I_i^{\mu_\ell}(F_K)\bigr)^2
    \le
    J\,
    \mathbb E
    \norm{F_K-F_{K\cup\{i\}}}^2.
\end{equation}

Starting with \(K=\varnothing\), whenever some live label satisfies
\(I_i^{\mu_\ell}(F_K)>\varepsilon\) for some \(\ell\), fix that label at
zero. By \eqref{ext:eq:influence-displacement}, this decreases the common
potential by more than \(\varepsilon^2/J\). The bound
\eqref{ext:eq:pinning-potential-bound} allows at most \(B\) such steps.
When the procedure stops, \eqref{ext:eq:small-influences} holds for every
remaining label and every sampling law.
\end{proof}

\subsection{Concentration after pinning}
\label{ext:app:concentration-proof}

We use the following form of the product-measure inequality of
Cordero-Erausquin and Ledoux~\cite[Corollary~3]{ext:CorderoLedoux}. Let
\(\mu\) be one of the laws in \eqref{ext:eq:sampling-law}. There is a
constant \(\Gamma_\mu<\infty\), independent of the number \(r\) of
independent reports, such that every sufficiently regular scalar function
\(g:E^r\to\mathbb R\) satisfies
\begin{equation}\label{ext:eq:talagrand}
    \operatorname{Var}_{\mu^{\otimes r}}(g)
    \le
    \Gamma_\mu
    \sum_{i=1}^r
    \frac{\norm{\nabla_i g}_{L^2}^2}
         {1+\log\!\left(
            \norm{\nabla_i g}_{L^2}/
            \norm{\nabla_i g}_{L^1}
          \right)}.
\end{equation}
Here \(\nabla_i g\) is the gradient with respect to physical report \(i\),
and the \(L^1,L^2\) norms are taken under the product law
\(\mu^{\otimes r}\). A term with zero gradient is interpreted as zero.

The cited inequality is stated for smooth functions. The coordinate
functions of a continuous strategyproof mechanism are Lipschitz and have at
most linear growth. Let \(\chi_T\) be a smooth cutoff on \(E^r\) that is
one on the ball of radius \(T\), vanishes outside the ball of radius
\(2T\), and satisfies \(\norm{\nabla\chi_T}\le2/T\). Gaussian tails
imply that \(g\chi_T\to g\) and
\(\nabla_i(g\chi_T)\to\nabla_i g\) in the weighted \(L^2\) norms;
the additional term \(g\nabla_i\chi_T\) tends to zero because \(g\)
has at most linear growth. Mollify each compactly supported function
\(g\chi_T\). On its compact support, the smooth positive product density
is bounded above and below, so the mollifications converge in the required
\(L^1\) and \(L^2\) derivative norms. Applying
\eqref{ext:eq:talagrand} to the mollifications and passing first through the
mollification limit and then \(T\to\infty\) extends the inequality to the
functions used below. If a limiting \(L^1\) derivative norm is zero, the
derivative itself is zero almost everywhere, so the corresponding
\(L^2\) norm and summand are also zero.

\begin{lemma}\label{ext:lem:concentration}
Let \(\mu\) be a law of the form \eqref{ext:eq:sampling-law}. There is a
finite constant \(\Gamma_\mu\), independent of \(r\), such that the following
holds. Let \(F:E^r\to E\) be continuous and strategyproof, and suppose
\[
    I_i^\mu(F)
    =
    \mathbb E_{\mu^{\otimes r}}
    \operatorname{tr}D_iF
    \le
    \varepsilon
    \qquad(1\le i\le r),
\]
where \(0<\varepsilon<e^{-1}\). Then
\begin{equation}\label{ext:eq:variance-bound}
    \mathbb E
    \norm{F-\mathbb EF}^2
    \le
    d\Gamma_\mu
    \left(
        d\sqrt\varepsilon
        +
        \frac4{\log(1/\varepsilon)}
    \right).
\end{equation}
\end{lemma}

\begin{proof}
Fix an output coordinate \(a\), and let
\(g=\langle e_a,F\rangle\). At a differentiability point, write
\(A_i=D_iF\). Since \(A_i\) is symmetric,
$\nabla_i g=A_ie_a$.
Define $u_i:=\norm{\nabla_i g}_{L^1}$,
and $v_i:=\norm{\nabla_i g}_{L^2}$.
The positive semidefinite matrix inequalities in
\cref{ext:lem:influence} give \(0\preceq A_i\preceq I\) and
\(\sum_iA_i\preceq I\). Hence
\(\norm{A_ie_a}\le\operatorname{tr}A_i\) and
\(A_i^2\preceq A_i\). Therefore
\[
\begin{aligned}
    u_i
    &\le \mathbb E\operatorname{tr}A_i
     = I_i^\mu(F)\le\varepsilon,\\
    \sum_i u_i
    &\le \mathbb E\sum_i\operatorname{tr}A_i\le d,\\
    \sum_i v_i^2
    &=\mathbb E\sum_i e_a^{\mathsf T}A_i^2e_a
      \le \mathbb E\,e_a^{\mathsf T}
             \left(\sum_iA_i\right)e_a\le1.
\end{aligned}
\]
If \(u_i=0\), then \(\nabla_i g=0\) almost everywhere and hence
\(v_i=0\); the corresponding summand in
\eqref{ext:eq:talagrand} is zero.

Apply \eqref{ext:eq:talagrand}. Split the indices into two groups. If
\(v_i\le u_i^{3/4}\), then
\(v_i^2\le u_i^{3/2}\), and the total contribution of this group is at
most $\sum_i u_i^{3/2}
    \le
    \sqrt\varepsilon\sum_i u_i
    \le
    d\sqrt\varepsilon$.
For the remaining indices,
\(v_i/u_i>u_i^{-1/4}\ge\varepsilon^{-1/4}\). Hence
$\log(v_i/u_i)
    \ge
    \frac14\log(1/\varepsilon)$,
and their total contribution is at most
$\frac4{\log(1/\varepsilon)}
    \sum_i v_i^2
    \le
    \frac4{\log(1/\varepsilon)}$.
Therefore
\[
    \operatorname{Var}(g)
    \le
    \Gamma_\mu
    \left(
        d\sqrt\varepsilon
        +
        \frac4{\log(1/\varepsilon)}
    \right).
\]
Summing over the \(d\) output coordinates proves
\eqref{ext:eq:variance-bound}.
\end{proof}

\begin{proof}[Proof of \cref{ext:lem:finite-profile-concentration}]
For each \(x\in\mathcal T\), let
\(\mu_x=\mu_x^{\sigma,\delta}\). By
\cref{ext:lem:sampling-regularity}, the finite family
\(\{\mu_x:x\in\mathcal T\}\) has finite common bounds \(M_2,J\) and finite
concentration constants.

Choose \(0<\varepsilon<e^{-1}\) so that the right-hand side of
\eqref{ext:eq:variance-bound} is at most \(\eta\) for every law in the
finite family. Apply \cref{ext:lem:pinning} to obtain one set \(K\) of at
most \(B\) labels for which every remaining influence is at most
\(\varepsilon\) under every law. Applying
\cref{ext:lem:concentration} to each virtual profile \(x\in\mathcal T\)
gives
\[
    \mathbb E
    \norm{Y_N(x)-\mathbb EY_N(x)}^2
    \le
    \eta.
\]
The pinning bound depends on the finite family and on \(\eta\), but not on
\(N\).
\end{proof}

\subsection{Complete proof of the extraction theorem}

\begin{proof}[Proof of \cref{ext:thm:extraction}]
By \cref{ext:cor:finite-regularity}, every source mechanism \(f_N\) is
continuous and unanimous.

Choose parameters \(\sigma_\ell,\delta_\ell\downarrow0\) and increasing
finite sets \(\mathcal T_\ell\) whose union is the set of rational profiles
in \(E^m\). Apply \cref{ext:lem:finite-profile-concentration} to
\(\mathcal T_\ell\) with variance target \(\ell^{-2}\), obtaining a bound
\(B_\ell\). Choose an available source population \(N_\ell\) so large that $N_\ell\ge(\ell+1)(B_\ell+1)$.
Let \(K_\ell\) be the common set of fixed labels, put
\(k_\ell:=|K_\ell|\), and let
\(r_\ell:=N_\ell-k_\ell\). Then $\frac{k_\ell}{r_\ell}\longrightarrow0$.

For every virtual profile \(x\), sample the \(r_\ell\) live physical reports
independently from
\(\mu_x^{\sigma_\ell,\delta_\ell}\), and define
$Y_\ell(x)
    :=
    f_{N_\ell}(0_{K_\ell},X),
    \qquad
    h_\ell(x):=\mathbb E Y_\ell(x)$.
The same set \(K_\ell\) is used for every \(x\). Since the sampling law
depends only on the multiset of virtual reports, \(h_\ell\) is anonymous.

Couple the samples at \(x\) and \(x'\) by using the same virtual indices,
bounded noises, Gaussian choices, and Gaussian reports. Each live physical
report changes by at most
\(\max_j\norm{x_j-x'_j}\). Applying
\eqref{ext:eq:profile-max-lipschitz} to the restriction with the labels in
\(K_\ell\) fixed gives
\[
    \norm{Y_\ell(x)-Y_\ell(x')}
    \le
    \sqrt d\,\max_j\norm{x_j-x'_j}
\]
almost surely. Taking expectations gives the same bound for \(h_\ell\).

On a bounded set of virtual profiles, the second moments of the sampling
laws are bounded uniformly in \(\ell\). The estimate preceding
\eqref{ext:eq:pinning-potential-bound}, applied after fixing \(K_\ell\),
therefore gives a uniform bound on
\(\mathbb E\norm{Y_\ell(x)}^2\), and Jensen's inequality gives the same
local bound for \(\norm{h_\ell(x)}\). Together with the common Lipschitz
estimate, Arzel\`a--Ascoli and a diagonal subsequence give local uniform
convergence $h_\ell\longrightarrow h$
for a continuous deterministic anonymous map \(h:E^m\to E\).

Let $V_\ell(x)
    :=
    \mathbb E
    \norm{Y_\ell(x)-h_\ell(x)}^2$.
For every rational profile \(q\), the construction gives
\(V_\ell(q)\to0\). For an arbitrary profile \(x\), couple
\(Y_\ell(x)\) and \(Y_\ell(q)\) as above. Since the mean minimizes expected
squared distance to a constant,
\[
    \sqrt{V_\ell(x)}\le
    \left(
        \mathbb E
        \norm{Y_\ell(x)-h_\ell(q)}^2
    \right)^{1/2}\le
    \sqrt{V_\ell(q)}
    +
    \sqrt d\,\max_j\norm{x_j-q_j}.
\]
First let \(\ell\to\infty\), then let the rational profile \(q\) approach
\(x\). This proves \eqref{ext:eq:variance-every-profile}.

We now prove strategyproofness. Let \(x,x'\in E^m\) differ only in one
virtual report, from true point \(a\) to a deviation \(b\). Use the same
sampling choices for both profiles. Remove every bounded bump noise while
keeping the Gaussian and fixed reports unchanged, and denote the resulting
source outputs by \(Z_\ell(x)\) and \(Z_\ell(x')\).

The physical agents assigned to the changed virtual report form a block.
In the truthful profile, every report in this block equals \(a\); in the
deviating profile, every report equals \(b\). Change these physical reports
one at a time. At each step, strategyproofness of the source mechanism,
evaluated at true point \(a\), gives the next inequality in the chain.
Hence
$\norm{Z_\ell(x)-a}
    \le
    \norm{Z_\ell(x')-a}$.

Restoring the bounded noises moves every physical report by at most
\(\sigma_\ell\). By \eqref{ext:eq:profile-max-lipschitz},
\[
    \norm{Y_\ell(x)-a}
    \le
    \norm{Y_\ell(x')-a}
    +
    2\sqrt d\,\sigma_\ell.
\]
Take expectations. Jensen's inequality gives
\(\norm{h_\ell(x)-a}\le\mathbb E\norm{Y_\ell(x)-a}\), while
\[
    \mathbb E\norm{Y_\ell(x')-a}
    \le
    \norm{h_\ell(x')-a}
    +
    \sqrt{V_\ell(x')}.
\]
Therefore
\[
    \norm{h_\ell(x)-a}
    \le
    \norm{h_\ell(x')-a}
    +
    \sqrt{V_\ell(x')}
    +
    2\sqrt d\,\sigma_\ell.
\]
Letting \(\ell\to\infty\) proves
$\norm{h(x)-a}
    \le
    \norm{h(x')-a}$.
Thus \(h\) is strategyproof.

It remains to transfer the approximation guarantee. For a one-report law
\(\mu\), define
$\Phi_\mu(y):=
    \mathbb E_{X\sim\mu}\norm{X-y}$.
Fix a virtual profile \(x\) and a comparison point \(z\in E\). Apply the
source mechanism's approximation inequality to the sampled physical
profile:
\[
    \sum_{i=1}^{r_\ell}
        \norm{X_i-Y_\ell(x)}
    +
    k_\ell\norm{Y_\ell(x)}
    \le
    R\left(
        \sum_{i=1}^{r_\ell}\norm{X_i-z}
        +
        k_\ell\norm z
    \right).
\]
Drop the nonnegative pinned cost on the left, divide by \(r_\ell\), and
use
\[
    \norm{X_i-h_\ell(x)}
    \le
    \norm{X_i-Y_\ell(x)}
    +
    \norm{Y_\ell(x)-h_\ell(x)}.
\]
After taking expectations,
\begin{equation}\label{ext:eq:approximation-transfer}
    \Phi_{\mu_x^{\sigma_\ell,\delta_\ell}}(h_\ell(x))
    \le
    R\Phi_{\mu_x^{\sigma_\ell,\delta_\ell}}(z)
    +
    R\frac{k_\ell}{r_\ell}\norm z
    +
    \sqrt{V_\ell(x)}.
\end{equation}

Couple a report sampled from
\(\mu_x^{\sigma_\ell,\delta_\ell}\) with a uniformly selected virtual
report. On the bounded-noise branch, their distance is at most
\(\sigma_\ell\). On the Gaussian branch, whose probability is
\(\delta_\ell\), the expected distance is at most
$\mathbb E_{G\sim\gamma_d}\norm G+
    \frac1m\sum_{j=1}^m\norm{x_j}$.
Hence the expected coupling distance tends to zero. Since distance to a
fixed point is one-Lipschitz and \(h_\ell(x)\to h(x)\), we may pass to the
limit in \eqref{ext:eq:approximation-transfer}. Using
\(k_\ell/r_\ell\to0\) and
\eqref{ext:eq:variance-every-profile}, we obtain
\[
    \frac1m\sum_{j=1}^m\norm{x_j-h(x)}
    \le
    \frac Rm\sum_{j=1}^m\norm{x_j-z}.
\]
Minimizing over \(z\) gives \(\AR(h)\le R\). At a unanimous virtual
profile, choose \(z\) equal to the common report; the right-hand side is
zero, so \(h\) is unanimous. This completes the proof.
\end{proof}


\section{Examples with phantom reports and unequal weights}
\label{ext:app:scope}
The first example shows how fixed phantom reports can anchor a
finite-ratio mechanism to the origin. The second shows that assigning
unequal weights to four agents can improve the approximation ratio.

\subsection{A mechanism with fixed phantom reports}
Adding two fixed reports at zero gives an anonymous five-agent rule
with bounded influence and an origin-dependent outcome.

\begin{example}
\label{ext:ex:phantom}
On the line with five actual agents, define
\[
 f(x_1,\ldots,x_5)=\med(x_1,\ldots,x_5,0,0).
\]
Social cost counts only the five actual agents. This mechanism is anonymous,
unanimous, continuous, and strategyproof, has bounded unilateral menus, and
has approximation ratio $3/2$, but is not translation-equivariant or a fixed
order-statistic rule of the actual reports.
\end{example}
\begin{proof}
With one report variable, the other six entries are fixed. If their third
and fourth order statistics are $L\le U$, the outcome is projection of the
variable report onto $[L,U]$. Thus unilateral menus are bounded and truthful
reporting is optimal. Anonymity and continuity are immediate, and five
identical actual reports determine the median of the seven entries, proving
unanimity.

For actual order statistics $x_{(1)}\le\cdots\le x_{(5)}$, the outcome belongs
to $[x_{(2)},x_{(4)}]$. The optimum is at $x_{(3)}$, and
\begin{align*}
 \SC(x,x_{(2)})&=\OPT(x)+x_{(3)}-x_{(2)},\\
 \SC(x,x_{(4)})&=\OPT(x)+x_{(4)}-x_{(3)}.
\end{align*}
The first two agents contribute at least $2(x_{(3)}-x_{(2)})$ to the
optimum, and the last two contribute at least $2(x_{(4)}-x_{(3)})$.
Both endpoint costs are therefore at most $(3/2)\OPT(x)$. Convexity of
social cost gives the same bound throughout the interval. The profile
$(1,1,2,2,2)$ attains ratio $3/2$.

However, $f(1,2,3,4,5)=2$ while $f(-2,-1,0,1,2)=0\ne2-3$, so the rule is
not translation-equivariant. Also $f(-2,-1,1,2,3)=0$ is not one of the
actual reports, ruling out a fixed order statistic. 
\end{proof}

\subsection{A weighted median for four agents}
We next count one designated report twice when taking a median.
Each agent's response remains a projection onto a box, and the resulting
rule has a sharp dimension-independent ratio of $\sqrt2$.

\begin{proposition}
\label{ext:prop:four-agents}
Fix an orthonormal coordinate system and define, for four actual agents,
\[
 M(x_1,x_2,x_3,x_4)=\CM(x_1,x_1,x_2,x_3,x_4).
\]
Social cost counts each of the four actual agents once.
The mechanism is continuous, unanimous, deterministic, translation-equivariant,
positively homogeneous, and strategyproof, with bounded unilateral
menus. It is nonanonymous and has $\AR(M)=\sqrt2$ for $d\ge2$.
For $d\ge3$, every coordinate median of four agents with fixed middle-rank
choices has ratio at least $(1+2\sqrt3)/3>\sqrt2$.
\end{proposition}
\begin{proof}
In each coordinate, the designated agent's response is clamping to the
interval between the minimum and maximum of the other three reports.
Equivalently, $M(x)=P_{\operatorname{box}(x_2,x_3,x_4)}x_1$.
For any other agent, fix the four entries consisting of the two designated
copies and the other two reports. If their second and third order statistics
are $L_j\le U_j$, the variable coordinate is clamped to $[L_j,U_j]$.
Thus every actual agent's response is projection onto a fixed compact
orthogonal box, proving strategyproofness. The remaining properties
follow from the odd scalar median. Nonanonymity is witnessed on the line by
$M(0,1,2,3)=1$ and $M(3,1,2,0)=2$.

We first prove the upper bound on total distance.
It suffices to prove, for every comparison point $z$,
$\SC(x,M(x))\le\sqrt2\,\SC(x,z)$.
Translate so that $M(x)=0$. If $z=0$ there is nothing to prove. Otherwise
scale jointly so that $\norm z=1$, and reflect coordinates so that $z_j\ge0$.
The rule commutes with these transformations. Put
\[
 a_i=\sum_{j:x_{i,j}>0}z_j^2,\qquad r_i=\norm{x_i}.
\]
The strictly positive reports in any coordinate have total weight at most
two under weights $(2,1,1,1)$. If the designated agent is positive, all
others are nonpositive; otherwise at most two of the other three are
positive. Therefore
\begin{equation}\label{ext:eq:weighted-constraints}
 0\le a_i\le1,\qquad
 a_1+a_i\le1\ (i=2,3,4),\qquad
 a_2+a_3+a_4\le2(1-a_1).
\end{equation}
These inequalities imply
\begin{equation}\label{ext:eq:weighted-sqrt}
 \sum_{i=1}^4\sqrt{a_i}\le\sum_{i=1}^4\sqrt{1-a_i};
\end{equation}
the algebraic verification is given below. Since $z$ is coordinatewise
nonnegative, dropping nonpositive terms in the inner product and applying
Cauchy--Schwarz gives $\ip{x_i}{z}\le r_i\sqrt{a_i}$. Hence
\begin{align*}
 \norm{x_i-z}^2\ge(r_i-\sqrt{a_i})^2+1-a_i,\qquad
 \sqrt2\norm{x_i-z}\ge r_i-\sqrt{a_i}+\sqrt{1-a_i}.
\end{align*}
The last inequality remains valid if $r_i-\sqrt{a_i}<0$.
Sum over $i$ and use \eqref{ext:eq:weighted-sqrt} to obtain
$\sqrt2\SC(x,z)\ge\sum_i r_i=\SC(x,M(x))$.

To complete the bound, we verify
\eqref{ext:eq:weighted-sqrt} from the coordinate constraints.
Write $s=a_1$ and $r=1-s$. The other numbers $b,c,d$ lie in $[0,r]$ and
have sum at most $2r$. Since $t\mapsto\sqrt t-\sqrt{1-t}$ is increasing,
it suffices to increase these three numbers within their bounds until
$b+c+d=2r$. Set $u=r-b$, $v=r-c$, $w=r-d$. Then $u,v,w\ge0$,
$u+v+w=r$, and $(b,c,d)=(v+w,u+w,u+v)$. The required inequality is
\begin{align}
 \sqrt s+\sqrt{u+v}+\sqrt{u+w}+\sqrt{v+w}
 \le\sqrt{u+v+w}+\sqrt{s+u}+\sqrt{s+v}+\sqrt{s+w}.
 \label{ext:eq:four-root-inequality}
\end{align}
It holds for all nonnegative $s,u,v,w$. To see this, use
\[
 \sqrt x=\frac1{2\sqrt\pi}\int_0^\infty(1-e^{-tx})t^{-3/2}\,dt.
\]
For each $t>0$, put $A=e^{-ts}$, $B=e^{-tu}$, $C=e^{-tv}$, $D=e^{-tw}$,
and $Q=BC+BD+CD-BCD$. The integrand for the right side of
\eqref{ext:eq:four-root-inequality} minus the left side, apart from the
positive common factor, is
\[
 Q+A(1-B-C-D)
 =A(1-B)(1-C)(1-D)+(1-A)Q\ge0,
\]
since $A,B,C,D\in[0,1]$ and $Q=BC(1-D)+BD+CD\ge0$.
This proves \eqref{ext:eq:weighted-sqrt} and the upper bound.

The upper bound is attained on a planar profile.
In two coordinates take
\[
 x_1=(1,0),\quad x_2=(-1,0),\quad x_3=x_4=q=(0,1).
\]
Then $M(x)=0$ and $\SC(x,0)=4$. For every $z$, two applications of the
triangle inequality give
\[
 \norm{x_1-z}+\norm{x_2-z}+2\norm{q-z}
 \ge\norm{x_1-q}+\norm{x_2-q}=2\sqrt2.
\]
Equality holds at $z=q$, proving ratio $\sqrt2$. Zero coordinates embed
this example in every higher dimension.

A three-dimensional profile gives a strictly larger lower bound
for every fixed-middle coordinate median.
In three dimensions consider
\[
 x_1=(1,1,1),\quad x_2=(2,0,0),\quad
 x_3=(0,2,0),\quad x_4=(0,0,2).
\]
Each coordinate has ordered entries $0,0,1,2$, so the lower coordinate
median is zero and its social cost is $6+\sqrt3$. The comparison point
$(1,1,1)$ has cost $3\sqrt3$. Thus the lower median's ratio is at least $\frac{6+\sqrt3}{3\sqrt3}=\frac{1+2\sqrt3}{3}>\sqrt2$.

Coordinate reflections give the same bound for every fixed lower/upper
choice, and zero coordinates embed the instance in all $d\ge3$.
\end{proof}

\section*{Disclosure of AI assistance}
\addcontentsline{toc}{section}{Disclosure of AI assistance}
OpenAI's ChatGPT was used during proof development and manuscript
preparation. In Section~\ref{sec:core-geometry}, AI suggested examining the
option intervals of a ternary core through median-algebra language and
helped identify and summarize the standard median-algebra results used
there. The author developed and verified the Euclidean straightening and
box-boundary arguments in Section~\ref{sec:core-geometry}, as well as the
global frame-alignment argument in Section~\ref{sec:classification}.

The author later developed the anonymous-extraction construction used to
remove anonymity from the large-population result, including the sampling,
pinning, and concentration strategy. AI was used to help formalize,
organize, and check the proof in Section~\ref{ext:sec:extraction} and
Appendix~\ref{ext:app:extraction}. AI also assisted with literature
searches, language revision, and proof writing. The author reviewed all mathematical claims, proofs, citations,
and final wording and takes full responsibility for the manuscript.

\printbibliography[heading=bibintoc]

@article{FilimonovM22,
author = {Filimonov, Alina and Meir, Reshef},
title = {Strategyproof facility location mechanisms on discrete trees},
year = {2022},
issue_date = {Jun 2023},
publisher = {Kluwer Academic Publishers},
address = {USA},
volume = {37},
number = {1},
issn = {1387-2532},
url = {https://doi.org/10.1007/s10458-022-09592-4},
doi = {10.1007/s10458-022-09592-4},
journal = {Autonomous Agents and Multi-Agent Systems},
month = dec,
numpages = {29}
}

@article{Satterthwaite75,
  author = {Satterthwaite, Mark. A.},
  description = {JMM master bibtex},
  journal = {Journal of Economic Theory},
  pages = {187-217},
  title = {Strategy-proofness and Arrow's conditions: Existence and
                  correspondence theorems for voting procedures and social welfare
                  functions},
  volume = 10,
  year = 1975
}

@article{Black48,
 ISSN = {00223808, 1537534X},
 URL = {http://www.jstor.org/stable/1825026},
 author = {Duncan Black},
 journal = {Journal of Political Economy},
 number = {1},
 pages = {23--34},
 publisher = {University of Chicago Press},
 title = {On the Rationale of Group Decision-making},
 urldate = {2024-07-05},
 volume = {56},
 year = {1948}
}

@inproceedings{DokowFMN12,
  author       = {Elad Dokow and
                  Michal Feldman and
                  Reshef Meir and
                  Ilan Nehama},
  title        = {Mechanism design on discrete lines and cycles},
  booktitle    = {Proceedings of the 13th {ACM} Conference on Electronic Commerce, {EC}
                  2012},
  pages        = {423--440},
  publisher    = {{ACM}},
  year         = {2012}
}

@inproceedings{TangYZ20,
  author       = {Pingzhong Tang and
                  Dingli Yu and
                  Shengyu Zhao},
  editor       = {P{\'{e}}ter Bir{\'{o}} and
                  Jason D. Hartline and
                  Michael Ostrovsky and
                  Ariel D. Procaccia},
  title        = {Characterization of Group-strategyproof Mechanisms for Facility Location
                  in Strictly Convex Space},
  booktitle    = {{EC} '20: The 21st {ACM} Conference on Economics and Computation,
                  Virtual Event, Hungary, July 13-17, 2020},
  pages        = {133--157},
  publisher    = {{ACM}},
  year         = {2020}
}

@inproceedings{FotakisT10,
  author       = {Dimitris Fotakis and
                  Christos Tzamos},
  title        = {Winner-Imposing Strategyproof Mechanisms for Multiple Facility Location
                  Games},
  booktitle    = {Internet and Network Economics - 6th International Workshop, {WINE}
                  2010},
  series       = {Lecture Notes in Computer Science},
  volume       = {6484},
  pages        = {234--245},
  publisher    = {Springer},
  year         = {2010}
}

@article{AlonFPT10,
  author       = {Noga Alon and
                  Michal Feldman and
                  Ariel D. Procaccia and
                  Moshe Tennenholtz},
  title        = {Strategyproof Approximation of the Minimax on Networks},
  journal      = {Mathematics of Operations Research},
  volume       = {35},
  number       = {3},
  pages        = {513--526},
  year         = {2010}
}

@inproceedings{AzizLSW22,
  author       = {Haris Aziz and
                  Alexander Lam and
                  Mashbat Suzuki and
                  Toby Walsh},
  editor       = {Sanmi Koyejo and
                  S. Mohamed and
                  A. Agarwal and
                  Danielle Belgrave and
                  K. Cho and
                  A. Oh},
  title        = {Random Rank: The One and Only Strategyproof and Proportionally Fair
                  Randomized Facility Location Mechanism},
  booktitle    = {Advances in Neural Information Processing Systems 35: Annual Conference
                  on Neural Information Processing Systems 2022, NeurIPS 2022, New Orleans,
                  LA, USA, November 28 - December 9, 2022},
  year         = {2022}
}

@article{SerafinoV16,
  title={Heterogeneous facility location without money},
  author={Serafino, Paolo and Ventre, Carmine},
  journal={Theoretical Computer Science},
  volume={636},
  pages={27--46},
  year={2016},
  publisher={Elsevier}
}

@article{FotakisT14,
  author       = {Dimitris Fotakis and
                  Christos Tzamos},
  title        = {On the Power of Deterministic Mechanisms for Facility Location Games},
  journal      = {ACM Transactions on Economics and Computation (TEAC)},
  volume       = {2},
  number       = {4},
  pages        = {15:1--15:37},
  year         = {2014}
}

@article{FotakisT16,
  author       = {Dimitris Fotakis and
                  Christos Tzamos},
  title        = {Strategyproof Facility Location for Concave Cost Functions},
  journal      = {Algorithmica},
  volume       = {76},
  number       = {1},
  pages        = {143--167},
  year         = {2016}
}

@inproceedings{EscoffierGTPS11,
  author       = {Bruno Escoffier and
                  Laurent Gourv{\`{e}}s and
                  Kim Thang Nguyen and
                  Fanny Pascual and
                  Olivier Spanjaard},
  title        = {Strategy-Proof Mechanisms for Facility Location Games with Many Facilities},
  booktitle    = {Algorithmic Decision Theory - Second International Conference, {ADT}
                  2011},
  series       = {Lecture Notes in Computer Science},
  volume       = {6992},
  pages        = {67--81},
  publisher    = {Springer},
  year         = {2011}
}

@inproceedings{LuSWZ10,
  author       = {Pinyan Lu and
                  Xiaorui Sun and
                  Yajun Wang and
                  Zeyuan Allen Zhu},
  title        = {Asymptotically optimal strategy-proof mechanisms for two-facility
                  games},
  booktitle    = {Proceedings 11th {ACM} Conference on Electronic Commerce (EC-2010)},
  pages        = {315--324},
  publisher    = {{ACM}},
  year         = {2010}
}

@inproceedings{FeldmanW13,
  author       = {Michal Feldman and
                  Yoav Wilf},
  title        = {Strategyproof facility location and the least  squares objective},
  booktitle    = {Proceedings of the fourteenth {ACM} Conference on Electronic Commerce,
                  {EC} 2013},
  pages        = {873--890},
  publisher    = {{ACM}},
  year         = {2013}
}

@inproceedings{procaccia2013approximate,
author = {Procaccia, Ariel D. and Tennenholtz, Moshe},
title = {Approximate mechanism design without money},
year = {2009},
isbn = {9781605584584},
publisher = {ACM},
url = {https://doi.org/10.1145/1566374.1566401},
doi = {10.1145/1566374.1566401},
booktitle = {Proceedings of the 10th ACM Conference on Electronic Commerce},
pages = {177–186},
numpages = {10},
series = {EC '09}
}

@inproceedings{AgrawalBGTX22, author = {Agrawal, Priyank and Balkanski, Eric and Gkatzelis, Vasilis and Ou, Tingting and Tan, Xizhi}, title = {Learning-Augmented Mechanism Design: Leveraging Predictions for Facility Location}, year = {2022}, isbn = {9781450391504}, publisher = {Association for Computing Machinery}, url = {https://doi.org/10.1145/3490486.3538306}, doi = {10.1145/3490486.3538306}, booktitle = {Proceedings of the 23rd ACM Conference on Economics and Computation}, pages = {497–528}, numpages = {32}, series = {EC '22} }

@inproceedings{meir2019strategyproof,
  author       = {Reshef Meir},
  title        = {Strategyproof Facility Location for Three Agents on a Circle},
  booktitle    = {Algorithmic Game Theory - 12th International Symposium, {SAGT} 2019},
  series       = {Lecture Notes in Computer Science},
  volume       = {11801},
  pages        = {18--33},
  publisher    = {Springer},
  year         = {2019},
  url          = {https://doi.org/10.1007/978-3-030-30473-7_2},
  doi          = {10.1007/978-3-030-30473-7_2},
  bibsource    = {dblp computer science bibliography, https://dblp.org}
}

@article{schummer2002strategy,
  title={Strategy-proof location on a network},
  author={Schummer, James and Vohra, Rakesh V},
  journal={Journal of Economic Theory},
  volume={104},
  number={2},
  pages={405--428},
  year={2002},
  publisher={Elsevier}
}

@article{GoelHC23,
  title={Optimality of the coordinate-wise median mechanism for strategyproof facility location in two dimensions},
  author={Goel, Sumit and Hann-Caruthers, Wade},
  journal={Social Choice and Welfare},
  volume={61},
  number={1},
  pages={11--34},
  year={2023},
  publisher={Springer}
}

@inproceedings{chan2021mechanism,
  author       = {Hau Chan and
                  Aris Filos{-}Ratsikas and
                  Bo Li and
                  Minming Li and
                  Chenhao Wang},
  title        = {Mechanism Design for Facility Location Problems: {A} Survey},
  booktitle    = {Proceedings of the Thirtieth International Joint Conference on Artificial
                  Intelligence, {IJCAI} 2021},
  pages        = {4356--4365},
  publisher    = {ijcai.org},
  year         = {2021}
}

@article{walsh2020strategy,
  title={Strategy proof mechanisms for facility location in {Euclidean} and {Manhattan} space},
  author={Walsh, Toby},
  journal={arXiv preprint arXiv:2009.07983},
  year={2020}
}

@article{moulin1980strategy,
  title={On strategy-proofness and single peakedness},
  author={Moulin, Herv{\'e}},
  journal={Public Choice},
  volume={35},
  number={4},
  pages={437--455},
  year={1980},
  publisher={Springer}
}

@article{peters1993range,
  title={Range convexity, continuity, and strategy-proofness of voting schemes},
  author={Peters, Hans and van der Stel, Hans and Storcken, Ton},
  journal={Zeitschrift f{\"u}r Operations Research},
  volume={38},
  pages={213--229},
  year={1993},
  publisher={Springer}
}

@article{peters1992pareto,
	author = {Peters, H. and van der Stel, H. and Storcken, T.},
	date = {1992-09-01},
	doi = {10.1007/BF01258276},
	id = {Peters1992},
	journal = {International Journal of Game Theory},
	number = {3},
	pages = {221--235},
	title = {Pareto optimality, anonymity, and strategy-proofness in location problems},
	url = {https://doi.org/10.1007/BF01258276},
	volume = {21},
	year = {1992}}

@article{istrate2022mechanism,
  title={Mechanism Design With Predictions for Obnoxious Facility Location},
  author={Istrate, Gabriel and Bonchis, Cosmin},
  journal={arXiv preprint arXiv:2212.09521},
  year={2022}
}

@article{barak2024mac,
  author       = {Zohar Barak and
                  Anupam Gupta and
                  Inbal Talgam{-}Cohen},
  title        = {{MAC} Advice for Facility Location Mechanism Design},
  journal      = {CoRR},
  volume       = {abs/2403.12181},
  year         = {2024},
  url          = {https://doi.org/10.48550/arXiv.2403.12181},
  doi          = {10.48550/ARXIV.2403.12181},
  eprinttype    = {arXiv},
  eprint       = {2403.12181},
  bibsource    = {dblp computer science bibliography, https://dblp.org}
}

@article{ChristodoulouSV24,
  title={Mechanism design augmented with output advice},
  author={Christodoulou, George and Sgouritsa, Alkmini and Vlachos, Ioannis},
  journal={arXiv preprint arXiv:2406.14165},
  year={2024}
}

@article{FeigenbaumSY17,
  author       = {Itai Feigenbaum and
                  Jay Sethuraman and
                  Chun Ye},
  title        = {Approximately Optimal Mechanisms for Strategyproof Facility Location:
                  Minimizing \emph{L\({}_{\mbox{p}}\)} Norm of Costs},
  journal      = {Mathematics of Operations Research},
  volume       = {42},
  number       = {2},
  pages        = {434--447},
  year         = {2017},
  url          = {https://doi.org/10.1287/moor.2016.0810},
  doi          = {10.1287/MOOR.2016.0810},
  bibsource    = {dblp computer science bibliography, https://dblp.org}
}

@article{border1983straightforward,
  title={Straightforward elections, unanimity and phantom voters},
  author={Border, Kim C and Jordan, James S},
  journal={The Review of Economic Studies},
  volume={50},
  number={1},
  pages={153--170},
  year={1983},
  publisher={Wiley-Blackwell}
}

@article{kim1984nonmanipulability,
  title={Nonmanipulability in two dimensions},
  author={Kim, Ki Hang and Roush, Fred W},
  journal={Mathematical Social Sciences},
  volume={8},
  number={1},
  pages={29--43},
  year={1984},
  publisher={Elsevier}
}

@article{barbera1993generalized,
  title={Generalized median voter schemes and committees},
  author={Barber{\`a}, Salvador and Gul, Faruk and Stacchetti, Ennio},
  journal={Journal of Economic Theory},
  volume={61},
  number={2},
  pages={262--289},
  year={1993},
  publisher={Elsevier}
}

@article{ching1997strategy,
  title={Strategy-proofness and “median voters”},
  author={Ching, Stephen},
  journal={International Journal of Game Theory},
  volume={26},
  pages={473--490},
  year={1997},
  publisher={Springer}
}

@article{BalkanskiGS24,
  title={Randomized Strategic Facility Location with Predictions},
  author={Balkanski, Eric and Gkatzelis, Vasilis and Shahkarami, Golnoosh},
  journal={arXiv preprint arXiv:2409.07142},
  year={2024}
}

@inproceedings{zou15, author = {Zou, Shaokun and Li, Minming}, title = {Facility Location Games with Dual Preference}, year = {2015}, isbn = {9781450334136}, publisher = {International Foundation for Autonomous Agents and Multiagent Systems}, booktitle = {Proceedings of the 2015 International Conference on Autonomous Agents and Multiagent Systems}, pages = {615–623}, numpages = {9}, series = {AAMAS '15} }

@article{cheng2013strategy,
  title={Strategy-proof approximation mechanisms for an obnoxious facility game on networks},
  author={Cheng, Yukun and Yu, Wei and Zhang, Guochuan},
  journal={Theoretical Computer Science},
  volume={497},
  pages={154--163},
  year={2013},
  publisher={Elsevier}
}

@article{barbera2011strategyproof,
  title={Strategyproof social choice},
  author={Barber{\`a}, Salvador},
  journal={Handbook of social choice and welfare},
  volume={2},
  pages={731--831},
  year={2011},
  publisher={Elsevier}
}

@article{barbera1998strategy,
  title={Strategy-proof voting on compact ranges},
  author={Barber{\`a}, Salvador and Mass{\'o}, Jordi and Serizawa, Shigehiro},
  journal={Games and Economic Behavior},
  volume={25},
  number={2},
  pages={272--291},
  year={1998},
  publisher={Elsevier}
}

@article{dekeijzer2022, author = {de Keijzer, Bart and Wojtczak, Dominik}, title = {Facility Reallocation on the Line}, year = {2022}, issue_date = {Oct 2022}, publisher = {Springer-Verlag}, address = {Berlin, Heidelberg}, volume = {84}, number = {10}, issn = {0178-4617}, url = {https://doi.org/10.1007/s00453-022-00993-1}, doi = {10.1007/s00453-022-00993-1}, journal = {Algorithmica}, month = oct, pages = {2898–2925}, numpages = {28} }

@inproceedings{fotakis2019, author = {Fotakis, Dimitris and Kavouras, Loukas and Kostopanagiotis, Panagiotis and Lazos, Philip and Skoulakis, Stratis and Zarifis, Nikos}, title = {Reallocating multiple facilities on the line}, year = {2019}, isbn = {9780999241141}, publisher = {AAAI Press}, booktitle = {Proceedings of the 28th International Joint Conference on Artificial Intelligence}, pages = {273–279}, numpages = {7}, location = {Macao, China}, series = {IJCAI'19} }

@article{peremans1997strategy,
  title={Strategy-proofness on {Euclidean} spaces},
  author={Peremans, W and Peters, Hans and van der Stel, H and Storcken, T},
  journal={Social Choice and Welfare},
  volume={14},
  pages={379--401},
  year={1997},
  publisher={Springer}
}

@inproceedings{fang2024mechanism,
  title={Mechanism Design with Predictions for Facility Location Games with Candidate Locations},
  author={Fang, Jiazhu and Fang, Qizhi and Liu, Wenjing and Nong, Qingqin},
  booktitle={Annual Conference on Theory and Applications of Models of Computation},
  pages={38--49},
  year={2024},
  organization={Springer}
}

@misc{barak26,
      title={Facility Location Mechanism Design: Breaking The Deterministic Barrier}, 
      author={Zohar Barak},
      year={2026},
      eprint={2605.24750},
      archivePrefix={arXiv},
      primaryClass={cs.GT},
      url={https://arxiv.org/abs/2605.24750}, 
}

@misc{ChanLW26,
      title={Strategyproof Mechanisms for Euclidean Facility Location Problems under $L_p$-norm Social Cost}, 
      author={Hau Chan and Jianan Lin and Chenhao Wang},
      year={2026},
      eprint={2606.08621},
      archivePrefix={arXiv},
      primaryClass={cs.GT},
      url={https://arxiv.org/abs/2606.08621}, 
}

@misc{hastings2026,
      title={Strategic Facility Location with $p$-Norm Social Costs}, 
      author={Jabari Hastings},
      year={2026},
      eprint={2606.12187},
      archivePrefix={arXiv},
      primaryClass={cs.GT},
      url={https://arxiv.org/abs/2606.12187}, 
}

@inproceedings{jia1, author = {Gravin, Nikolai and Jia, Jianhao}, title = {Approximation Guarantees of Median Mechanism in $\mathbb{R}^d$}, year = {2025}, isbn = {9798400715105}, publisher = {Association for Computing Machinery}, address = {New York, NY, USA}, url = {https://doi.org/10.1145/3717823.3718156}, doi = {10.1145/3717823.3718156}, booktitle = {Proceedings of the 57th Annual ACM Symposium on Theory of Computing}, pages = {495–506}, numpages = {12}, location = {Prague, Czechia}, series = {STOC '25} }

@article{BordesLaffondLeBreton2011,
  title={Euclidean preferences, option sets and strategyproofness},
  author={Georges Bordes and Gilbert Laffond and Michel {Le Breton}},
  journal={SERIEs},
  year={2011},
  volume={2},
  pages={469-483},
  url={https://api.semanticscholar.org/CorpusID:154613111}
}

@article{BauschkeWangYeYuan2009,
title = {Bregman distances and Chebyshev sets},
journal = {Journal of Approximation Theory},
volume = {159},
number = {1},
pages = {3-25},
year = {2009},
note = {Special Issue in memory of Professor George G. Lorentz (1910–2006)},
issn = {0021-9045},
doi = {https://doi.org/10.1016/j.jat.2008.08.014},
url = {https://www.sciencedirect.com/science/article/pii/S0021904508001652},
author = {Heinz H. Bauschke and Xianfu Wang and Jane Ye and Xiaoming Yuan}
}

@article{Bowditch2016,
author = {Bowditch, Brian},
year = {2016},
month = {02},
pages = {279-317},
title = {Some properties of median metric spaces},
volume = {10},
journal = {Groups, Geometry, and Dynamics},
doi = {10.4171/GGD/350}
}

@article{CHATTERJI2010882,
title = {Kazhdan and Haagerup properties from the median viewpoint},
journal = {Advances in Mathematics},
volume = {225},
number = {2},
pages = {882-921},
year = {2010},
issn = {0001-8708},
doi = {https://doi.org/10.1016/j.aim.2010.03.012},
url = {https://www.sciencedirect.com/science/article/pii/S0001870810001003},
author = {Indira Chatterji and Cornelia Druţu and Frédéric Haglund}
}

@article{Fioravanti2021ConvexCF,
  title={Convex cores for actions on finite-rank median algebras},
  author={Elia Fioravanti},
  journal={Annales de l'Institut Fourier},
  year={2021},
  url={https://api.semanticscholar.org/CorpusID:231419047}
}

@article{NehringPuppe2007,
title = {The structure of strategy-proof social choice — Part I: General characterization and possibility results on median spaces},
journal = {Journal of Economic Theory},
volume = {135},
number = {1},
pages = {269-305},
year = {2007},
issn = {0022-0531},
doi = {https://doi.org/10.1016/j.jet.2006.04.008},
url = {https://www.sciencedirect.com/science/article/pii/S0022053106001050},
author = {Klaus Nehring and Clemens Puppe}
}

@InProceedings{Lin2020,
author="Lin, Jianan",
editor="Wu, Weili
and Zhang, Zhongnan",
title="Nearly Complete Characterization of 2-Agent Deterministic Strategyproof Mechanisms for Single Facility Location in $L_p$ Space",
booktitle="Combinatorial Optimization and Applications",
year="2020",
publisher="Springer International Publishing",
address="Cham",
pages="411--425",
isbn="978-3-030-64843-5"
}

@book{Rockafellar1970,
  title={Convex analysis},
  author={Rockafellar, R Tyrrell},
  volume={28},
  year={1997},
  publisher={Princeton University Press}
}

@book{EvansGariepy2015,
  title={Measure Theory and Fine Properties of Functions, Revised Edition},
  author={Evans, L.C. and Gariepy, R.F.},
  isbn={9781482242393},
  series={Textbooks in Mathematics},
  url={https://books.google.com/books?id=e3R3CAAAQBAJ},
  year={2015},
  publisher={CRC Press}
}

@article{ext:CorderoLedoux,
	author = {Dario Cordero{-}Erausquin and Michel Ledoux},
	journal = {Lecture Notes in Mathematics -Springer-Verlag-},
	title = {Hypercontractive Measures, Talagrand's Inequality, and Influences},
	volume = {2050},
	year = {2011}
}

@online{ext:Dyatlov,
  author    = {Dyatlov, Semyon},
  title     = {Lecture Notes for 18.155: Distributions, Elliptic Regularity, and Applications to {PDEs}},
  url       = {https://math.mit.edu/~dyatlov/18.155/155-notes.pdf},
}

@article{lebreton1999separable,
author = {Breton, Michel Le and Sen, Arunava},
title = {Separable Preferences, Strategyproofness, and Decomposability},
journal = {Econometrica},
volume = {67},
number = {3},
pages = {605-628},
doi = {https://doi.org/10.1111/1468-0262.00038},
url = {https://onlinelibrary.wiley.com/doi/abs/10.1111/1468-0262.00038},
eprint = {https://onlinelibrary.wiley.com/doi/pdf/10.1111/1468-0262.00038},
year = {1999}
}

@misc{ma2026,
      title={Breaking the 4-Approximation Barrier in Strategyproof Two-Facility Location}, 
      author={Mengfan Ma and Bo Peng},
      year={2026},
      eprint={2608.09061},
      archivePrefix={arXiv},
      primaryClass={cs.GT},
      url={https://arxiv.org/abs/2608.09061}, 
}

@misc{jia2026,
      title={Product Gap Mechanisms for Multi-Facility Location}, 
      author={Jianhao Jia},
      year={2026},
      eprint={2608.19633},
      archivePrefix={arXiv},
      primaryClass={cs.GT},
      url={https://arxiv.org/abs/2608.19633}, 
}

@misc{chan2026,
      title={Randomized Strategyproof Facility Location: Two Facilities and Beyond}, 
      author={Hau Chan and Jianan Lin and Chenhao Wang},
      year={2026},
      eprint={2608.22484},
      archivePrefix={arXiv},
      primaryClass={cs.GT},
      url={https://arxiv.org/abs/2608.22484}, 
}

@misc{nguyen2026,
      title={Strategic Facility Location in Euclidean Spaces}, 
      author={Kim Thang Nguyen and Lucas Perotin and Bertrand Simon},
      year={2026},
      eprint={2609.05132},
      archivePrefix={arXiv},
      primaryClass={cs.GT},
      url={https://arxiv.org/abs/2609.05132}, 
}

@misc{barak2026exante,
      title={Ex-ante versus Ex-post: Egalitarian Facility Location Mechanism Design}, 
      author={Zohar Barak and Inbal Talgam-Cohen},
      year={2026},
      eprint={2608.14499},
      archivePrefix={arXiv},
      primaryClass={cs.GT},
      url={https://arxiv.org/abs/2608.14499}, 
}

@misc{lin2026opt,
      title={Optimality and Impossibility in Multidimensional Strategyproof Facility Location}, 
      author={Jianan Lin},
      year={2026},
      eprint={2610.36613},
      archivePrefix={arXiv},
      primaryClass={cs.GT},
      url={https://arxiv.org/abs/2610.36613}, 
}

@InProceedings{chen2024,
author="Chen, Qingyun
and Gravin, Nick
and Im, Sungjin",
editor="Mavronicolas, Marios
and Qi, Qi
and Schoenebeck, Grant",
title="Strategic Facility Location via Predictions",
booktitle="Web and Internet Economics",
year="2026",
publisher="Springer Nature Switzerland",
address="Cham",
pages="401--419",
isbn="978-3-032-08560-3"
}
\end{document}